\documentclass[a4paper,10pt]{article}
\usepackage[utf8]{inputenc}

\newcommand{\SV}[1]{#1}
\newcommand{\LV}[1]{}

\newcommand{\shortbib}[1]{}

\usepackage{hyperref}
\usepackage[utf8]{inputenc}
\usepackage{float}
\usepackage{xspace} 
\usepackage{comment} 
\usepackage{diagbox}
\usepackage{tikz}
\usepackage[color=green!40]{todonotes}
\usepackage{amsmath,amsfonts,amssymb} \usepackage{amsthm}
\usepackage{mathtools}
\usepackage{amsopn}
\usepackage{tabularx,lipsum,environ}
\usepackage{paralist}
\usepackage{microtype}
\usepackage{authblk}
\usepackage{hyperref}
\usepackage{fancyhdr}
\usepackage{rotating}
\usepackage{overpic}
\usepackage{ucs}
\usepackage{enumerate}
\usepackage{graphicx}
\usepackage{capt-of}
\usepackage{booktabs}
\usepackage{varwidth}
 
\usepackage{amssymb,amsfonts,textcomp,setspace,multirow,tabu}
\usepackage{stmaryrd}
\usepackage[noend]{algpseudocode}
\usepackage{algpascal}
\usepackage{paralist}
\usepackage{comment}
\usepackage[ruled,vlined,linesnumbered]{algorithm2e}
\usepackage[noend]{algpseudocode}

\tikzstyle{vertex}=[circle, draw, inner sep=2pt, minimum width=4pt, minimum size=0.6cm]
\usetikzlibrary{decorations,decorations.pathmorphing,decorations.pathreplacing,fit,positioning}

\renewcommand\leq\leqslant
\renewcommand\geq\geqslant

\newcommand{\yes}{\textsl{yes}}

\newcommand{\suchthat}{\colon}
\newcommand{\np}{\mathsf{NP}}
\newcommand{\npo}{\mathsf{NPO}}
\newcommand{\paranp}{\mathsf{para}\textsc{-}\mathsf{NP}}
\newcommand{\fpt}{\mathsf{FPT}}
\newcommand{\xp}{\mathsf{XP}}
\newcommand{\wone}{\mathsf{W[1]}}
\newcommand{\wtwo}{\mathsf{W[2]}}
\newcommand{\wthree}{\mathsf{W[3]}}

\newcommand{\warbitrary}{\mathsf{W}}

\newcommand{\Oh}{\mathcal{O}}

\newcommand{\MCI}{\textsc{Multicolored Independent Set}\xspace}

\newcommand{\MCD}{\textsc{Multicolored Dominating Set}\xspace}

\tikzstyle{vertex}=[circle, draw, inner sep=1.2pt, minimum width=4pt, minimum size=0.4cm]

\tikzstyle{vertex2}=[circle, draw, inner sep=0pt, minimum width=4pt, minimum size=0.15cm]

\usetikzlibrary{decorations,decorations.pathmorphing,decorations.pathreplacing,fit,positioning}

\title{On the Complexity of Solution Extension of Optimization Problems
}

\author[2]{Katrin Casel}
\author[3]{Henning Fernau}
\author[1]{Mehdi Khosravian Ghadikolaei}
\author[1]{J\'er\^ome Monnot}
\author[1]{Florian Sikora}
\affil[1]{Université Paris-Dauphine, PSL University, CNRS, LAMSADE, 75016 Paris, France\\\texttt{\{mehdi.khosravian-ghadikolaei,jerome.monnot,florian.sikora\}@dauphine.fr}}

\affil[2]{Hasso Plattner Institute, University of Potsdam, Potsdam, Germany\\\texttt{casel@informatik.uni-trier.de}}

\affil[3]{Universit\"at Trier, Fachbereich 4, Informatikwissenschaften, CIRT,\\ Theoretische Informatik,  D-54286 Trier, Germany\\\texttt{fernau@informatik.uni-trier.de}}

\theoremstyle{plain}
\newtheorem{theorem}{Theorem}
\newtheorem*{theorem*}{Theorem}
\newtheorem{lemma}[theorem]{Lemma}
\newtheorem{corollary}[theorem]{Corollary}

\newtheorem{proposition}[theorem]{Proposition}

\newtheorem{rmk}[theorem]{\bf Remark}
\newtheorem{obs}[theorem]{\bf Observation}

\newlength{\btw}
\newlength{\stw}
\newsavebox\tmpbox 
\date{}
\begin{document}

\maketitle

\begin{abstract}
The question if a given partial solution to a problem can be extended reasonably occurs in many algorithmic approaches for optimization problems. For instance, when enumerating minimal dominating sets of a graph $G=(V,E)$, one usually arrives at the problem to decide for a vertex set $U \subseteq V$, if there exists a \textit{minimal} dominating set $S$ with $U\subseteq S$. We propose a general, partial-order based formulation of such extension problems and study a number of specific problems which can be expressed in this framework. Possibly contradicting intuition, these problems tend to be ${\sf NP}$-hard, even for problems where the underlying optimisation problem  can be solved in polynomial time. This raises the question of how fixing a partial solution causes this increase in difficulty. In this regard, we study the parameterised complexity of extension problems with respect to parameters related to the partial solution, as well as the optimality of simple exact algorithms under the Exponential-Time Hypothesis. 
All complexity considerations are also carried out in very restricted scenarios, be it degree restrictions or topological restrictions (planarity) for graph problems or the size of the given partition for the considered extension variant of \textsc{Bin Packing}.
\end{abstract}

\section{Introduction and Motivation}
The very general problem of determining the quality of a given partial solution occurs basically in every algorithmic approach which computes solutions in some sense gradually. Pruning search-trees, proving approximation guarantees or the efficiency of enumeration strategies usually requires a suitable way to decide if a partial solution is a reasonable candidate to pursue. Consider for example the classical concept of minimal dominating sets for graphs. The task of finding a maximum cardinality minimal dominating set (or an approximation of it) as well as enumerating all minimal dominating sets naturally leads to solving the following extension problem: Given a graph $G=(V,E)$ and a vertex set $U \subseteq V$, does there exists a \textit{minimal} dominating set $S$ with $U\subseteq S$.\par
In this paper, we want to consider these kinds of subproblems which we call \textit{extension problems}. Informally, in an extension version of a problem, we will be given in the input a partial solution to be extended into a \textit{minimal} or a \textit{maximal} solution for the problem (but not necessarily to a solution of globally minimum or maximum value). 
Extension problems as studied in this paper are encountered in many situations when dealing with classical decision problems, say, with \textsc{Vertex Cover}.
\begin{itemize}
\item When running a search tree algorithm, usually parts of the constructed solution are fixed. It is highly desirable to be able to prune branches of the search tree as early as possible. Hence, it would be very nice to tell (in polynomial time) if such a solution part (that we will call a \emph{pre-solution} below) can be part of a minimal vertex cover.
\item The consideration of the previous item is valid both if we look for a minimum vertex cover (or, as a decision variant, for a vertex cover of a size upper-bounded by a given number $k$), or if we look for a minimal vertex cover of maximum size (or, for a minimal vertex cover lower-bounded by a given number $k$).
\item The same type of reasoning is true if one likes to enumerate or count all minimal vertex covers or similar structures~\LV{\cite{BorEGKM2001,doi:10.1080/10556789808805708,BorGKM2000,Dam2006,Fer02a,GolHKKV2016,GolHKKV2017,GolHKV2015,KhaBEGM2005,KhaBEG2008,MooMos65,SchSpe2002}}\SV{\cite{BorEGKM2001,doi:10.1080/10556789808805708,BorGKM2000,Dam2006,GolHKKV2017,KhaBEGM2005,MooMos65,SchSpe2002}}. 
It was this scenario where the question of complexity was 
asked for \textsc{Vertex Cover Extension} in~\cite{Dam2006}.
\item More generally, the question of finding extensions to minimal solutions was encountered in the context of proving hardness results for (efficient) enumeration algorithms for Boolean formulae, in the context of matroids and similar situations; see \cite{doi:10.1080/10556789808805708,KhaBEGM2005}.
\item The following question was already asked in 1956 by Kurt Gödel in a famous letter to Joh(an)n von Neumann~\cite{DBLP:conf/stoc/Sipser92}: \emph{It
would be interesting to know, \dots how strongly in general the number of steps in finite
combinatorial problems can be reduced with respect to simple exhaustive search.} The mentioned pruning of search branches and hence the question of finding (pre-)solution extensions lies at the heart of this question.
\item Extensions also play some role in approaches for polynomial-time approximation schemes.
\end{itemize}

This paper is not the first one to consider extension problems, yet it proposes a general framework to host this type of problems.
In~\cite{doi:10.1080/10556789808805708}, it is shown that this kind of extension of partial solutions is $\np$-hard for the problem of computing prime implicants of the dual of a Boolean function; a problem which can also be seen as the problem of finding a minimal hitting set for the set of prime implicants of the input function. Interpreted in this way, the proof from~\cite{doi:10.1080/10556789808805708} yields  $\np$-hardness for the minimal extension problem for {\sc 3-Hitting Set}. This result was extended in \cite{BazBCFJKLLMP2018} to prove  $\np$-hardness for the extension of minimal dominating sets (denoted by \textsc{Ext DS} here), even when restricted to planar cubic graphs.
Similarly, it was shown in \cite{BazganBCF16}
that extension for minimum vertex cover ({\sc Ext VC}) is $\np$-hard, even restricted to planar cubic graphs.
However, we are not aware of a \emph{systematic} study of this type of problems.\par
In an attempt to study the nature of such extension tasks, we propose a general framework to express a broad class of what we refer to as extension problems. This framework is based on a partial order approach, reminiscent of what has been endeavored for \emph{maximin} problems in~\cite{Man98}. In essence, we consider optimisation problems in {\sf NPO} with an additionally specified set of partial solutions which we call \emph{pre-solutions} (including the set of solutions) and a partial order on those. This partial order  $\leq$ reflects not only the notion of \emph{extension} but also of \emph{minimality} as follows. For a pre-solution~$U$ and a solution $S$, $S$ \emph{extends}~$U$ if $U\leq S$. A solution $S$ is \emph{minimal}, if there exists no solution $S'\not=S$ with $S'\leq S$. The resulting extension problem is formally defined as the task to decide for a given pre-solution~$U$, if there exits a minimal solution $S$ which extends $U$.\par
%
%
%
We  systematically study the complexity of extension problems of quite a number of classical problems which fit into our framwork. Interestingly, extension variants tend to be  $\np$-hard, even if the underlying classical optimization problem is solvable in polynomial time. We encountered this observation, for instance, in connection with \textsc{Maximum Matching}. We also study the parameterized complexity of these problems, in particular under the parameterization of the size (or more general the value) of the fixed pre-solution or of its complement, as well as the complexity of extension problems assuming that the Exponential Time Hypothesis (ETH) holds\footnote{ETH is a conjecture asserting that there is no $2^{o(n)}$ (i.e., no sub-exponential) algorithms for solving 3-SAT, where $n$ is the number of variables; the number of clauses is somehow subsumed into this expression, as this number can be assumed to be sub-exponential in $n$ (after applying the famous sparsification procedure); cf.~\cite{ImpPatZan2001}.}. 
Further, for the graph problems considered here, we discuss the restriction to planar graphs, as for a certain argument to obtain PTAS for this restricted graph class, pioneered by Baker~\cite{Bak94}, it is important to know if solutions computed for certain parts of the graph can be (easily) extended to a solution of the overall instance.
Due to space limitations, many technical proofs have been moved to an appendix.

\section{A General Framework of Extension Problems and Notation}
We use standard notations from graph theory and only deal with undirected simple graphs that can be specified as $G=(V,E)$.
A subgraph $G'=(V',E')$ of $G$ is a partial graph if $V'=V$. If $U\subseteq V$, then $G[U]$ denotes the subgraph induced by $U$. For $X\subseteq E$,  $G[X]$ denotes the graph $(V,X)$ and $V(X)$ denotes the set of vertices incident to edges in~$X$. \par
In order to formally define our concept of minimal extension, we define what we call  \emph{monotone problems} which can be thought of as problems in $\npo$ with the addition of a set of pre-solutions (which includes the set of solutions) together with a partial ordering on this new set. Formally we define such problems as 5-tuples $\Pi=(\mathcal{I},\mathop{presol},\mathop{sol},\leq,m)$ (where $\mathcal{I},\mathop{sol},m$ with an additional $goal\in\{\min,\max\}$ yields an $\npo$ problem) defined by:
\begin{itemize}
\item $\mathcal{I}$ is the set of instances, recognizable in polynomial time.
\item For $I\in \mathcal{I}$, $\mathop{presol}(I)$ is the set of \emph{pre-solutions} and, in a reasonable representation of instances and pre-solutions, the length of the encoding of any $y\in \mathop{presol}(I)$ is polynomially bounded in the length of the encoding of $I$.
\item For $I\in \mathcal{I}$, $\mathop{sol}(I)$ is the \emph{set of solutions}, which is a subset of $\mathop{presol}(I)$.
\item There exists an algorithm which, given $(I,y)$, decides in polynomial time if $y\in\mathop{presol}(I)$; similarly there is an algorithm which decides in polynomial time if $y\in\mathop{sol}(I)$.
\item For $I\in \mathcal{I}$, $\leq_I$ is a partial ordering on $\mathop{presol}(I)$ and there exists an algorithm that, given an instance $I$ and $y,z\in \mathop{presol}(I)$, can decide in polynomial time if $y\leq_I z$.
\item For each  $I\in \mathcal{I}$, the set of solutions $\mathop{sol}(I)$ is upward closed with respect to the partial ordering $\leq_I$, i.e., $U\in\mathop{sol}(I)$ implies $U'\in\mathop{sol}(I)$ for all $U,U'\in \mathop{presol}(I)$ with $U\leq_I U'$.
\item $m$ is a polynomial-time computable function which maps pairs $(I,x)$ with $I\in \mathcal{I}$ and $x\in\mathop{presol}(I)$ to non-negative rational numbers; $m(I,x)$ is the \emph{value of $x$}.
\item For  $I\in \mathcal{I}$, $m(I,\cdot)$ is \emph{monotone} with respect to $\leq_I$, meaning that the property $y\leq_I z$ for some $y,z\in \mathop{presol}(I)$ either always implies $m(I,y)\leq m(I,z)$ or $m(I,y)\geq m(I,z)$.
\end{itemize}
The requirement that the set of solutions is upward closed with respect to the partial ordering relates to \emph{independence systems}, see \cite{AS03}.\par
Given a monotone problem $\Pi=(\mathcal{I},\mathop{presol},\mathop{sol},\leq,m)$, we denote by $\mu(\mathop{sol}(I))$ the set of \emph{minimal feasible solutions of $I$}, 
formally given by $$\mu(\mathop{sol}(I))=\{S\in \mathop{sol}(I)\suchthat  ((S'\leq_I S) \land (S'\in\mathop{sol}(I)))\rightarrow S'=S \}\,.$$
Further, given $U\in\mathop{presol}(I)$, we define $\mathop{ext}(I,U)=\{U'\in\mu(\mathop{sol}(I))\suchthat U\leq_I U'\}$ to be the 
set of \emph{extensions} of $U$. Sometimes, $\mathop{ext}(I,U)=\emptyset$, which makes the question of the existence of such extensions interesting. Hence, finally, the \emph{extension problem} for $\Pi$, written $\textsc{Ext}\ \Pi$, is defined as follows:
An instance of $\textsc{Ext}\ \Pi$ consists of an instance $I\in\mathcal{I}$ together with some $U\in\mathop{presol}(I)$, and the associated decision problems asks if  $\mathop{ext}(I,U)\neq\emptyset$.\par 
Although we strongly linked\LV{the definition of} monotone problems to $\npo$, the corresponding extension problems do not generally belong to $\np$ (in contrast to the\LV{canonical} decision problems associated to $\npo$ problems). 
Consider the monotone problem $\Pi_\tau=(\mathcal{I},\mathop{presol},\mathop{sol},\leq,m)$ with: \begin{itemize}
\item $\mathcal{I}=\{F\colon F \text{ is a Boolean formula}\}$. 
\item For a formula $F\in\mathcal{I}$ on $n$ variables, $\mathop{presol}(F)=\mathop{sol}(F)=\{\phi \mid \phi\colon\{1,\dots,n\}\rightarrow\nolinebreak \{0,1\}\}$.
\item For $\phi, \psi \in \mathop{presol}(F)$, $\phi\leq_F \psi$ if either $\phi=\psi$, or assigning variables according to $\psi$ satisfies $F$ while an assignment according to $\phi$ does not.
\item $m\equiv 1$ (plays no role for the extension problem)
\end{itemize}

The associated extension problem $\textsc{Ext}\ \Pi_\tau$ corresponds to the $\mathsf{co\text{-}NP}$-complete  {\sc Tautology Problem} in the following way: Given a Boolean formula $F$ which, w.l.o.g., is satisfied by the all-ones assignment $\psi_1\equiv 1$, it follows that $(F,\psi_1)$ is a \yes-instance for $\textsc{Ext}\ \Pi_\tau$ if and only if $F$ is a tautology, as $\psi_1$ is in $\mu(\mathop{sol}(F))$ if and only if there does not exist some $\psi_1\not=\phi\in\mathop{sol}(F)$ with $\phi\leq_F \psi_1$, so, by definition of the partial ordering, an assignment $\phi$ which does not satisfy $F$. Consequently $\textsc{Ext}\ \Pi_\tau$ is not in $\np$, unless  $\mathsf{co\text{-}NP}=\mathsf{NP}$.

Let us mention some\LV{well-known}  monotone graph problems,\LV{ some of which we will discuss later,} for which $\mathcal{I}$ is the set of undirected graphs, denoting instances by $I=(V,E)$, and 
$m(I,U)=|U|$ for all $U\in \mathop{presol}(I)$:
\begin{itemize}
 \item \textsc{Vertex Cover} (\textsc{VC}): 
 $\leq_I=\subseteq$, 
 $\mathop{presol}(I)=2^V$, $C\in \mathop{sol}(I)$ iff each 
 $e\in E$ is incident to at least one $v\in C$;
 \item \textsc{Edge Cover} (\textsc{EC}): 
 $\leq_I=\subseteq$, 
 $\mathop{presol}(I)=2^E$, $C\in \mathop{sol}(I)$ iff each $v\in V$
  is incident to at least one $e\in C$;
  \item \textsc{Independent Set} (\textsc{IS}):  $\leq_I=\supseteq$, 
  $\mathop{presol}(I)=2^V$, $S\in \mathop{sol}(I)$ iff
  $G[S]$ contains no edges;
  \item \textsc{Edge Matching} (\textsc{EM}): $\leq_I=\supseteq$, 
  $\mathop{presol}(I)=2^E$, $S\in \mathop{sol}(I)$ iff
  none of the vertices in $V$ is incident to more than one edge in $S$;
  \item \textsc{Dominating Set} (\textsc{DS}): $\leq_I=\subseteq$,
 $\mathop{presol}(I)=2^V$, $D\in \mathop{sol}(I)$ iff $N[D]=V$;
  \item \textsc{Edge Dominating Set} (\textsc{EDS}): $\leq_I=\subseteq$,
 $\mathop{presol}(I)=2^E$, $D\in \mathop{sol}(I)$ iff each edge belongs to $D$ or is adjacent to some $e\in D$.
\end{itemize}
We hence arrive at problems like $\textsc{Ext VC}$ (or $\textsc{Ext IS}$, resp.), where the instance is specified by a graph $G=(V,E)$ and a vertex set~$U$, and the question is if there 
is some \textit{minimal} vertex cover $C\supseteq U$ (or some \textit{maximal} independent set $I\subseteq U$). 
Notice that the instance $(G,\emptyset)$ of $\textsc{Ext VC}$ can be solved by a greedy approach that gradually adds vertices, starting from $\emptyset$ as a feasible solution, until this is no longer possible without violating feasibility (since we do not request the solution to be \textit{minimum}). 
Similarly, $(G,V)$ is an easy instance of $\textsc{Ext IS}$. We will show that this impression changes for other instances.

So far, it might appear that only few examples exist for defining instance orderings $\leq_I$. The reader is referred to \cite{Man98} as a rich source of further instance orderings. 
Let us mention one other example.
\smallskip\noindent
\emph{Bin Packing.}
Here, we make use of the well-known partition ordering.
The underlying optimization problem is \textsc{Bin Packing}, or \textsc{BP} for short, formalized as follows. The input consists of a set 
$X=\{x_1,\dots,x_n\}$ of items and a weight function $w$ that associates
rational numbers $w(x_i)\in(0,1)$ to items. A feasible solution is a partition $\pi$ of $X$ such that, for each set $Y\in\pi$, $\sum_{y\in Y}w(y)\leq 1$. 
The traditional aim is to find a feasible $\pi$ such that $|\pi|$ is minimized. 
Now, $\mathop{presol}(X)$ collects all partitions of $X$. 
For two partitions $\pi_1,\pi_2$ of $X$, we write $\pi_1\leq_X\pi_2$ if $\pi_2$ is a refinement of~$\pi_1$, i.e., 
$\pi_2$ can be obtained from $\pi_1$ by splitting up  its sets into a larger number of smaller sets.
Hence, $\{X\}$ is the smallest partition with respect to $\leq_X$.
As a partition $\pi$ is a set, we can measure its size by its cardinality.
Clearly, the set of solutions is upward closed. Now, a solution is minimal 
if merging any two of its sets into a single set yields a partition~$\pi$ such that there is some $Y\in\pi$ with $w(Y):=\sum_{y\in Y}w(y)> 1$. 
In the extension variant, we are given a partition $\pi_U$ of $X$ (together with $X$ and $w$) and ask if there is any minimal feasible partition $\pi_U'$ with $\pi_U\leq_X\pi_U'$.
This describes the problem \textsc{Ext BP}.
One could think of encoding knowledge about which items should not be put together in one bin within the given partition $\pi_U$.

Further, we discuss the parameterized complexity of several extension problems, where we define the \emph{standard parameter} for an extension problem $\textsc{Ext}\ \Pi$ for a monotone problem  $\Pi=(\mathcal{I},\mathop{presol},\mathop{sol},\leq,m)$ to be the value of the given pre-solution, i.e., the parameter for instance $(I,U)$ of $\textsc{Ext}\ \Pi$ is $m(U)$. Accordingly, for dual parameterization, we consider the difference of the value of the given pre-solution to the maximum $m_{max}(I):=\max\{m_I(y)\colon y\in \mathop{presol}(I)\}$, so the parameter for instance $(I,U)$ of $\textsc{Ext}\ \Pi$ is $m_{max}(I)-m(U)$.
\begin{table}[t]
\begin{center}\scalebox{1}{\begin{tabular}{@{\,\,}c@{\,\,}|@{\,\,}c@{\,\,}|@{\,\,}c@{\,\,}|@{\,\,}c@{\,\,}|@{\,\,}c@{\,\,}|@{\,\,}c@{\,\,}}
\diagbox{\small Param.
}{\small Ext. 
of}  &  \sc EC& \sc EM & \sc DS & \sc EDS&\sc BP\\\hline
standard &$\fpt$ & $\fpt$ & $\wthree$-complete & $\wone$-hard & $\paranp$-hard \\
dual&   $\fpt$ & $\fpt$ &  $\fpt$ &  $\fpt$ & $\fpt$\\
\end{tabular}}
\end{center}
\caption{Survey on parameterized complexity results for extension problems\label{tab-survey}}
\end{table}
\paragraph*{Summary of Results}
For all problems that we consider, we obtain $\np$-completeness results, for graph problems even when restricted to planar bipartite graphs of maximum degree three.
Clearly, we know that these hardness results are optimal with respect to the degree bound. 
Observe that extension problems can behave quite differently from the classical (underlying) decision problems with respect to \text{(in-)}\linebreak[3]tractability. 
For instance, \textsc{EC} is solvable in polynomial time, while \textsc{Ext EC} is $\np$-hard. We also study this phenomenon more in depth by defining generalizations of edge cover and matching problems such that the simple optimization problems can be solved in polynomial time, while the extension variants are  $\np$-hard.
All our  $\np$-hardness results translate into ETH-hardness results, as well. All ETH-hardness results are matched by corresponding algorithmic results. 
We further obtain parameterized complexity results as surveyed in Table~\ref{tab-survey}. The hardness results  for graph problems  contained in this table also hold  for the restriction to bipartite graph instances.\par
\section{\texorpdfstring{$\np$}{NP}-Completeness Results}
In this section, we present computational complexity results for some well known graph problems.
Most results are reductions from one of the following two variants of satisfiability. The first is\LV{known as}  
{\sc 2-balanced 3-SAT}, denoted by \textsc{$(3,B2)$-SAT}. An instance $I=({\cal C},{\cal X})$  of \textsc{$(3,B2)$-SAT} is \LV{given by} a set ${\cal C}$ of CNF clauses defined over a set
${\cal X}$ of Boolean variables such that each clause
has exactly $3$ literals, and \LV{such that} each variable appears exactly $4$ times
in ${\cal C}$, twice negated and twice unnegated.
The bipartite graph associated to instance $I=({\cal C},{\cal X})$ is \LV{the graph} $BP=(C\cup X,E(BP))$ with $C=\{c_1,\dots,c_m\}$, $X=\{x_1,\dots,x_n\}$ and $E(BP)=\{c_jx_i\suchthat x_i$ or $\neg  x_i$ is literal of $c_j\}$. 
\LV{Deciding whether an instance of \textsc{$(3,B2)$-SAT} is satisfiable is $\np$-complete \cite[Theorem~1]{ECCC-TR03-049}.}
\SV{ \textsc{$(3,B2)$-SAT} is $\np$-complete by~\cite[Theorem~1]{ECCC-TR03-049}.}
The other problem used in our reductions is {\sc 4-Bounded Planar 3-Connected SAT} (\textsc{4P3C3SAT} for short), the restriction of {\sc exact 3-satisfiability} to clauses in ${\cal C}$ over variables in ${\cal X}$, where each variable occurs in at most four clauses (at least one time negated and one time unnegated) and the associated bipartite graph $BP$ is planar of maximum degree~4. This restriction is also  $\np$-complete~\cite{Kra94}.

Let us summarize our findings in the following statement.
\begin{theorem}\label{thm:np-completeness-summary}
Let $\mathcal{P}\in\{\textsc{EC},\textsc{EM},\textsc{DS},\textsc{EDS}\}$. Then, $\textsc{Ext }\mathcal{P}$ is $\np$-complete on bipartite graphs of maximum degree 3.
\end{theorem}

We make one of the possibly most surprising results explicit, as the underlying optimization problem is well-known to be polynomial-time solvable. Not only in this case, there are additional properties that graph instances might satisfy, still maintaining $\np$-hardness.

\begin{theorem}\label{Bip_Ext_EC}
\textsc{Ext EC} is $\np$-complete on bipartite graphs of maximum degree 3, even if the given pre-solution forms an edge matching.
\end{theorem}
\begin{proof}
We reduce from $(3,B2)$-{\sc SAT}, so let ${\cal I}$ be an instance of $(3,B2)$-{\sc sat} with clauses ${\cal C}=\{c_1,\dots,c_m\}$ and variables ${\cal X}=\{x_1,\dots,x_n\}$.\par  
From the bipartite graph $BP$ associated to $I$, we build a graph $G=(V,E)$ by splitting every vertex $x_i\in X$ by a $P_5$ denoted $P_i=(x_i,l_i,m_i,r_i,\neg  x_i)$ where now $x_i$ (resp., $\neg  x_i$) is linked to $c_j$ if $x_i$ appears unnegated (resp., negated) in $c_j$ (see \autoref{GadgfigExtEC}). It is easy to see that $G$ is bipartite of maximum degree 3.
Finally, let $U=\{x_il_i,\neg x_ir_i\suchthat 1\leq i \leq n\}$. We claim that $I$ is satisfiable iff $G$ admits a minimal edge cover containing~$U$.

\begin{figure}[tbh]
\centering
\begin{tikzpicture}[scale=0.8, transform shape]
\tikzstyle{vertex}=[circle, draw, inner sep=0pt, inner sep=0pt, minimum size=0.8cm]
\node[vertex] (c1) at (0,0) {$c_1$};
\node[vertex, below of=c1] (c2) {$c_2$};
\node[vertex, below of=c2] (c3) {$c_3$};
\node[below of=c3,node distance=1cm] (cdots) {$\vdots$};
\node[vertex,below of=cdots] (cm) {$c_m$};

\node[vertex, right of=c1, node distance=4cm] (x1) {$x_1$};
\node[vertex, below of=x1] (x2) {$x_2$};
\node[below of=x2] (xdots) {$\vdots$};
\node[vertex, below of=xdots] (xn) {$x_n$};

\node[vertex, right of=x1, node distance=2cm] (l1) {$l_1$};
\node[vertex, below of=l1] (l2) {$l_2$};
\node[below of=l2] (ldots) {$\vdots$};
\node[vertex, below of=ldots] (ln) {$l_n$};

\node[vertex, right of=l1, node distance=2cm] (m1) {$m_1$};
\node[vertex, below of=m1] (m2) {$m_2$};
\node[below of=m2] (mdots) {$\vdots$};
\node[vertex, below of=mdots] (mn) {$m_n$};

\node[vertex, right of=m1, node distance=2cm] (r1) {$r_1$};
\node[vertex, below of=r1] (r2) {$r_2$};
\node[below of=r2] (rdots) {$\vdots$};
\node[vertex, below of=rdots] (rn) {$r_n$};

\node[vertex, right of=r1, node distance=2cm] (nx1) {$\neg x_1$};
\node[vertex, below of=nx1] (nx2) {$\neg x_2$};
\node[below of=nx2] (nxdots) {$\vdots$};
\node[vertex, below of=nxdots] (nxn) {$\neg x_n$};

\draw (x1) edge[ultra thick] (l1) (l1) -- (m1) -- (r1) (r1) edge[ultra thick] (nx1);
\draw (x2) edge[ultra thick] (l2) (l2) -- (m2) -- (r2) (r2) edge[ultra thick] (nx2);
\draw (xn) edge[ultra thick] (ln) (ln) -- (mn) -- (rn) (rn) edge[ultra thick] (nxn);

\draw (c1) edge[bend left=20] (nx1);
\draw (c1) edge[] (x2);
\draw (c1) edge[] (xn);
\draw (cm) -- (x1) -- (c2) -- (x2);
\draw (cm) edge[bend right=14] (nxn);

\end{tikzpicture}
 \caption{The graph $G=(V,E)$ for \textsc{Ext EC} built from $I$ with $m+5n$ vertices and $3m+4n$ edges. Edges of $U$ are drawn bold. In this example, $c_1=\{\neg x_1,x_2,x_n\}$.}
 \label{GadgfigExtEC}
\end{figure}
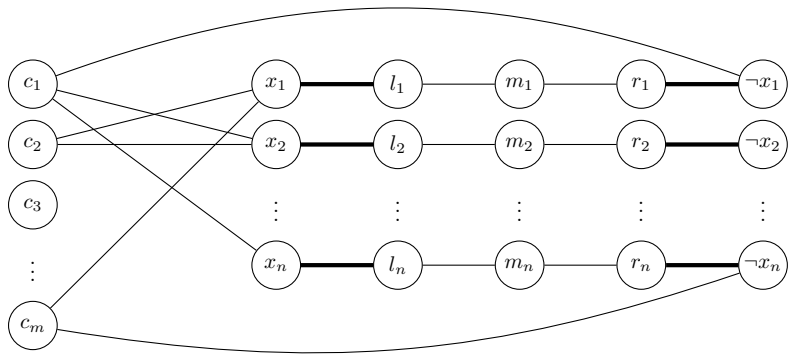

Assume $I$ is satisfiable and let $T$ be a truth assignment which satisfies all clauses. For each clause $c_j$, let $f(j)$ be an index such that variable $x_{f(j)}$ satisfies clause $c_j$ for $T$ and let $J=[n]\setminus f([m])$ be the indices not used by mapping $f$. We set
\begin{align*}
S&=\{x_{f(j)}c_j,m_{f(j)}r_{f(j)}\suchthat T(x_{f(j)})=\textit{true},~x_{f(j)}\mbox{~appears positively in~}c_j\}\\&\cup \{\neg x_{f(j)}c_j,m_{f(j)}\ell_{f(j)}\suchthat T(x_{f(j)})=\textit{false},~x_{f(j)}\mbox{~appears negatively in~}c_j\}\\&\cup U\cup \{m_ir_i\suchthat i\in J\}\,.\end{align*}
We can easily check that $S$ is a minimal edge cover containing $U$.

Conversely, assume that $G$ has a minimal edge cover $S$ containing $U$. In order to cover vertex $m_i$ for
each $i\in\{1,\dots,n\}$, the edge cover $S$ contains either $l_im_i$ or $r_im_i$ (not both by minimality). This means
that if we set $T(x_i)=\textit{true}$ if $r_im_i\in S$ and $T(x_i)=\textit{false}$ if $l_im_i\in S$, we obtain a valid assignment $T$.
This  assignment satisfies all clauses of $I$, since $S$ must cover all vertices of $C$.
\end{proof}

\newcommand{\allofExtEM}{\begin{theorem}\label{degree3_Ext_EM}
\textsc{Ext EM} is $\np$-complete on bipartite graphs of maximum degree~3.
\end{theorem}

\newcommand{\eopExtEM}{If $T$ is a truth assignment of $I$ which satisfies all clauses, then we add the set of crossing edges
$x_cx^{c}$ as well as $U_x$ if $T(x)=\textit{true}$; otherwise, if $T(x)=\textit{false}$, we add the crossing edges $\neg x_c\neg x^{c}$ and all edges in  $U_{\neg x}$.
For each clause $c$, we determine the literals which satisfy the clause (either one, two or three; let $\sharp c$ be the number of such literals); then, we add $3-\sharp  c$ edges saturating vertices $1_c$, $2_c$ and $3_c$. For instance, assume  it is only $\ell_2$ 
(thus, $\sharp  c=1$). Then, we add two edges saturating
vertices $1_c$ and $3_c$ and the unnamed vertices in gadget clause $H(c)$. The resulting matching $S$ is inclusion-wise maximal 
and satisfies $S\subseteq U$.
 \par 
Conversely, assume the existence  of  a maximal matching $S$ with $S\subseteq U$. Hence, for each variable $x\in X$ exactly one edge
between $e_x$ and $e_{\neg x}$ belongs to $S$ (in order to block edge $3_x4_x$). If it is $e_x\in S$ (resp., $e_{\neg x}\in S$),
then $U_x\subseteq S$ (resp., $U_{\neg x}\subseteq S$). Hence, $S$ does not contain any crossing edges saturating $\neg x^c$ (resp. $x^c$)).
Now, for each clause $c=\ell_1\lor \ell_2\lor \ell_3$,
at least one vertex among $\ell^1_c,\ell^2_c,\ell^3_c$
must be adjacent to a crossing edge of $S$. In conclusion, by setting $T(x)=\textit{true}$ if at least one vertex $x^{c_1}$ or $x^{c_2}$ of $H(x)$ is saturated by $S$ and $T(x)=\textit{false}$ otherwise we get a valid assignment $T$  satisfying all clauses.}

\begin{proof}
The proof is based on a reduction from \textsc{$(3,B2)$-SAT}. Consider an instance of \textsc{$(3,B2)$-SAT} with clauses ${\cal C}=\{c_1,\dots,c_m\}$ and variables ${\cal X}=\{x_1,\dots,x_n\}$. We build a bipartite graph $G=(V,E)$ of maximum degree 3, together with a set $U$ of permitted edges (among which a maximal matching should be chosen) as follows:
\begin{itemize}

\item[$\bullet$] For each clause $c=\ell_1\vee \ell_2 \vee \ell_3$ where $\ell_1,\ell_2,\ell_3$ are literals, introduce a subgraph $H(c)=(V_c,E_c)$ with 8 vertices and 7 edges. $V_c$ contains three specified literal vertices $\ell^1_c, \ell^2_c, \ell^3_c$.
Moreover, $F_c=\{\ell^1_c1_c, \ell^2_c2_c, \ell^3_c3_c\}$
is the set of three forbidden edges included in $H(c)$, so that $U_c=E_c\setminus F_c$. The gadget $H(c)$ is
illustrated in the left part of Figure~\ref{fig:degree3_Ext_EM}.

\item[$\bullet$] 
For each variable $x$ introduce 12 new vertices which induce the subgraph $H(x)=(V_x,E_x)$ illustrated in Figure~\ref{fig:degree3_Ext_EM}. The vertex set $V_x$ contains four special vertices $x^{c_1}$, $x^{c_2}$, $\neg x^{c_3}$ and $\neg x^{c_4}$,
where it is implicitly  supposed that variable $x$ appears positively in clauses $c_1,c_2$ and negatively in clauses $c_3,c_4$. 
Define $U_x=\{e_x\}\cup \{2_x^{c_3}\neg x^{c_3},2_x^{c_4}\neg x^{c_4}\}$
and $U_{\neg x}=\{e_{\neg x}\}\cup\{1_x^{c_1}x^{c_1},1_x^{c_2}x^{c_2}\}$. 

\item[$\bullet$] We interconnect  $H(x)$ and $H(c)$ where $x$ is a 
variable occurring in literal $\ell_i$
of clause $c$ by adding the edge $\ell_i^{c}\ell^i_c$, with $\ell_i^c$ from $H(x)$ and $\ell^i_c$ from $H(c)$. 
These {\em crossing edges}  
are always permitted and collected in the set $U_{\text{cross}}$.
\end{itemize}

Let $U=(\bigcup_{c\in C} U_c)\cup (\bigcup_{x\in X} (U_x\cup U_{\neg x}))\cup U_{\text{cross}}$. This construction is computable in polynomial time and $G$ is a bipartite graph of maximum degree~3.

\begin{figure}
\centering
\begin{tikzpicture}[scale=0.7, transform shape]
\tikzstyle{vertex}=[circle, draw, inner sep=0pt, inner sep=0pt, minimum size=0.6cm]


\node () at (2,-3) {$H(c)$ for $c=\ell_1\vee \ell_2 \vee \ell_3$}; 
\node[vertex] (xi) at (0.5,0) {$\ell^1_c$};
\node[vertex] (xj) at (2,0) {$\ell^2_c$};
\node[vertex] (xk) at (3.5,0) {$\ell^3_c$};
\node[vertex] (a) at (0.5,-1) {$1_c$};
\node[vertex] (b) at (2,-1) {$2_c$};
\node[vertex] (c) at (3.5,-1) {$3_c$};
\node[vertex] (e) at (1.25,-2) {};
\node[vertex] (f) at (2.75,-2) {};

\draw (xi) edge[ultra thick] (a);
\draw (xj) edge[ultra thick] (b);
\draw (xk) edge[ultra thick] (c);
\draw (a) -- (e);
\draw (b) -- (e);
\draw (b) -- (f);
\draw (c) -- (f);

\begin{scope}[xshift=6cm]
\node () at (1.5,-3) {$H(x)$};
\node[vertex] (4) at (0,0) {$1_x$};
\node[vertex] (5) at (-0.5,-1) {$1^{c_1}_x$};
\node[vertex] (3) at (0.5,-1) {$1^{c_2}_x$};
\node[vertex] (xi1) at (-0.5,-2) {$x^{c_1}$};
\node[vertex] (xi2) at (0.5,-2) {$x^{c_2}$};
%
\node[vertex] (4p) at (3,0) {$2_x$};
\node[vertex] (5p) at (2.5,-1) {$2^{c_3}_x$};;
\node[vertex] (3p) at (3.5,-1) {$2^{c_4}_x$};;
\node[vertex] (nxi1) at (2.5,-2) {$\neg x^{c_3}$};
\node[vertex] (nxi2) at (3.5,-2) {$\neg x^{c_4}$};
%
\node[vertex] (7) at (1.5,1) {$3_x$};
\node[vertex] (8) at (1.5,2) {$4_x$};

\draw (4) edge[ultra thick] (5) (5) -- (xi1); 

\draw (xi2)--(3) (3) edge[ultra thick] (4);
\draw (4p) edge[ultra thick] (5p) (5p) -- (nxi1);
\draw (nxi2)--(3p) (3p) edge[ultra thick] (4p);

\draw (4) -- node[midway,above] {$e_{x}$} (7);
\draw (4p) -- node[midway,above] {$e_{\neg x}$} (7);
\draw (7) edge[ultra thick] (8);

\end{scope}
\end{tikzpicture}
\caption{The Gadgets $H(c)$ and $H(x)$ for \textsc{Ext EM}. Edges not in $U$ are marked with bold line.}\label{fig:degree3_Ext_EM}
\end{figure}

\smallskip

We claim that there is a truth assignment of $I$ which satisfies all clauses iff there is a maximal edge matching $S$ of $G$
which contains only edges from $U$. 
\smallskip
\eopExtEM
\end{proof}}

\newcommand{\eopExtDS}{
Suppose $T$ is a satisfying assignment for $I$. Create a dominating set $S$ from $U$ by adding for each $x\in \mathcal{X}$  the literal vertex $x$ if $T(x)=\textit{true}$ and the literal vertex $\neg x$ if $T(x)=\textit{false}$. Also, add to $S$ for each clause  $c\in \mathcal{C}$ the  vertex $1_c$ if $1'_c$  is not dominated by a  literal vertex in $S$ and $2_c$ if $2'_c$  is not dominated by a literal vertex in $S$. The resulting set $S$ is obviously a dominating set for $G$ which contains $U$. Since $T$ satisfies all clauses in $I$, $S$ does not contain both $1_c$ and $2_c$ for any clause $c$, so the vertex $3_c$ has at least one private neighbor for each $c\in \mathcal{C}$. Since $T$ further is a valid assignment, $S$ only contains $x$ or $\neg x$ for each variable $x\in \mathcal{X}$, so $1_x$ also has a private neighbor. If $S$ is not minimal, it can hence be turned into a minimal dominating set $S'$ by successively removing vertices without private neighbors from the ones that we added to $U$. This could happen if there is a variable $x$ whose setting does not matter, as all clauses that contain $x$ or $\neg x$ are already satisfied by the other literals. Then, the corresponding literal vertex ($x$ or $\neg x$) put into $S$ can be removed. 
The resulting minimal dominating set $S'$ still contains all vertices from $U$; observe that vertices $4_c$ also have a private neighbor $5_c$. The vertices $4_c$ and $5_c$ are present in the gadgets to prevent $3_c$ to consider itself as its private neighbor.
\medskip

Conversely, assume that $S$ is a minimal dominating set of $G$ with $U\subseteq S$. Because of minimality, $3_c$ needs a private neighbor, either $1_c$ or $2_c$. Hence, 
$S$ contains at most one vertex from $\{1_c,1'_c,2_c,2'_c\}$ for each clause $c$. In particular, there is at least one vertex among $\{1'_c,2'_c\}$ which needs to be dominated by a literal vertex. 
Again by minimality, $1_x$ needs a private neighbor, either $x$ or $\neg x$. Hence,  the two literal vertices $x$ and $\neg x$ cannot be together in~$S$. Thus, by setting $T(x)=\textit{true}$ (resp., $\textit{false}$) if $x\in S$ (resp., $x\notin S$), we arrive at a partial truth assignment of $I$, satisfying all clauses, that can be easily completed.}

\newcommand{\allofExtDS}{\begin{theorem}\label{Bip_Ext_DS}
\textsc{Ext DS} is $\np$-complete on bipartite graphs of maximum degree 3, even if the subgraph $G[U]$ induced by the pre-solution $U$ is an induced matching. 
\end{theorem}
\begin{proof}
The proof is based on a reduction from \textsc{$(3,B2)$-sat} as defined before. For an instance $I$ of \textsc{$(3,B2)$-sat} with clause set ${\cal C}=\{c_1,\dots,c_m\}$ and variable set ${\cal X}=\{x_1,\dots,x_n\}$, we build a bipartite graph $G=(V,E)$ with maximum degree 3,
together with a set $U\subseteq V$ of forced vertices  as an instance of \textsc{EXT  DS} as follows:
\begin{itemize}

\item[$\bullet$] For each clause  $c=\ell_1\vee \ell_2\vee \ell_3$ where $\ell_1,\ell_2,\ell_3$ are literals, we introduce the subgraph $H(c)=(V_c,E_c)$ with 7 vertices and 6 edges as illustrated on the left side of Fig.~\ref{fig:degree3_Ext_DS}. The vertices $1'_c$ and $2'_c$ represent literals in clause $c$ ($1'_c$ represents literals $\ell_1$ and $\ell_2$ while $2'_c$ represents $\ell_3$) and the vertices $\{3_c,4_c\}$ are included in the forced set $U_c$. 

\item[$\bullet$] For each variable $x$, we introduce a gadget $H(x)=(V_x,E_x)$ which is a path of length 2. The vertex $1_x$ is in the set of forced vertices $U_x$. An illustration of variable subgraph $H(x)$ is on the right side of Fig.~\ref{fig:degree3_Ext_DS}.

\item[$\bullet$] We interconnect the subgraphs $H(x)$ and $H(c)$ in the following way: for each clause~$c$ with literals $\ell_1,\ell_2,\ell_3$, corresponding to variables $y_1,y_2,y_3$, respectively, connect $1'_c$ (representing $\ell_1$ and $\ell_2$) to the literal vertices  $\ell_1$ in $H(y_1)$ and $\ell_2$ in $H(y_2)$ and connect~$2'_c$ (representing $\ell_3$) to the literal vertex $\ell_3$ in $H(y_3)$.
\item[$\bullet$] Also we set $U=(\bigcup_{c\in {\cal C}}U_c)\cup (\bigcup_{x\in{{\cal X}}}U_x)$.
\end{itemize}
 This construction computes,  in polynomial time, a bipartite graph $G$ with maximum degree~3. 
Moreover, $G[U]$ is an induced matching. 
\LV{In the following w}\SV{W}e can prove that there exists a satisfying assignment $T$ for $I$ iff $(G,U)$ is a \yes-instance of \textsc{Ext DS}.
\eopExtDS
\end{proof}
\begin{figure}
\centering
\begin{tikzpicture}[scale=0.9, transform shape]
\tikzstyle{vertex}=[circle, draw, inner sep=0pt, inner sep=0pt, minimum size=0.6cm]

\node () at (0.5,-3) {$H(c)$ for $c=\ell_1\vee \ell_2 \vee \ell_3$};
\node[vertex] (xyc) at (2,-2) {$1'_c$};
\node[vertex] (zc) at (2,0) {$2'_c$};
\node[vertex] (1c) at (1,-2) {$1_c$};
\node[vertex] (2c) at (1,0) {$2_c$};
\node[vertex, ultra thick] (3c) at (0,-1) {$3_c$};
\node[vertex, ultra thick] (4c) at (-1,-1) {$4_c$};
\node[vertex] (5c) at (-2,-1) {$5_c$};

\draw (xyc) -- (1c)-- (3c);
\draw (zc) -- (2c) -- (3c);
\draw (3c) -- (4c);
\draw (4c) -- (5c);

\begin{scope}[xshift=1cm]

\node () at (4.5,-3) {$H(x)$};
\node[vertex] (x) at (4,0) {$x$};
\node[vertex] (nx) at (4,-2) {$\neg x$};
\node[vertex, ultra thick] (x1) at (5,-1) {$1_x$};

\draw (nx) -- (x1)--(x);

\end{scope}
\end{tikzpicture}
\caption{The Gadgets $H(c)$ and $H(x)$ for \textsc{Ext DS}. Vertices in $U$ illustrated by their bold border.}\label{fig:degree3_Ext_DS}
\end{figure}}

\LV{\noindent With a similar construction\SV{ (see appendix)}, we can show:}
\newcommand{\allofExtEDS}{\begin{theorem}\label{Bip_Ext_EDS}
\textsc{Ext EDS} is $\np$-complete on bipartite graphs of maximum degree~3, even if the partial subgraph $(V,U)$ induced by the pre-solution $U$
is an induced collection of $P_3$.
\end{theorem}}
\newcommand{\pfExtEDS}{
\begin{proof}
\begin{figure}
\centering
\begin{tikzpicture}[scale=1, transform shape]
\tikzstyle{vertex}=[circle, draw, inner sep=0pt, inner sep=0pt, minimum size=0.6cm]
\node () at (1.5,-3) {$H(c)$ for clause $c=\ell_1\vee \ell_2 \vee \ell_3$};
\node[vertex] (xyc) at (4,-2) {$1'_c$};
\node[vertex] (zc) at (4,0) {$2'_c$};
\node[vertex] (1c) at (3,-2) {$1_c$};
\node[vertex] (3c) at (3,0) {$2_c$};
\node[vertex] (4c) at (2,-1) {$3_c$};
\node[vertex] (5c) at (1,-1) {$4_c$};
\node[vertex] (6c) at (0,-1) {$5_c$};
\node[vertex] (7c) at (-1,-1) {$6_c$};
\draw (xyc) -- (1c)-- (4c);
\draw (zc) -- (3c) -- (4c);
\draw (4c) edge [ultra thick] (5c);
\draw (6c) edge [ultra thick] (5c);
\draw (6c) -- (7c);
\begin{scope}[xshift=2cm]
\node () at (6.5,-3) {$H(x)$ for variable $x$};
\node[vertex] (x) at (4,0) {$x$};
\node[vertex] (nx) at (4,-2) {$\neg x$};
\node[vertex] (x1) at (5,-1) {$1_x$};
\node[vertex] (x2) at (6,-1) {$2_x$};
\node[vertex] (x3) at (7,-1) {$3_x$};
\node[vertex] (x4) at (8,-1) {$4_x$};
\draw (x) -- (x1) edge [ultra thick] (x2) (x2) edge [ultra thick] (x3) (x3) -- (x4);
\draw (nx) -- (x1);
\end{scope}
\end{tikzpicture}
\caption{The Gadgets $H(c)$ and $H(x)$ for \textsc{Ext EDS}. Edges in $U$ are marked with bold line.}\label{fig:degree3_Ext_EDS}
\end{figure}
The proof is similar to the proof of Theorem \ref{Bip_Ext_DS}. We start with instance of  \textsc{($3,B2$-sat)} with clauses ${\cal C}=\{c_1,\dots,c_m\}$ and variables ${\cal X}=\{x_1,\dots,x_n\}$ and  build a bipartite graph $G=(V,E)$ of maximum degree 3 as described in Figure \ref{fig:degree3_Ext_EDS}. Here the clause gadget $H(c)=(V_c,E_c)$ just has a new vertex $6_c$ and a new edge $5_c6_c$ while the variable gadget $H(x)=(V_x,E_x)$ has an additional vertex $4_x$ and additional edge $3_x4_x$. As forced edges we include the sets $U_c=\{3_c4_c,4_c5_c\}$ and $U_x=\{1_x2_x,2_x3_x\}$  for each clause gadget $H(c)$ and variable gadget $H(x)$, respectively, and the overall set of forced edges is given by $U=(\bigcup_{x\in {\cal X}}U_x)\cup (\bigcup_{c\in {\cal C}}U_x)$. Fig.~\ref{fig:degree3_Ext_EDS} proposes an illustration of $H(c)$ and $H(x)$.\par
Clearly $G$ is bipartite with maximum degree 3 and is constructed in polynomial time. Similarly to Theorem~\ref{Bip_Ext_DS}, we claim that $I$ is satisfiable iff $G$ has a minimal edge dominating set containing $U$.
\end{proof}}
\LV{\pfExtEDS}
Discussing the complexity of \textsc{Ext BP} requires a quite different approach. We reduce from \textsc{3-Partition} which asks for a given multiset $S=\{s_1,\dots,s_{3m}\}$ of intergers and $b\in \mathbb N$ if~$S$ can be partitioned into $m$ triples such that the sum of each subset equals~$b$. 
\textsc{3-Partition} is $\np$-complete even if $b/4< s_i < b/2$ for each $i\in\{1,\dots,3m\}$, see~\cite{GJ79}. As corresponding \textsc{Ext BP} instance we build the set $X=\{x_0,x_1,\dots,x_{3m}\}$ with weights $w(x_0)=\frac{m}{m+1}$ and $w(x_i)=\frac{s_i}{b}$ for each $1\leq i \leq 3m$ and set $\pi_U=\{\{x_0\}, \{x_1,\dots,x_{3m}\}\}$ as a partial partition of~$X$.
It can be shown that $(S,b)$ is a \yes-instance of \textsc{3-Partition} if and only if $(X,\pi_U)$ is a \yes-instance of \textsc{Ext BP} which yields:
\begin{theorem}\label{thm:Ext BP np-hard}
\textsc{Ext BP} is $\np$-hard, even if the pre-solution $\pi_U$ contains only two sets.
\end{theorem}
\newcommand{\pfBPnp}{
\begin{proof}
The proof consists of a reduction from \textsc{3-Partition} which is defined as follows:  given a multiset $S=\{s_1,\dots,s_{3m}\}$ of positive integers and a positive integer $b$ as input, decide if $S$ can be partitioned into $m$ triples $S_1,\dots,S_m$ such that the sum of each subset equals $b$. 
\textsc{3-Partition} is $\np$-complete even if each integer satisfies $b/4< s_i < b/2$; see~\cite{GJ79}.
Let $(S=\{s_1,\dots,s_{3m}\},b)$ be the input of \textsc{3-Partition}, where $b/4< s_i < b/2$ for each $1\leq i \leq 3m$. 
We  build a set $X=\{x_0,x_1,\dots,x_{3m}\}$ of items and a weight function $w$ where $w(x_0)=\frac{m}{m+1}$ and $w(x_i)=\frac{s_i}{b}$ for each $1\leq i \leq 3m$ and set $\pi_U=\{\{x_0\}, \{x_1,\dots,x_{3m}\}\}$ as a partial partition of~$X$. 
We claim that $(S,b)$ is a \yes-instance of \textsc{3-Partition} if and only if $(X,\pi_U)$ is a \yes-instance of \textsc{Ext BP}.
\smallskip

Suppose first that $S$ can be partitioned into $m$ triples $S_1,\dots,S_m$ where $\sum_{s_j\in S_i}s_j=b$ for each $S_i\in S$. 
We build a set $X_i=\{x_j\colon 1\leq j\leq 3m, s_j\in S_i\}$, $1\leq i\leq m$. 
Considering $\pi_U$, $\pi_U'=\{\{x_0\}, X_1,\dots,X_m\}$ is a feasible partition and since for each $S_i\in S$, $\sum_{s_j\in S_i}s_j=b$, we have $
w(X_i)=1$ for each $X_i\in \pi_U'$. Hence $\pi_U'$ is not the refinement of any other feasible partition for $(S,b)$, as especially $x_0$ cannot be added to any subset $X_i\in \pi_U'$. Since $\pi_U'$ is obviously a refinement of $\pi_U$, $\pi_U'$ is a minimal feasible partition with $\pi_U \leq_X \pi_U'$.

Conversely, assume that $\pi_U'$ is a minimal partition of $X$ as a refinement of $\pi_U$. As the set $\{x_0\}$ in the partition $\pi_U$ can not be split up further, it follows that the extension $\pi_U'$ is of the form $\{\{x_0\},X_1,\dots,X_k\}$. By using the minimality of $\pi_U'$, it follows especially that $\sum_{x_l\in X_i}w(x_l)+w(x_0)>1$ for all $i\in\{1,\dots,k\}$, as otherwise $\pi_U''=\{X_1,\dots,X_{i-1},X_i\cup\{x_0\},X_{i+1},\dots,X_k\}$ would be a feasible partition of $X$ with $\pi_U''\leq_X\pi_U'$. We claim that $k=m$. As $k<m$ is not possible, assume that $k>m$. Since $\sum_{i=1}^{3m}w(x_i)=\frac1b\sum_{i=1}^{3m}s_i=m$, this means that $5\sum_{x_l\in X_j}x_l
w(X_j)
\leq \frac{m}{k}$ for at least one $j\in\{1,\dots,k\}$, which contradicts $\sum_{x_l\in X_j}w(x_l)+w(x_0)>1$ by the definition of $x_0$. Consider the collection of the sets $S_i=\{s_j\colon 1\leq j\leq 3m, x_j\in X_i\}$, $1\leq i\leq m$ as a partition for $S$. By feasibility of $\pi_U'$, it follows that $
w(X_i)\leq 1$, which means $\sum_{s_l\in S_i}s_l\leq b$ and $k=m$ implies that indeed $\sum_{s_l\in S_i}s_l= b$ for each $i\in \{1,\dots,m\}$. The requirement $b/4< s_i <b/2$ for each $1\leq i\leq 3m$ then implies that the size of each $X_i$ equals 3, which overall means that $S_1,\dots,S_m$ is a solution for  \textsc{3-Partition} on $(S,b)$.
\end{proof}
}

\subsubsection*{Planar Graphs}

The following statement appears to be only strengthening Theorem~\ref{thm:np-completeness-summary}, but the details behind can be different indeed. We exemplify this by one concrete example theorem.

\begin{theorem}\label{thm:np-completeness-planar-summary}
Let $\mathcal{P}\in\{\textsc{EC},\textsc{EM},\textsc{DS},\textsc{EDS}\}$. Then, $\textsc{Ext }\mathcal{P}$ is $\np$-complete on planar bipartite graphs of maximum degree~3.
\end{theorem}

All reductions are from \textsc{4P3C3SAT}. This gives us a planar \emph{vertex-clause-graph} $G$, corresponding to the original \textsc{SAT} instance $I$. 
The additional technical difficulties come with the embeddings that need to be preserved. 
Suppose that a variable $x_i$ appears in at most four clauses $c_1,c_2,c_3,c_4$ of $I$ such that in the induced (embedded) subgraph $G_i=G[\{x_i,c_1,c_2,c_3,c_4\}]$, $c_1x_i$, $c_2x_i$, $c_3x_i$, $c_4x_i$ is an anti-clockwise ordering of edges around $x_i$. By looking at $G_i$ and considering how variable $x_i$ appears negated or non-negated in the four clauses $c_1,c_2,c_3,c_4$ in $I$, the construction should handle the~3 following cases:

\begin{itemize}
\item case 1: $x_i\in c_1, c_2$ and $\neg x_i \in c_3,c_4$,
\item case 2: $x_i\in c_1,c_3$ and $\neg x_i \in c_2,c_4$,
\item case 3: $x_i\in c_1,c_2,c_3$ and $\neg x_i \in c_4$.
\end{itemize}

\noindent
All other cases are included in these 3 cases by rotations and / or interchanging $x_i$ with~$\neg x_i$. We illustrate how these cases are used in the reductions explicitly for \textsc{Ext EC}. While the interconnections of the clause gadgets and the variable gadgets are similar to the non-planar case, the variable gadgets differ according to the cases listed above, see Figure~\ref{Fig:Ext planar EC}.

\begin{figure}
\centering
\begin{tikzpicture}[scale=0.8, transform shape]
\tikzstyle{vertex}=[circle, draw, inner sep=2pt,  minimum width=1 pt, minimum size=0.1cm]
\tikzstyle{vertex1}=[circle, draw, inner sep=2pt, fill=black!100, minimum width=1pt, minimum size=0.1cm]

\node[vertex] (x) at (-4,-3) {};
\node[vertex] (c1) at (-5,-1.5) {};
\node[vertex] (c2) at (-5,-2.5) {};
\node[vertex] (c3) at (-5,-3.5) {};
\node[vertex] (c4) at (-5,-4.5) {};
\draw (x)--(c1);
\draw (x)--(c2);
\draw (x)--(c3);
\draw (x)--(c4);
\node () at (-3.7,-3) {$x_i$};
\node () at (-5.3,-1.5) {$c_1$};
\node () at (-5.3,-2.5) {$c_2$};
\node () at (-5.3,-3.5) {$c_3$};
\node () at (-5.3,-4.5) {$c_4$};

\node () at (-1.5,-6.6) {case 1};

\node[vertex] (11) at (-1,-1) {};
\node[vertex,below of=11,node distance=1cm](12){};
\node[vertex,below of=12,node distance=1cm](13){};
\node[vertex,below of=13,node distance=1cm](14){};
\node[vertex,below of=14,node distance=1cm](15){};
\node[vertex] (1c1) at (-2,-0.5) {};
\node[vertex] (1c2) at (-2,-1.5) {};
\node[vertex] (1c3) at (-2,-4.5) {};
\node[vertex] (1c4) at (-2,-5.5) {};

\draw (1c1)--(11)--(1c2);
\draw (1c3)--(15)--(1c4);
\draw (12)--(13)--(14);
\draw (11) edge[ultra thick] (12);
\draw (14) edge[ultra thick] (15);

\node () at (-0.7,-1) {$t_i$};
\node () at (-0.7,-2) {$l_i$};
\node () at (-0.65,-3) {$m_i$};
\node () at (-0.7,-4) {$r_i$};
\node () at (-0.7,-5) {$f_i$};
\node () at (-2.3,-0.5) {$c_1$};
\node () at (-2.3,-1.5) {$c_2$};
\node () at (-2.3,-4.5) {$c_3$};
\node () at (-2.3,-5.5) {$c_4$};
\node () at (1.5,-6.6) {case 2};

\node[vertex] (21) at (2,0) {};
\node[vertex,below of=21,node distance=0.6cm](22){};
\node[vertex,below of=22,node distance=0.6cm](23){};
\node[vertex,below of=23,node distance=0.6cm](24){};
\node[vertex,below of=24,node distance=0.6cm](25){};

\node[vertex,below of=25,node distance=1 cm](26){};
\node[vertex,below of=26,node distance=0.6cm](27){};
\node[vertex,below of=27,node distance=0.6cm](28){};
\node[vertex,below of=28,node distance=0.6cm](29){};
\node[vertex,below of=29,node distance=0.6cm](210){};

\node[vertex] (212) at (3,-2.9) {};
\node[vertex] (211) at (4,-2.9) {};

\node[vertex,left of=21,node distance=1 cm](2c1){};
\node[vertex,left of=25,node distance=1 cm](2c2){};
\node[vertex,left of=26,node distance=1 cm](2c3){};
\node[vertex,left of=210,node distance=1 cm](2c4){};

\node () at (2.3,0) {$t_i^{1}$};
\node () at (1.7,-0.6) {$l_i^{1}$};
\node () at (1.65,-1.2) {$m_i^{1}$};
\node () at (1.7,-1.8) {$r_i^{1}$};
\node () at (2.25,-2.5) {$f_i^{1}$};
\node () at (2.25,-3.4) {$t_i^{2}$};
\node () at (1.7,-4) {$l_i^{2}$};
\node () at (1.65,-4.6) {$m_i^{2}$};
\node () at (1.7,-5.2) {$r_i^{2}$};
\node () at (2.3,-5.8) {$t_i^{2}$};
\node () at (3.3,-2.9) {$p_i^{1}$};
\node () at (4.3,-2.9) {$p_i^{2}$};
\node () at (0.7,-0) {$c_1$};
\node () at (0.7,-2.4) {$c_2$};
\node () at (0.7,-3.4) {$c_3$};
\node () at (0.7,-5.8) {$c_4$};
\node () at (7.5,-6.6) {case 3};

\draw (21) edge[ultra thick] (22);
\draw (22)--(23)--(24);
\draw (24) edge[ultra thick] (25);
\draw (26) edge[ultra thick] (27);
\draw (27)--(28)--(29);
\draw (29) edge[ultra thick] (210);
\draw (2c1)--(21);
\draw (2c2)--(25);
\draw (2c3)--(26);
\draw (2c4)--(210);
\draw (22)--(211)--(29);
\draw (24)--(212)--(27);


\node[vertex] (31) at (7,-2) {};
\node[vertex] (32) at (7,-3) {};

\node[vertex] (33) at (7.6,-2.5) {};
\node[vertex] (34) at (7.6,-3.5) {};

\node[vertex] (35) at (8.2,-3) {};
\node[vertex] (36) at (8.2,-4) {};

\node[vertex] (37) at (8.8,-4.5) {};
\node[vertex] (38) at (7,-4.5) {};

\node[vertex] (3c1) at (6,-1.5) {};
\node[vertex] (3c2) at (6,-2.5) {};
\node[vertex, left of=32, node distance=1 cm] (3c3)  {};
\node[vertex, left of=38, node distance=1cm] (3c4)  {};

\node () at (7,-1.65) {$t_i^{1}$};
\node () at (7,-2.65) {$t_i^{2}$};
\node () at (7.6,-2.2) {$l_i^{1}$};
\node () at (7.6,-3.2) {$l_i^{2}$};
\node () at (8.25,-2.7) {$m_i^{1}$};
\node () at (8.2,-3.7) {$m_i^{2}$};
\node () at (8.8,-4.8) {$r_i$};
\node () at (7,-4.8) {$f_i$};

\node () at (5.7,-1.5) {$c_1$};
\node () at (5.7,-2.5) {$c_2$};
\node () at (5.7,-3) {$c_3$};
\node () at (5.7,-4.6) {$c_4$};

\draw (31) edge [ultra thick](33);
\draw (32) edge [ultra thick](34);
\draw (37) edge [ultra thick](38);
\draw (33)--(35)--(37)--(36)--(34);
\draw (3c1)--(31)--(3c2);
\draw (3c3)--(32);
\draw (3c4)--(38);
\end{tikzpicture}
\caption{Construction for \textsc{Ext EC} (planar).
On the left: A variable $x_i$ appearing in four clauses $c_1,c_2,c_3,c_4$ in $I$. On the right, cases 1, 2, 3: The gadgets $H(x_i)$ in the constructed instance depend on how $x_i$ appears (negated or non-negated) in the four clauses. Bold edges denote elements of~$U$.}\label{Fig:Ext planar EC}
\end{figure}

\section{Parameterized Perspective} \label{fptsection}
For notions undefined in this paper, we refer to the  textbook \cite{DowFel2013}. Generally, our model for extension allows for problems which are not in $\np$, which is due to the difficulty of deciding minimality, i.e., checking if $S\in\mu(\mathop{sol}(I))$ for $S\in \mathop{sol}(I)$. For the specific problems we discuss here however, minimality can obviously always be tested efficiently. From our parameterized perspective, this immediately yields membership in $\fpt$ for all cases where the set of  $S\in \mathop{presol}(I)$ with $U\leq_I S$ can be enumerated in a function in the parameter. \par
As $U\leq_I S$ means  $U\supseteq S$ for $\textsc{Ext EM}$, it follows that we only have to consider the $2^{|U|}=2^{m(U)}$ subsets of $U$ as candidates for a minimal extension of $U$, which yields:
\begin{corollary}
$\textsc{Ext EM}$ with standard parameter is in~$\fpt$. 
\end{corollary}
As we can list all supersets of a given set $U\subseteq X$ in time $O(2^{|X|-|U|})$, and similarly all partitions of $X$ refining a given partition $\pi_U$ in time $O(c^{|X|-|\pi_U|})$, we also get the following (observe that for a graph instance $I=(V,E)$ we have $m_{max}(I)=|V|$ for $\textsc{Ext EC}$ and $\textsc{Ext DS}$ and $m_{max}(I)=|E|$ for $\textsc{Ext EDS}$ and that $P=\{\{x\}\colon x\in X\}$ gives the partition of $X$ of value $|X|$ which gives the maximum for dual parametrization of  $\textsc{Ext BP}$):
\begin{corollary}
$\textsc{Ext EC}$, $\textsc{Ext DS}$, $\textsc{Ext EDS}$ and  $\textsc{Ext BP}$ with dual parameter are in~$\fpt$. 
\end{corollary}
In the following we derive a less obvious $\fpt$ membership result which is based on  enumeration of minimal vertex covers; see~\SV{\cite{Dam2006}}\LV{\cite{Dam2006,Fer02a}}. We discuss \textsc{Ext EM} where it is sometimes more convenient to think about this problem as follows: Given a graph $G=(V,E)$ and an edge set $A$, does there exist an inclusion-wise maximal matching $M$ (given as a set of edges) of $G$ that avoids $A$, i.e., $M\cap A=\emptyset$.
With dual parameterization the parameter then is $|A|$.
Assume there is a  maximal matching $M$ of $G$ such
that $M\cap A=\emptyset$, then the next property is quite immediate but very helpful\LV{; also see~\cite{Fer02a}}.
\begin{lemma}\label{lem: VC and EM}
$V(M)\cap V(A)$ is a vertex cover of $G''=(V,A)$.
\end{lemma}
%
In order to use this observation, we need the following construction to compute matchings according to a fixed vertex cover. For a minimal vertex cover~$S$ of $G''=(V,A)$, let $(G',w^S)$ be the edge-weighted graph defined as follows: For $v\in V$, $d^S(v)=1$ if $v\notin S$, and  $d^S(v)=|E|+1$ if $v\in S$. We define $w^S$, the edge weight, by: $w^S(e)=d^S(x)+d^S(y)$ for $e=xy\in E\setminus A$. This way, we link the profit of an edge in a weighted matching for  $(G',w^S)$ to how much it covers of the vertex cover $S$ of  $G''=(V,A)$, which formally yields:
\begin{theorem}\label{theo: VC and EM}
There is a maximal matching $M$ of $G$ such that $M\cap A=\emptyset$ if and only if there is a
minimal vertex cover $S$ of $G''$ such that the maximum weighted matching of $(G',w^S)$ is at least $|S|(|E|+1)$.
\end{theorem}

\newcommand{\pfVCandEM}{\begin{proof} 
Let $G=(V,E)$ be a graph and let $A\subseteq E$. 
Let $M$ be any maximal matching of~$G$ such that $M\cap A=\emptyset$ (if any). From Lemma \ref{lem: VC and EM}, we know
$V(M)\cap V(A)$ is a vertex cover of~$G''$. Let $S\subseteq V(M)\cap V(A)$ be any minimal vertex cover of $G''$. By construction of $w^S$, we have $w^S(M)\geq |S|(|E|+1)$.

Conversely, assume that $S$ is a minimal vertex cover of $G''$ such that the maximum weighted matching of $(G',w^S)$ is at least $|S|(|E|+1)$. Let $M$ be any maximum weighted matching of $(G',w^S)$. By construction of $w^S$, matching $M$ is
incident to every vertex of $S$. In conclusion, $M$ is a maximal matching of $G$ with $M\cap A=\emptyset$.
\end{proof}}

\LV{\pfVCandEM}

\noindent
Using the characterization given in Theorem \ref{theo: VC and EM}, the following result follows simply from enumerating all minimal vertex covers of $G''$.

\begin{theorem}\label{theo:Neg Ext EM}
\textsc{Ext EM} with dual
parameter is in $\fpt$. 
\end{theorem}

\newcommand{\pfNegExtEM}{\begin{proof} Consider
an algorithm that lists all minimal vertex covers of $G''$ and checks the matching condition of $(G',w^S)$ according to  Theorem \ref{theo: VC and EM} in polynomial time.
The running time is dominated by the procedure that lists all minimal vertex covers. As the number of edges in a graph is an upper bound on any minimal vertex cover of that graph, it is clear that we can enumerate all minimal vertex covers of $G''$ in time $O^*(2^{|A|})$ by \LV{\cite{Dam2006,Fer02a,Fer06b}}\SV{\cite{Dam2006,Fer06b}}.
 \end{proof}}

{\sc Minimum Hitting Set} as $\npo$ problem is defined by instances $I=(X,\mathcal{S})$ where $X$ is a finite ground set and $\mathcal{S}=\{S_1,\dots, S_m\}$ is a collection of sets $S_i \subseteq X$ (usually referred to as \emph{hyperedges}) and feasible solutions are subsets $H\subseteq X$ such that $H\cap S_i\not=\emptyset$ for all $i\in\{1,\dots,m\}$. In \cite{addarxivcitation} the associated extension problem, where pre-solutions are all subsets $U\subseteq X$, the partial ordering is set-inclusion, in the following referred to as \textsc{Ext HS}, appears as a subproblem for the enumeration of minimal hitting sets in lexicographical order and \textsc{Ext HS} is there shown to be $\wthree$-complete with respect to the standard parameter $m(I,U)=|U|$. By a slight adjustment of the classical reduction from the hitting set problem to {\sc Dominating Set}, this result transfers and formally yields:
\begin{theorem}\label{thm:ExtDSparam}
\textsc{Ext DS} with standard parameter is $\wthree$-complete, even when restricted to bipartite instances.
\end{theorem}

\newcommand{\pfDSWparam}{
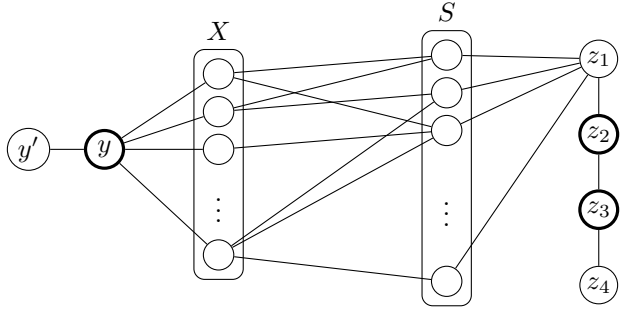
\begin{figure}
\tikzstyle{vertex1}=[circle, draw, inner sep=0pt, minimum width=4pt, minimum size=0.5cm]
\centering
\begin{tikzpicture}[scale=1, transform shape]
\node[vertex] (x1) at (0,0) {};
\node[vertex, below of=x1, node distance=0.5 cm] (x2) {};
\node[vertex, below of=x2, node distance=0.5 cm] (x3) {};
\node[below of=x3, node distance=0.7 cm] (vdots) {$\vdots$};
\node[vertex, below of=vdots, node distance=0.7 cm] (xn) {};
\node[draw, rectangle,rounded corners,label=above:$X$,fit= (x1) (xn)] (){};

\node[vertex] (s1) at (3,0.25) {};
\node[vertex, below of=s1, node distance=0.5 cm] (s2) {};
\node[vertex, below of=s2, node distance=0.5 cm] (s3) {};
\node[below of=s3, node distance=1 cm] (v1dots) {$\vdots$};
\node[vertex, below of=v1dots, node distance=1 cm] (sm) {};
\node[draw, rectangle,rounded corners,label=above:$S$,fit= (s1) (sm)] (){};

\draw (x2)--(s1)--(x1);
\draw (x2)--(s2)--(xn);
\draw (x3)--(s3)--(xn);
\draw (s3)--(x1);
\draw (xn)--(sm);

\node[vertex1, very thick] (y) at (-1.5,-1) {$y$};
\node[vertex] (yp) at (-2.5,-1) {$y'$};
\draw (y)--(yp);

\node[vertex] (z1) at (5,0.2) {$z_1$};
\node[vertex, very thick] (z2) at (5,-0.8) {$z_2$};
\node[vertex, very thick] (z3) at (5,-1.8) {$z_3$};
\node[vertex] (z4) at (5,-2.8) {$z_4$};
\draw (z1)--(z2)--(z3)--(z4);
\draw (x1)--(y)--(x2);
\draw (x3)--(y)--(xn);
\draw (s1)--(z1)--(s2);
\draw (s3)--(z1)--(sm);
\end{tikzpicture}
\caption{The graph $G=(V,E)$ for \textsc{Ext DS}, Vertices in $U'$ are drawn bold.}\label{Fig:Par DS}
\end{figure}

\begin{proof} 
\textsc{Ext DS} can obviously be modeled as special case of the extension for hitting sets by interpreting the closed neighborhoods as subsets of the ground set of vertices. This immediately gives membership in $\wthree$ for \textsc{Ext DS}. \par
Conversely, given an instance $(I,U)$ with $I=(X,\mathcal{S})$, $\mathcal{S}=\{S_1,\dots, S_m\}$ for \textsc{Ext HS} we create a graph for the corresponding instance for \textsc{Ext DS} as follows:   \begin{itemize}
\item Start with the bipartite graph on vertices $X\cup \{s_1,\dots,s_m\}$ containing edges $xs_i$ iff $x\in S_i$.
\item Add two new vertices $y,y'$ with edges $y'y$ and $xy$ for all $x\in X$.
\item Add four new vertices $z_1,z_2,z_3,z_4$ with edges $z_1z_2,z_2z_3,z_3z_4$ and $z_1s_i$ for all $1\leq i\leq m$.
\end{itemize}
The construction is illustrated in Fig~\ref{Fig:Par DS}.
Let $G=(V,E)$ denote the obviously still bipartite graph created in this way. With the set $U'$ containing the vertex $y$ to dominate $X$, $z_2$ and $z_3$ to forbid including any vertex $s_i$ in the extension (as this would make $z_2$ obsolete) and the vertices corresponding to the pre-solution $U$ for \textsc{Ext HS}, it is not hard to see that $(G,U')$ is a \yes-instance for \textsc{Ext DS} iff $(I,U)$ is a \yes-instance for \textsc{Ext HS}. As the parameters relate by $m(G,U')=|U'|=|U|+3=m(I,U)+3$, this reduction transfers the $\wthree$-hardness of \textsc{Ext HS} to \textsc{Ext DS} on bipartite graphs.
\end{proof}}

To show that \textsc{Ext EDS} is not likely to be \LV{fixed parameter tractable}\SV{$\fpt$} by standard parameterization, we designed a reduction from \textsc{Ext VC} which is shown to be $\wone$-complete in \cite{CasFKMS2018a}.

\begin{theorem}\label{theoParameterized:Ext EDS}
\textsc{Ext EDS} with standard parameter is $\wone$-hard, even when restricted to bipartite instances.
\end{theorem}
Since edge domination can be seen as vertex domination on line graphs, the $\wthree$-membership of \textsc{Ext DS} transfers to \textsc{Ext EDS}. It remains open where \textsc{Ext EDS} exactly belongs in the $\warbitrary$ hierarchy, as we were not able to place it in $\wtwo$ or even $\wone$ or to show  $\wthree$-hardness.

\paragraph*{Discussing treewidth}
The crucial combinatorial properties of many classical decision problems can be expressed in monadic second order logic (MSO), or counting variants thereof, the actual characterization of the property can be often even written in FO. 
This is well known for the problems considered in this paper. For instance, 
$\text{Cov}(C):=\forall x,y \left(\text{adj}(x,y)\rightarrow (x\in C\lor y\in C)\right) $
says that $C$ is a vertex cover. 
Based on such a formula,  it is also possible to express inclusion-wise minimality or maximality (the most important cases for our paper). 
For instance,
$$\mu\text{-Cov}(C):=\forall x,y \left(\text{adj}(x,y)\rightarrow (x\in C\lor y\in C)\right)\land \forall x\in C\exists y \left(\text{adj}(x,y)\land y\notin C \right) $$
says that $C$ is a minimal vertex cover. 
The fact that $C$ is superset of a given set $U$ can be easily expressed by an implication. Also see~\cite[Sec.~7.4.2]{CygFKLMPPS2015}\LV{, based on~\cite{ArnLagSee91,Cou90,Cou92a}}.
Hence, we \SV{state}\LV{can summarize}:
\begin{proposition}
All extension graph problems considered in this paper can be written as MSO formulae. Hence, applying Courcelle's theorem, these decision problems, parameterized by   treewidth, are in $\fpt$.
\end{proposition}

We can also state the following result which does not use the previous meta-theorem:

\begin{proposition}\label{prop-tw-specific}
For each of the graph extension problems $P$  discussed in this paper, there is a constant $c_P$ such that $P$ can be solved in time  $\Oh^*(c_P^t)$ on graphs of treewidth at most $t$.
\end{proposition}

\section{Optimal Exact Algorithms under the Exponential Time Hypothesis}

Due to Proposition~\ref{prop-tw-specific},
we can also see that all our extension graph problems can be solved in time $\Oh(c_P^n)$ for some problem-$P$ specific constant $c_P$.
As planar graphs of order $n$ have treewidth $\Oh(\sqrt{n})$, 
Proposition~\ref{prop-tw-specific}
yields the following result\LV{; also compare to~\cite{Mar2013}}.
\begin{corollary}
Each of our extension graph problems  can be solved in time  $\Oh(2^{\Oh(\sqrt{n})})$ on planar graphs of order $n$.
\end{corollary}

For \textsc{Ext BP}, the situation is a bit more complicated. The validity of the dynamic programming algorithm that we designed depends on a somewhat special combinatorial characterization of minimal solutions. Namely, 
a partition solution $\pi$ is  minimal if and only if there is a constant $0<\delta< 1/2$
such that for the two bins (sets) $X_1,X_2\in\pi$ of smallest weight, 
with $w(X_1)\leq w(X_2)$, 
 $w(X_1)>\delta$ and $w(X_2)\geq 1-\delta$.
\begin{theorem}\label{thm:ExtBPDP}
\textsc{Ext BP} with $n$ items can be solved in time  $O^*(6^{n})$, using space $O^*(3^{n})$.
\end{theorem}

The Exponential Time Hypothesis (ETH) was introduced to provide evidence for lower bounds on exponential-time algorithms; see \cite{ImpPatZan2001,JonLNZ2017,LokMarSau2011b}. 
Now, if we look at the construction proving $\np$-hardness of $(3,B2)$-{\sc SAT}~\cite[Theorem~1]{ECCC-TR03-049}, it is clear that these transformations blow up the instance by a constant factor only. 
Hence, we can state:

\begin{lemma} Assuming ETH, there is no $2^{o(n+m)}$-algorithm for solving $n$-variable, $m$-clause instances of 
$(3,B2)$-{\sc SAT}.
\end{lemma}

As all our reductions, starting out from $(3,B2)$-{\sc SAT}, are blowing up the size of the instances only in a linear fashion, we can immediately conclude:

\begin{theorem}\label{thm:eth-summary}
Let $\mathcal{P}\in\{\textsc{EC},\textsc{EM},\textsc{DS},\textsc{EDS}\}$. There is no $2^{o(n+m)}$-algorithm for  $n$-vertex, $m$-edge bipartite instances of maximum degree~3 of $\textsc{Ext }\mathcal{P}$, unless ETH fails.
\end{theorem}

For planar graphs, the situation is a bit more involved. Kratochv\'il's construction for showing $\np$-hardness of \textsc{4P3C3SAT} is based on Lichtensteins's \cite{Lichtenstein82} and offers a quadratic blow-up in total, compared to \textsc{3SAT} itself.

\begin{proposition}\label{prop:SAT}
There is no algorithm that solves 
\textsc{4P3C3SAT} on instances with $n$ variables and $m$ clauses in time $2^{o(\sqrt{n+m})}$, unless ETH fails.
\end{proposition}

\begin{corollary}\label{cor:ETH-planar}
For any  graph extension problem studied in this paper, there is no $2^{o(\sqrt{n})}$ algorithm for solving planar instances of order~$n$, unless ETH fails.
\end{corollary}

\textsc{Ext BP} delivers a nice example that it is not always that easy to take the textbook construction from Garey and Johnsson for
$\np$-hardness to immediately get hardness results under ETH that match existing algorithms. 
According to~\cite[Corollary 3.2]{JanLanLan2016}, there is no algorithm deciding \textsc{4-Partition} with $n$ items in time $2^{o(n)} \times|I|^{O(1)}$, unless  ETH fails.
As we can adapt our reduction from Theorem~\ref{thm:Ext BP np-hard} to reduce from this problem, we can formulate:

\begin{corollary}\label{cor:eth-ExtBP}
Assuming ETH, 
\textsc{Ext\nolinebreak~BP} with $n$ items cannot be solved in time  $2^{o(n)}$.
\end{corollary}

In conclusion, all exact algorithms  we considered are optimal under ETH.


\section{Generalizations of Matching and Edge Cover Problems}
In this section, we want to consider more general versions of matchings and edge covers of a graph.
A partial subgraph $G'=(V,S)$ of $G=(V,E)$ is called \emph{$r$-degree constrained} if the maximum degree of $G'$ is upper-bounded by $r$. The case $r=1$ corresponds to a matching.
A maximum $r$-degree constrained partial subgraph can be found in polynomial time~\cite{Gabow83}. 
Here, we are interested in the corresponding extension variant, called \textsc{Ext $r$-DCPS}.

\begin{theorem}\label{thm: Ext_r-DS}

For every fixed $r\geq 2$, \textsc{Ext $r$-DCPS} is $\np$-complete in bipartite graphs with maximum degree $r+1$, even if the set of forbidden edges induces a matching.

\end{theorem}
On the positive side, it is also possible to generalize the $\fpt$-result from Section~\ref{fptsection}:

\begin{theorem}\label{theo:Neg Ext r-DS}
\textsc{Ext $r$-DCPS} with dual
parameter is in $\fpt$.
\end{theorem}

This generalization  of Theorem~\ref{theo:Neg Ext EM} is not at all trivial, it combines the previous idea of listing minimal vertex covers with solving the weighted version of  \textsc{Max $r$-DCPS} on an auxiliary graph.

Similarly, we can generalize the notion of edge cover. A partial subgraph $G'$ of a given graph $G$ is called \emph{$r$-edge cover} if the minimum degree of $G'$ is lower-bounded by~$r$. A problem called \textsc{Min Lower-Upper-Cover Problem}, or \textsc{MinLUCP} for short, generalizes the problem of finding an $r$-edge cover of minimum size and can be solved in polynomial time~\cite{AS03}. The optimization variant of \textsc{MinLUCP} is: given $G=(V,E)$ and two non-negative functions $a,b$ from $V$ such that $\forall v\in V$, $0\leq a(v)\leq b(v)\leq d_G(v)$,  find a subset $M\subseteq E$ such that the partial graph $G[M]=(V,M)$ induced by $M$ satisfies $a(v)\leq d_{G[M]}(v)\leq b(v)$ (such a solution is called a {\em lower-upper-cover}), minimizing its total size $|M|$ among all such solutions (if any). An $r$-EC solution
corresponds to a  lower-upper-cover with $a(v)=r$ and $b(v)=d_G(v)$ for every $v\in V$. Our name for the extension variant thereof
 is \textsc{Ext $r$-EC}.

\begin{theorem}\label{thm: Ext_r-LEC}
For every fixed $r\geq 1$, \textsc{Ext $r$-EC} is $\np$-complete in bipartite graphs with maximum degree $r+2$, even if the pre-solution is an induced matching.
\end{theorem}

\noindent
We exploit a combinatorial relationship between \textsc{Ext $r$-EC} and \textsc{Min\nolinebreak LUCP} to show:

\begin{theorem}\label{theo:ext r-LEC FPT}
\textsc{Ext $r$-EC} with standard parameter is in $\fpt$.
\end{theorem}

\section{Conclusions}
In this paper, we introduced a general framework to model extension for monotone problems with the attempt to highlight the unified strucutre of such types of problems that seem to appear in many different scenarios. Admittedly, our framework does not cover all problems of this flavour. Quite similar problems have for example been considered in the area of graph coloring, under the name of \emph{pre-coloring extension}, which contains the completion of partial Latin squares as a special case~\cite{BirHujTuz92,Col84,Mar2005a}. However, there is a crucial difference with our approach: While with our problems, the minimality condition on the permissible extensions is essential for all our considerations, they become at best uninteresting for pre-coloring extension problems, although it is pretty straightforward to define partial orderings on pre-colorings so that the set of proper colorings is upward closed as required in our setting.
It would be interesting to study such forms of extension problems also in a wider framework.

Extension problems as introduced in this paper are also different from \emph{refinement problems} as studied in \cite{BodlaenderDFH09}, which ask if a given solution is optimum. They also differ from \emph{incremental problems} where it is asked if a given solution provides sufficient information to obtain a good solution to an instance obtained after modifying the instance in a described way, which is also very much related to reoptimization; see \cite{BorPas2010,ManMat2017}.\LV{Another related class of problems is formed by reconfiguration problems; see \cite{Mou2015}.\todo{HF: I am still not sure what type of problems we really want to discuss, I added some more not to forget them, but finally we might want to omit this discussion from the paper ...}}

We only focused on few specific problems. In view of the richness of combinatorial problems, many other areas could be looked into with this new approach. Further, it would be interesting to investigate to what extend enumeration problems can be improved by a clever solution to extension or, conversely, how the difficulty of extension implies bounds on enumeration problems. Finally, let us give one concrete open question in the spirit of the mentioned letter of Gödel to von Neumann: Is it possible to design an exact algorithm for \textsc{Upper Domination} that avoids enumerating all minimal dominating sets? This still unsolved question already triggered quite some recent research; see \cite{BazganBCF16,BazBCFJKLLMP2018}.

\paragraph*{Acknowledgements}
Partially supported by the Deutsche Forschungsgemeinschaft (FE 560/6-1) and  by the project “ESIGMA” (ANR-17-CE23-0010). 

\bibliographystyle{abbrv} %
\bibliography{biblio}

\newpage
\section{Appendix: Omitted proofs}

\subsection{
Proof of Theorem~\ref{thm:np-completeness-summary} for \texorpdfstring{$\mathcal{P}=\textsc{EM}$}{P=EM}}


\allofExtEM

\subsection{Proof of Theorem~~\ref{thm:np-completeness-summary} for \texorpdfstring{$\mathcal{P}=\textsc{DS}$}{P=DS}}

\allofExtDS

\subsection{Proof of Theorem~~\ref{thm:np-completeness-summary} for \texorpdfstring{$\mathcal{P}=\textsc{EDS}$}{P=EDS}}

\allofExtEDS
\pfExtEDS

\subsection{Proof of Theorem~\ref{thm:Ext BP np-hard}}
\begin{theorem*}
\textsc{Ext BP} is $\np$-hard, even if the  given  pre-solution $\pi_U$ contains only two sets.
\end{theorem*}
\pfBPnp

\subsection{Proofs for Theorem~\ref{thm:np-completeness-planar-summary}}
These proves are summarized below in Section~\ref{sec-planar}.

\subsection{Proof of Theorem~\ref{theo: VC and EM}}

\begin{theorem*}
There is a maximal matching $M$ of $G$ such that $M\cap A=\emptyset$ if and only if there is a
minimal vertex cover $S$ of $G''$ such that the maximum weighted matching of $(G',w^S)$ is at least $|S|(|E|+1)$.
\end{theorem*}
\pfVCandEM

\subsection{Proof of Theorem~\ref{theo:Neg Ext EM}}

\begin{theorem*}
\textsc{Ext EM} with dual
parameter is in $\fpt$. 
\end{theorem*}
\pfNegExtEM



\subsection{Proof of Theorem~\ref{thm:ExtDSparam}}
\begin{theorem*}
\textsc{Ext DS} with standard parameter is $\wthree$-complete, even when restricted to bipartite instances.
\end{theorem*}
\pfDSWparam

\subsection{Proof of Theorem~\ref{theoParameterized:Ext EDS}}
\begin{theorem*}
{\em \textsc{Ext EDS} (with standard parameter) is $W[1]$-hard, even when restricted to bipartite graphs.}
\end{theorem*}

\begin{proof}
The hardness result comes from a reduction from  \textsc{Ext VC} on bipartite graphs. Let $I = (G,U)$ be an instance of \textsc{Ext VC}, where $G=(V,E)$ is a bipartite graph with partition $(V_1,V_2)$ of $V$ and $U\subseteq V$. We build an instance $I'=(G',U')$ of \textsc{Ext EDS} as follows. Let us first construct a new graph $G^{\prime}=(V^{\prime},E^{\prime})$ with $V^{\prime}=V\cup\{x_i,y_i,z_i\suchthat i=1,2\}$ and $$E^{\prime}=E\cup\bigcup_{i=1,2} \big(\{x_iy_i,y_iz_i\}\cup \{vx_i \suchthat v\in V_i\}\big)$$ by adding six new vertices (three for each part). $G^{\prime}$ is obviously bipartite with partition into $V^{\prime}_1=V_1\cup \{x_2,y_1,z_2\}$ and $V^{\prime}_2=V_2\cup \{x_1,y_2,z_1\}$. Let $$U^{\prime}=\big(\{ux_1\suchthat u\in U\cap V_1\}\cup \{ux_2\suchthat u\in U\cap V_2\}\big)\cup \{x_1y_1,x_2y_2\}\,;$$ so, $|U^{\prime}|=|U|+2$. This construction is illustrated in Fig.~\ref{Fig:Ext EDS}. We claim that $(G^{\prime},U^{\prime})$ is a \yes-instance of \textsc{Ext EDS} if and only if $(G,U)$ is a \yes-instance of \textsc{Ext VC}.
\bigskip

\noindent
Suppose $(G, U)$ is a \yes-instance for \textsc{Ext VC}; so there exists a minimal vertex cover $S$ for $G$ such that $U\subseteq S$. Consider the set $S^{\prime}=\{vx_1\suchthat v\in V_1\cap S\}\cup \{vx_2\suchthat v\in V_2\cap S\}\cup \{x_1y_1,x_2y_2\}$. $S^{\prime}$ is an edge dominating set of $G^{\prime}$ which includes $U^{\prime}$ because $S$ contains $U$. Since $S$ is minimal,  $S^{\prime}$ is minimal, too; observe that private edges of a vertex $v\in S\cap V_1$ translate to private edges of $vx_1\in S^{\prime}$, analogously for $x\in S\cap V_2$. By construction, $y_iz_i$ is a private edge for $x_iy_i$, $i=1,2$.

\begin{figure}
\centering
\begin{tikzpicture}[]
\node[vertex] (v1) at (0,0) {$v_1$};
\node[vertex,ultra thick, right of=v1,node distance=1cm](v2){$v_2$};
\node[vertex,below of=v2,node distance=1cm](v3){$v_3$};
\node[vertex,below of=v1,node distance=1cm](v4){$v_4$};
\node[vertex,left of=v1,node distance=1cm](v5){$v_5$};
\node[vertex,right of=v2,node distance=1cm](v6){$v_6$};

\draw (v1)--(v2)--(v3)--(v4)--(v1)--(v5);
\draw (v2)--(v6);

\node[vertex, right of=v6,node distance=2cm](nv5){$v_5$};
\node[vertex, right of=nv5,node distance=1cm](nv1){$v_1$};
\node[vertex, right of=nv1,node distance=1cm](nv2){$v_2$};
\node[vertex, right of=nv2,node distance=1cm](nv6){$v_6$};
\node[vertex, below of=nv2,node distance=1cm](nv3){$v_3$};
\node[vertex, below of=nv1,node distance=1cm](nv4){$v_4$};
\node[vertex] (x1) at (5.4,1.4) {$x_1$};
\node[vertex] (y1) at (6.3,1.4) {$y_1$};
\node[vertex] (z1) at (7.2,1.4) {$z_1$};

\node[vertex] (x2) at (5.4,-2.4) {$x_2$};
\node[vertex] (y2) at (6.3,-2.4) {$y_2$};
\node[vertex] (z2) at (7.2,-2.4) {$z_2$};

\draw (nv1)--(nv2)--(nv3)--(nv4)--(nv1)--(nv5);
\draw[ultra thick](x2)--(y2);
\draw(y2)--(z2);
\draw(nv2)--(nv6);
\draw[ultra thick](x1)--(y1);
\draw(y1)--(z1);
\draw (nv6) edge[bend left=15] (x2);
\draw(nv5)--(x1);
\draw[ultra thick](nv2)--(x1);
\draw (nv1) edge[bend right=45] (x2);
\draw(nv3)--(x2);
\draw (nv4) edge[bend right=25] (x1);

\end{tikzpicture}
\caption{$(G,U)$ as an instance of \textsc{Ext VC} is shown on the left, with $V_1=\{v_2,v_4,v_5\}$ and $V_2=\{v_1,v_3,v_6\}$ and $U=\{v_2\}$. The constructed instance $(G^{\prime},U^{\prime})$ of \textsc{Ext EDS} is shown on the right. The vertices and edges of $U$ and $U^{\prime}$ are in marked with bold lines.}\label{Fig:Ext EDS}
\end{figure}
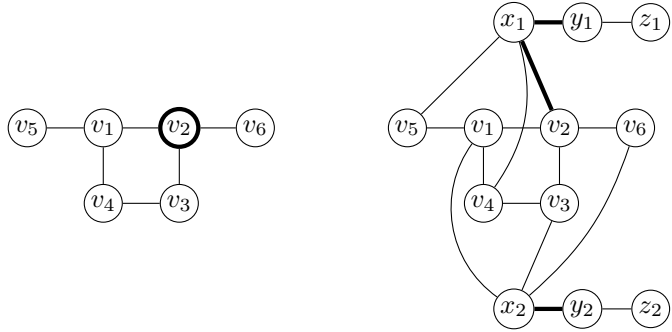

\bigskip
\noindent
Conversely, suppose $S^{\prime}$ is a minimal edge dominating set of $G^{\prime}$ containing $U^{\prime}$.
Since $S^{\prime}$ is minimal, then for each $e\in S^{\prime}$ there is a private edge set $S_{e}\subseteq E^{\prime}$, $S_e\neq\emptyset$, which is dominated only by $e$.
Moreover, we have, for $i\in\{1,2\}$: 
$$\forall v\in V_i\left((vx_i\in S^{\prime})\iff (\forall u\in V_{3-i}(vu\notin S'\cap E)\right)$$
since $S^{\prime}$ is minimal and $\{x_1y_1,x_2y_2\}\subseteq U^{\prime}$. We will now show how to safely modify $S^{\prime}$ such that $S^{\prime}\cap E = \emptyset$. If it is not already the case, there is some edge, w.l.o.g., $e=uv\in S^{\prime}\cap E$ with $u\in V_1$ and $v\in V_2$. In particular from the above observations, we deduce $u\notin U$, $v\notin U$ and $S_{e}\subseteq E$.
Modify $S'$ by the following procedure.

\begin{itemize}

\item[$\bullet$] If the private solution set $S_{e}\setminus \{e\}$ contains some edges incident to $u$ and some edges incident to $v$, then $e\in S^{\prime}$ will be replaced by $ux_1$ and $vx_2$;

\item[$\bullet$]  if every edge in the private solution $S_e$ is adjacent to $u$, replace $e$ in $S^{\prime}$ by $ux_1$, otherwise if every edge in the private solution $S_e$ is adjacent to $v$, replace $e$ in $S^{\prime}$ by $vx_2$.
\end{itemize}

The case distinction is necessary to guarantee that $S'$ stays a 
minimal edge dominating set after each modification step. 
We repeat this procedure until $S^{\prime}\cap E=\emptyset$. At the end of the process, any vertex $v\in V$ covers the same set of edges as $vx_1$ or $vx_2$ dominates. 
Hence, by setting $S=\{v\in V\suchthat vx_1\in S^{\prime}\text{ or } vx_2 \in S^{\prime}\}$, we build a minimal vertex cover of $G$ containing $U$.
\end{proof}

\begin{rmk} \label{rem:theoParameterized:Ext EDS}
Note that the procedure of local modifications given in Theorem \ref{theoParameterized:Ext EDS} does not preserve optimality, but only inclusion-wise minimality.
\end{rmk}

\section{Generalization of matching extension}\label{Gene:Ext matching}

\textsc{$r$-Degree Constrained Partial Subgraph} (or \textsc{$r$-DCPS} for short) is defined as follows: the instance $I=(V,E)$ is a graph, $\leq_I=\supseteq$,
  $\mathop{presol}(I)=2^E$, $S\in \mathop{sol}(I)$ iff
  none of the vertices in $V$ is incident to more than $r$ edges in $S$.
\bigskip

\noindent The particular case of $r=1$ corresponds to the famous matching problem in graphs. The optimization problem associated to \textsc{$r$-DCPS},  denoted here \textsc{Max $r$-DCPS}, consists of finding an edge subset $E'$ of maximum cardinality that is a solution to \textsc{$r$-DCPS}. \textsc{Max $r$-DCPS} is known to be solvable in polynomial time even for the edge weighted version (here, we want to maximize $w(E')$)~\cite{Gabow83}. When additionally the constraint $r$ is not uniform and depends on each vertex (i.e., at most $b(v)=r_v$ edges incident to vertex $v$), \textsc{Max $r$-DCPS} is usually known as \textsc{Simple $b$-Matching}
and remains solvable in polynomial time  even for the edge weighted version (Theorem 33.4, Chapter 33 of Volume A in \cite{AS03}). 

We are considering the associated extension problem, formally described as follows.

\begin{center}
\fbox{\begin{minipage}{.95\textwidth}
\noindent{\textsc{Ext $r$-DCPS}}\\\nopagebreak
{\bf Input:} A graph $G=(V,E)$ and $U\subseteq E$. \\\nopagebreak
{\bf Question:} Does there exists $E'\subseteq U$ such that the partial graph $G=(V,E')$ has maximum degree at most $r$ and is maximal in $G$?
\end{minipage}}
\end{center}

\subsection{Complexity results}

We are first re-stating Theorem~\ref{thm: Ext_r-DS}:
For every fixed $r\geq 2$, \textsc{Ext $r$-DCPS} is $\np$-complete in bipartite graphs with maximum degree $r+1$, even if the set of forbidden edges induces a matching.

\begin{proof}
The proof is based on a reduction from \textsc{$(3,B2)$-SAT}. A main building block of our construction is a subgraph, denoted  $B_{k}(v)$ with $k< r$, containing $(kr)+1$ vertices which are arranged as a tree of depth 2 with root $v$ such that $v$ has $k$ children within this gadget and each child $w$ of $v$ has $r$ children. For each child $w$ of $v$ one edge connecting $w$ to a leaf in $B_{k}(v)$ will be forbidden in our construction, and we will use $F_{B_{k}(v)}$ to denote a fixed set of $k$ edges in $B_{k}(v)$ such that each child of $v$ is adjacent to an edge in  $F_{B_{k}(v)}$ and $v$ is not adjacent to an edge in  $F_{B_{k}(v)}$. The left part of Fig.~\ref{fig: Ext_r-DCPS} gives an illustration of $B_{k}(v)$. The purpose of this construction is that the root $v$ will connect to other parts of the graph, and the structure of $B_{k}(v)$ with the forbidden edges will make sure that a maximum $r$-degree constrained partial subgraph contains all edges between $v$ and its children in $B_{k}(v)$, and can consequently only contain $r-k$ edges connecting $v$ to a vertex outside $B_{k}(v)$.
Namely, if one edge~$e$ would be missing from the edges incident to $v$ in $B_{k}(v)$ in any maximal edge set~$E'$, say, $e=vv'$, then $E'$ would have to include the forbidden edge incident to~$v'$ by maximality.

\bigskip

Consider now an instance $I$ of \textsc{$(3,B2)$-SAT} with clauses ${\cal C}=\{c_1,\dots,c_m\}$ and variables ${\cal X}=\{x_1,\dots,x_n\}$. We build a bipartite graph $G=(V,E)$ of maximum degree $r+1$, together with a set $U$ of permitted edges (among which a maximal partial subgraph of degree at most $r$ should be chosen) as follows:
\begin{itemize}

\item[$\bullet$] For each clause $c\in \mathcal{C}$, build a clause gadget $H(c)=(V_c,E_c)$ which is a $B_{(r-2)}(c)$ (the root $c$ of $B_{(r-2)}(c)$ has $r-2$ children). Hence, we denote $U_c=E_c\setminus F_{B_{(r-2)}(c)}$ set of permitted edges in $H(c)$.
\bigskip

\item[$\bullet$]
For each variable $x$ introduce $3r$ new vertices which induce the primary subgraph denoted $H'(x)=(V'_x,E'_x)$. The vertex set $V'_x$ contains four special vertices $x,x',\neg x,\neg x' $. The vertices $x$ and $\neg x$ have $r-2$ distinct vertices in its neighborhoods and $x'$ and $\neg x'$ are connected to $r$ common vertices $v_x^{1},v_x^{2},...,v_x^{r}$. Also we connect $x,\neg x$ to $x', \neg x'$ respectively with two forbidden edges in $H'(x)$. The right part of Fig.~\ref{fig: Ext_r-DCPS} gives an illustration of $H'(x)$. By adding a component $B_{(r-1)}(y)$ for each vertex $y\in \{v_x^{i}\colon 1\leq i\leq r\}$ and identifying the root of  $B_{(r-1)}(y)$ with $y$, we construct a new subgraph $H(x)=(V_x,E_x)$. We define the set of forbidden edges in $H(x)$ by $F_x=\{xx',\neg x \neg x'\}\cup (\bigcup_{1\leq i\leq r} F_{B_{(r-1)}(v_x^{i})})$ and hence $U_x=E_x\setminus F_x$ denotes the set of permitted edges in $H(x)$.
\bigskip

\item[$\bullet$] We interconnect  $H(x)$ and $H(c)$ by adding edge $xc$ if $x$ appears positively in clause $c$ and $\neg xc$ if $x$ appears negatively. These {\em crossing edges} are always permitted and collected in the set $U_{\text{cross}}$.
\end{itemize}
\bigskip
\begin{figure}
\centering
\begin{tikzpicture}[scale=0.85, transform shape]
\tikzstyle{vertex}=[circle, draw, inner sep=0pt, inner sep=0pt, minimum size=0.6cm]
\tikzstyle{vertex1}=[circle, draw, inner sep=0pt, inner sep=0pt, minimum size=0.3cm]

\node () at (-1,-3) {$B_k(v)$};
\node[vertex] (x) at (0.2,0) {$v$};
\node[vertex1] (k1) at (-1,1.5) {};
\node[vertex1] (k2) at (-1,0.2) {};
\node[below of=k2, node distance=0.8 cm] (vdots) {$\vdots$};
\node[vertex1] (kr) at (-1,-1.5) {};

\node[vertex1] (k11) at (-2,2.05) {};
\node[vertex1] (k12) at (-2,1.7) {};
\node[below of=k12, node distance=0.25 cm] (vdots) {$\vdots$};
\node[vertex1] (k1r) at (-2,1) {};

\node[vertex1] (k21) at (-2,0.6) {};
\node[vertex1] (k22) at (-2,0.25) {};
\node[below of=k22, node distance=0.25 cm] (vdots) {$\vdots$};
\node[vertex1] (k2r) at (-2,-0.45) {};

\node[vertex1] (kr1) at (-2,-1) {};
\node[vertex1] (kr2) at (-2,-1.35) {};
\node[below of=kr2, node distance=0.25 cm] (vdots) {$\vdots$};
\node[vertex1] (krr) at (-2,-2.05) {};

\draw (x) -- (k1);
\draw (x) -- (k2);
\draw (x) -- (kr);
\draw (k1) -- (k12);
\draw (k1) -- (k1r);
\draw (k2) -- (k22);
\draw (k2) -- (k2r);
\draw (kr) -- (kr2);
\draw (kr) -- (krr);
\draw (k1) edge[ultra thick] (k11);
\draw (k2) edge[ultra thick] (k21);
\draw (kr) edge[ultra thick] (kr1);
\draw[decorate,decoration={brace,amplitude=4pt}] (k1r)++(-7pt,-5pt) -- ([xshift=-7pt,yshift=5pt]k11.center) node[midway, anchor=east]  {$r$};

\draw[decorate,decoration={brace,amplitude=4pt}] (k1)++(7pt,5pt) -- ([xshift=7pt,yshift=-5pt]kr.center) node[above, anchor=west]  {$k$};

\begin{scope}[xshift=4 cm]
\node () at (2,-3) {$H'(x)$};

\node[vertex] (x) at (0,0.5) {$x$};
\node[vertex, below of=x, node distance=1 cm] (nx) {$\neg x$};
\node[vertex, right of=x, node distance=1.5 cm] (xp) {$x'$};
\node[vertex, below of=xp, node distance=1 cm] (nxp) {$\neg x'$};
\node[vertex] (vx1) at (3.5,1.5) {$v_x^{1}$};
\node[vertex, below of=vx1, node distance=0.8 cm] (vx2) {$v_x^{2}$};
\node[below of=vx2, node distance=1 cm] (vxdots) {$\vdots$};
\node[vertex, below of=vxdots, node distance=1 cm] (vxr) {$v_x^{r}$};

\draw (x) edge[ultra thick] (xp);
\draw (nx) edge[ultra thick] (nxp);
\draw (xp)--(vx1)--(nxp)--(vx2)--(xp)--(vxr)--(nxp);

\node[vertex1] (x1) at (-0.8,1.5) {};
\node[vertex1] (x2) at (-0.4,1.5) {};
\node[right of=x2, node distance=0.5 cm] (xdots) {$\hdots$};
\node[vertex1](xr-2) at (0.6,1.5) {};
\draw[decorate,decoration={brace,amplitude=4pt}] (x1)++(-7pt,5pt) -- ([xshift=7pt,yshift=5pt]xr-2.center) node[midway, anchor=south]  {$r-2$};
\draw (x1)--(x)--(x2);
\draw (xr-2)--(x);

\node[vertex1] (nx1) at (-0.8,-1.5) {};
\node[vertex1] (nx2) at (-0.4,-1.5) {};
\node[right of=nx2, node distance=0.5 cm] (nxdots) {$\hdots$};
\node[vertex1](nxr-2) at (0.6,-1.5) {};

\draw (nx1)--(nx)--(nx2);
\draw (nxr-2)--(nx);

\end{scope}
\end{tikzpicture}
\caption{The gadgets $B_k(v)$ and $H'(x)$. Edges from the forbidden subset in  $F_{B_k(v)}$ are marked with bold line of the left side and more generally, edges not in $U$ are marked with bold line.}\label{fig: Ext_r-DCPS}
\end{figure}
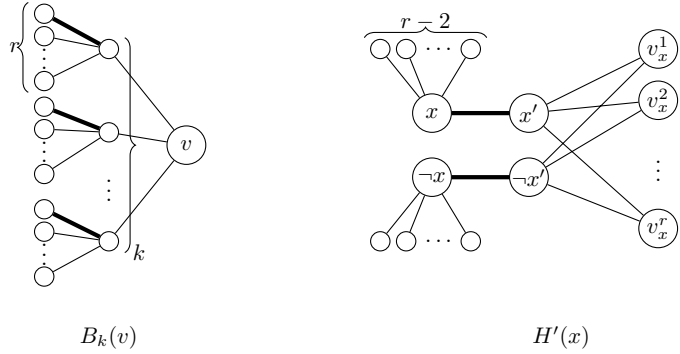

Let $U=(\bigcup_{c\in C} U_c)\cup (\bigcup_{x\in X} (U_x)\cup U_{\text{cross}}$ be the global set of permitted edges. This construction is computable in polynomial time, yielding a graph $G$ that is a bipartite graph of maximum degree~$r+1$.

We claim that there is a truth assignment $T$ of $I$ which satisfies all clauses iff there is a maximal $r$-degree constrained partial subgraph $G_S=(V,S)$ where $S\subseteq U$ of $G$.
\smallskip

If $T$ is a truth assignment of $I$ which satisfies all clauses, a maximal $r$-degree constrained partial subgraph $G_S=(V,S)$ with $S\subseteq U$ can be constructed as follows:\begin{enumerate}

\item For each variable gadget $H(c)$, by maximality $U_c\subseteq S$.

\item For each variable $x$ we add edges according to the assignment as follows: If $T(x)=\textit{true}$, we add $v_x^{i}x'$ for all $1\leq i\leq r$ and the two crossing edges connecting $\neg x$ with their respective clause gadgets. If $T(x)=\textit{false}$, we add $v_x^{i}\neg x'$ for all $1\leq i\leq r$ and the two crossing edges connecting $x$ with their respective clause gadgets. In both cases by maximality we add $2(r-2)$ pendent edges incident to $x$ and $\neg x$ and also all permitted edges in $B_{(r-1)}(y)$ for all $y\in \{v_x^{i}\colon 1\leq i\leq r\}$.

\item At last, for each variable $c$ which has more than one true literal in assignment $T$, add some arbitrary crossing edges to $c$ such that $d_{G_S}(c)=r$.
\end{enumerate}

The resulting subgraph is a maximal $r$-degree constrained partial subgraph $G_S=(V,S)$ with $S\subseteq U$.
\bigskip

Conversely, assume the existence of a maximal $r$-degree partial subgraph $G_S=(V,S)$ with $S\subseteq U$. First, recall that for every gadget $B_k(v)$, we must have $E(B_k(v))\setminus F_{B_k(v)}\subseteq S$ for $k\in \{r-1,r-2\}$.
Moreover, for each variable gadget $H(x)$, at least one of the pairs of crossing edges incident to $x$ and $\neg x$ have to be in $S$ (by maximality). Hence we set $T(x)=\textit{true}$ if both crossing edges incident to $\neg x$ are in $S$ and otherwise we set $T(x)=\textit{false}$ (if both, we choose arbitrarily one of them). This assignment is valid and since for each clause $c$, at most two crossing edges incident to vertex $c$ are in $S$ ($G_S$ is a subgraph with maximum degree $r$), then $T$ satisfies all clauses of $I$.
\end{proof}

\subsection{Parameterized perspective}

\begin{proposition}\label{rem:thm: Ext_r-DCPS}
For graphs with maximum degree $r+1$, \textsc{Ext $r$-DCPS} is polynomial-time decidable when $r=1$ and is in \textbf{FPT} respect to the number of isolated edges in $E\setminus U$ for $r\geq 2$.
\end{proposition}

\begin{proof}
Consider the subgraph $G'=(V,E\setminus U)$ induced by $E\setminus U$, i.e. the graph induced by the forbidden edges. Since maximum degree of $G$ is $r+1$, so
if $G'$ contains a triangle or a path of length at least $3$, the answer is no. Therefore, we can suppose that $G'$ is a collection of stars. If one leaf of a star of $G'$ has a degree at most $r$ in $G$, then this star might be an isolated edge in $G'$ and it is exactly for one of these two endpoints (otherwise,
the answer is no). Hence, let $P_{G'}$ be the set of the stars which are isolated edges in $G'$ and such that both extremities are of degree $r+1$ in $G$.
\bigskip

$\bullet$ For $r\geq 2$, for the set of stars of $G'$ with more than one edge, leaves and center are clearly determined and for for each single edge of $G'$ not in $P_{G'}$, the extremity with degree less than $r$ is chosen as a center. Now, for each star in $P_{G'}$ we have to determine one of the extremities as a center and the other one as a leaf. We can now build the set $L$ of leaves for all stars of $G'$. Let $E'_L=\{uv\in E\colon u\in L\}$ and $G'_L=(V,E'_L)$. We check for all possible labelings, if there is a label which satisfies two following conditions the answer is yes, else the answer is no.
\begin{enumerate}
\item for each $v\in L$, $d_{G'_L=}(v)=r+1$.
\item for each vertex $v\in V\setminus L$, $d_{G'_L}(v)\leq r$.
\end{enumerate}

$\bullet$ For $r=1$, we make a new graph $H$ by omitting all sets of vertices $\{u'_i,u_i,v_i,w_i,w'_i\}$ for the stars $[u_i,v_i,w_i]$ of $G'$ (with center $v_i$) where $u'_i,w'_i$ are neighbors of the leaves $u_i,w_i$ (without $v_i$). Notice at each time $u'_i,v'_i$ have to be disjoint from one star to another one, otherwise the answer is no. Now $H$ is a collection of paths (maybe trivial) and cycles where the forbidden edges induce a matching. Remove from $H$ all cycles and the paths where both end edges are in $U$. Now $H$ is a collection of paths where at least one of end edges is forbidden. For all of these paths, start from one side and satisfy the maximality by assigning the first possible edge to a forbidden edge, if there is a path does not satisfy the maximality, the answer is no, else the answer is yes.
\end{proof}

\begin{rmk}
For graphs with maximum degree $r+1$, \textsc{Ext $r$-DCPS} with $r\geq 2$ is parameterized equivalent to \textsc{SAT} with respect to the number of isolated edges in $E\setminus U$ and variables, respectively.
\end{rmk}

\section{Generalization of edge cover extension}\label{Gene:Ext edge_cover}

We assume $r$ is a constant fixed greater than one (but all results given here hold even if $r$ depends on the graph).
The \textsc{$r$-degree edge-cover problem} (\textsc{$r$-EC} for short) is defined as follows: the instance $G=(V,E)$ is a graph,  $\leq_G=\subseteq$,
$\mathop{presol}(G)=2^E$, and $E'\in \mathop{sol}(I)$ is a feasible solution iff each $v\in V$
is incident to at least $r\geq 1$ distinct edges $e\in C$. The particular case of $r=1$ corresponds to the famous edge cover problem in graphs.
A partial graph $G'=(V,E')$ of \textsc{$r$-EC} will be called an {\em $r$-EC solution} in the following.\footnote{A different generalization of edge cover was considered in \cite{FerMan2009}, requiring that each connected component induced by the edge cover solution contains at least $t$ edges. Clearly, if every vertex is incident to at least~$r$ edges from the cover, then each connected component induced by the edge cover solution contains at least $r$ edges.}
\bigskip

The optimization version of a generalization of \textsc{$r$-EC} known as the \textsc{Min lower-upper-cover problem} (\textsc{Min\nolinebreak LUCP}),
consists of, given a graph  $G$ where $G=(V,E)$ and two non-negative functions $a,b$ from $V$ such that $\forall v\in V$, $0\leq a(v)\leq b(v)\leq d_G(v)$, of finding a subset $M\subseteq E$ such that the partial graph $G[M]=(V,M)$ induced by $M$ satisfies $a(v)\leq d_{G[M]}(v)\leq b(v)$ (such a solution will be called a {\em lower-upper-cover}) and minimizing its total size $|M|$ among all such solutions (if any). Hence, an $r$-EC solution
corresponds to a  lower-upper-cover with $a(v)=r$ and $b(v)=d_G(v)$ for every $v\in V$. \textsc{Min\nolinebreak LUCP} is known to be solvable in polynomial time even for edge-weighted graphs (Theorem 35.2 in Chapter~35 of Volume~A in \cite{AS03}).

\begin{center}
\fbox{\begin{minipage}{.95\textwidth}
\noindent{\textsc{Ext $r$-EC}}\\\nopagebreak
{\bf Input:} A graph $G=(V,E)$ and $U\subseteq E$. \\\nopagebreak
{\bf question:} Does there exists $E' \supseteq U$ such that the partial graph $G'=(V,E')$ has minimum degree at least $r$ and is minimal in $G$?
\end{minipage}}
\end{center}

\subsection{Complexity results}

We are now proving Theorem~\ref{thm: Ext_r-LEC}, which we are re-formulating here for convenience.

For every fixed $r\geq 1$, \textsc{Ext $r$-EC} is $\np$-complete in bipartite graphs with maximum degree $r+2$, even if the pre-solution $(V,U)$ is a collection of paths of length at most~1.

\begin{proof}
The proof is based on a reduction from \textsc{$(3,B2)$-SAT}. A main building block $B_{2r}(v)$ (or $B_{2r+1}(v)$) in our construction is based on a complete bipartite subgraph of $2r$ (or $2r+1$) vertices where one specified edge between two special vertices $v$ and $v'$  has been deleted. So,
$B_{2r}(v)=K_{r,r}-\{vv'\}$ and $B_{2r+1}(v)=K_{r+1,r}-\{vv'\}$. Except for these two vertices $v,v'$, the other vertices of $B_{2r}(v)$ are not linked to any other vertex in the whole construction, while for $B_{2r+1}(v)$, it is only the case of $v$ (i.e., only $v$ is also linked outside $B_{2r+1}(v)$).  
Block $B_{2r}(v)$ will appear five times in each {\em variable gadget} and block $B_{2r+1}(v)$ will correspond to  each {\em clause gadget} (see Fig.~\ref{Fig:r-EC} for an illustration). By construction, all edges of $B_{2r}(v)$ will belong to any $r$-EC solution (in fact, vertices $v$ and $v'$ still need one more edge to satisfy the minimum degree constraint) and for $B_{2r+1}(c)$, it will be almost the case (except
for few edges of $B_{2r+1}(c)$ incident to $c$, as all neighbors of $c$ in  $B_{2r+1}(c)$ have degree $r+1$, and all edges between $N(c)$ and $N(N(c))\setminus\{c\}$ have to be in the edge cover; $c$ will need one more incident edge in the edge cover besides (some of) the edges from $B_{2r+1}(c)$).

Now, consider an instance $I$ of \textsc{$(3,B2)$-SAT} with clauses ${\cal C}=\{c_1,\dots,c_m\}$ and variables ${\cal X}=\{x_1,\dots,x_n\}$. We build a bipartite graph $G=(V,E)$ of maximum degree $r+2$, together with a set $U$ of permitted edges as follows:

\begin{itemize}
\item[$\bullet$] 
For each clause $c \in \cal C$, we build a clause gadget $B_{2r+1}(c)$ which is a component $K_{r,r-1}$ plus two vertices $c,c'$. An illustration of $B_{2r+1}(c)$ is given in the left side of Fig.~\ref{Fig:r-EC}.
\item[$\bullet$]  For each variable $x\in \cal X$,  we construct a subgraph $H(x)=(V_x,E_x)$ as follows: build two $P_5$ denoted $P=(x,l,m,r,\neg  x)$
and $P'=(x',l',m',r',\neg  x')$ respectively; then between each pair of vertices  $v,v'$ of $P$ and $P'$ a block $B_{2r}(v)$ is added for each $v$ on $P$; this interconnects $v$ on $P$ with the corresponding vertex $v'$ on $P'$, as $v$ and $v'$ are special to  $B_{2r}(v)$. The variable gadget
$H(x)=(V_x,E_x)$ is illustrated to the right of Figure~\ref{Fig:r-EC}.
\item[$\bullet$]  We interconnect $H(x)$ and $B_{2r+1}(c)$ where $x$ is a literal of clause $c$ by adding edge $xc$ if $x$ appears positively in $c$ and the edge $\neg xc$ if $x$ appears negated.
Such edges will be called \emph{crossing}.
\end{itemize}

Now, it is easy to see that $G$ is bipartite of maximum degree $r+2$. Finally, let $U=\{x_il_i,\neg x_ir_i\suchthat 1\leq i \leq n\}$, picking the corresponding vertices and edges in each $H(x_i)$.

\begin{figure}[tbp]
\centering
\begin{tikzpicture}[scale=0.65, transform shape]
\tikzstyle{vertex}=[circle, draw, inner sep=0pt, inner sep=0pt, minimum
 size=0.8cm]

 \tikzstyle{vertex1}=[circle, draw, inner sep=0pt, inner sep=0pt, minimum size=0.5cm]

\node[vertex] (ci) at (0,1.7) {$c$};
\node[vertex1] (ch1ci) at (-1.6,0.6) {};
\node[vertex1, above of=ch1ci, node distance =0.7 cm] (ch2ci) {};
\node[above of=ch2ci, node distance=0.8 cm] (cidots) {$\vdots$};
\node[vertex1, above of=cidots, node distance =0.8 cm] (chr-1ci) {};

\draw (ci) -- (ch1ci);
\draw (ci) -- (ch2ci);
\draw (ci) -- (chr-1ci);

\node[vertex1] (2ch1ci) at (-3,0.4) {};
\node[vertex1, above of=2ch1ci, node distance =0.6 cm] (2ch2ci) {};
\node[above of=2ch2ci, node distance=1 cm] (2cidots) {$\vdots$};
\node[vertex1, above of=2cidots, node distance =1 cm] (2chrci) {};

\draw (2ch1ci) -- (ch1ci);
\draw (2ch1ci) -- (ch2ci);
\draw (2ch1ci) -- (chr-1ci);

\draw (2ch2ci) -- (ch1ci);
\draw (2ch2ci) -- (ch2ci);
\draw (2ch2ci) -- (chr-1ci);

\draw (2chrci) -- (ch1ci);
\draw (2chrci) -- (ch2ci);
\draw (2chrci) -- (chr-1ci);

\node[vertex] (cpi) at (-4.5,1.7) {$c'$};

\draw (2ch1ci) -- (cpi);
\draw (2ch2ci) -- (cpi);
\draw (2chrci) -- (cpi);

\node () at (-2.4,3.5) {$K_{r,r-1}$};
\node () at (-2.7,-0.6) {$B_{2r+1}(c)$};

\begin{scope}[xshift=2cm]

\node[vertex] (xi) at (0,0) {$x$};
\node[vertex, right of=xi, node distance=2.8cm] (li) {$l$};
\node[vertex, right of=li, node distance=2.8cm] (mi) {$m$};
\node[vertex, right of=mi, node distance=2.8cm] (ri) {$r$};
\node[vertex, right of=ri, node distance=2.8cm] (nxi) {$\neg x$};

\draw [ultra thick] (xi) -- (li);
\draw [ultra thick] (ri) -- (nxi);
\draw (li) -- (mi) -- (ri);

\node[vertex1] (ch1xi) at (-1.1,1.3) {};
\node[vertex1, right of=ch1xi, node distance =0.6 cm] (ch2xi) {};
\node[right of=ch2xi, node distance=0.7 cm] (xidots) {$\cdots$};
\node[vertex1, right of=xidots, node distance= 0.7 cm] (chr-1xi) {};

\draw (xi) -- (ch1xi);
\draw (xi) -- (ch2xi);
\draw (xi) -- (chr-1xi);

\node[vertex1] (2ch1xi) at (-1.1,2.3) {};
\node[vertex1, right of=2ch1xi, node distance =0.6 cm] (2ch2xi) {};
\node[right of=2ch2xi, node distance=0.7 cm] (2xidots) {$\cdots$};
\node[vertex1, right of=2xidots, node distance= 0.7 cm] (2chr-1xi) {};

\draw (ch1xi) -- (2ch1xi);
\draw (ch1xi) -- (2ch2xi);
\draw (ch1xi) -- (2chr-1xi);

\draw (ch2xi) -- (2ch1xi);
\draw (ch2xi) -- (2ch2xi);
\draw (ch2xi) -- (2chr-1xi);

\draw (chr-1xi) -- (2ch1xi);
\draw (chr-1xi) -- (2ch2xi);
\draw (chr-1xi) -- (2chr-1xi);

\node[vertex] (xpi) at (0,3.6) {$x'$};

\draw (2ch1xi) -- (xpi);
\draw (2ch2xi) -- (xpi);
\draw (2chr-1xi) -- (xpi);

\node[vertex1] (ch1li) at (1.6,1.3) {};
\node[vertex1, right of=ch1li, node distance =0.6 cm] (ch2li) {};
\node[right of=ch2li, node distance=0.7 cm] (lidots) {$\cdots$};
\node[vertex1, right of=lidots, node distance= 0.7 cm] (chr-1li) {};

\draw (li) -- (ch1li);
\draw (li) -- (ch2li);
\draw (li) -- (chr-1li);

\node[vertex1] (2ch1li) at (1.6,2.3) {};
\node[vertex1, right of=2ch1li, node distance =0.6 cm] (2ch2li) {};
\node[right of=2ch2li, node distance=0.7 cm] (2lidots) {$\cdots$};
\node[vertex1, right of=2lidots, node distance= 0.7 cm] (2chr-1li) {};

\draw (ch1li) -- (2ch1li);
\draw (ch1li) -- (2ch2li);
\draw (ch1li) -- (2chr-1li);

\draw (ch2li) -- (2ch1li);
\draw (ch2li) -- (2ch2li);
\draw (ch2li) -- (2chr-1li);

\draw (chr-1li) -- (2ch1li);
\draw (chr-1li) -- (2ch2li);
\draw (chr-1li) -- (2chr-1li);

\node[vertex] (lpi) at (2.8,3.6) {$l'$};

\draw (2ch1li) -- (lpi);
\draw (2ch2li) -- (lpi);
\draw (2chr-1li) -- (lpi);

\node[vertex1] (ch1mi) at (4.4,1.3) {};
\node[vertex1, right of=ch1mi, node distance =0.6 cm] (ch2mi) {};
\node[right of=ch2mi, node distance=0.7 cm] (midots) {$\cdots$};
\node[vertex1, right of=midots, node distance= 0.7 cm] (chr-1mi) {};

\draw (mi) -- (ch1mi);
\draw (mi) -- (ch2mi);
\draw (mi) -- (chr-1mi);

\node[vertex1] (2ch1mi) at (4.4,2.3) {};
\node[vertex1, right of=2ch1mi, node distance =0.6 cm] (2ch2mi) {};
\node[right of=2ch2mi, node distance=0.7 cm] (2midots) {$\cdots$};
\node[vertex1, right of=2midots, node distance= 0.7 cm] (2chr-1mi) {};

\draw (ch1mi) -- (2ch1mi);
\draw (ch1mi) -- (2ch2mi);
\draw (ch1mi) -- (2chr-1mi);

\draw (ch2mi) -- (2ch1mi);
\draw (ch2mi) -- (2ch2mi);
\draw (ch2mi) -- (2chr-1mi);

\draw (chr-1mi) -- (2ch1mi);
\draw (chr-1mi) -- (2ch2mi);
\draw (chr-1mi) -- (2chr-1mi);

\node[vertex] (mpi) at (5.6,3.6) {$m'$};

\draw (2ch1mi) -- (mpi);
\draw (2ch2mi) -- (mpi);
\draw (2chr-1mi) -- (mpi);

\node[vertex1] (ch1ri) at (7.2,1.3) {};
\node[vertex1, right of=ch1ri, node distance =0.6 cm] (ch2ri) {};
\node[right of=ch2ri, node distance=0.7 cm] (ridots) {$\cdots$};
\node[vertex1, right of=ridots, node distance= 0.7 cm] (chr-1ri) {};

\draw (ri) -- (ch1ri);
\draw (ri) -- (ch2ri);
\draw (ri) -- (chr-1ri);

\node[vertex1] (2ch1ri) at (7.2,2.3) {};
\node[vertex1, right of=2ch1ri, node distance =0.6 cm] (2ch2ri) {};
\node[right of=2ch2ri, node distance=0.7 cm] (2ridots) {$\cdots$};
\node[vertex1, right of=2ridots, node distance= 0.7 cm] (2chr-1ri) {};

\draw (ch1ri) -- (2ch1ri);
\draw (ch1ri) -- (2ch2ri);
\draw (ch1ri) -- (2chr-1ri);

\draw (ch2ri) -- (2ch1ri);
\draw (ch2ri) -- (2ch2ri);
\draw (ch2ri) -- (2chr-1ri);

\draw (chr-1ri) -- (2ch1ri);
\draw (chr-1ri) -- (2ch2ri);
\draw (chr-1ri) -- (2chr-1ri);

\node[vertex] (rpi) at (8.4,3.6) {$r'$};

\draw (2ch1ri) -- (rpi);
\draw (2ch2ri) -- (rpi);
\draw (2chr-1ri) -- (rpi);

\node[vertex1] (ch1nxi) at (10,1.3) {};
\node[vertex1, right of=ch1nxi, node distance =0.6 cm] (ch2nxi) {};
\node[right of=ch2nxi, node distance=0.7 cm] (nxidots) {$\cdots$};
\node[vertex1, right of=nxidots, node distance= 0.7 cm] (chr-1nxi) {};

\draw (nxi) -- (ch1nxi);
\draw (nxi) -- (ch2nxi);
\draw (nxi) -- (chr-1nxi);

\node[vertex1] (2ch1nxi) at (10,2.3) {};
\node[vertex1, right of=2ch1nxi, node distance =0.6 cm] (2ch2nxi) {};
\node[right of=2ch2nxi, node distance=0.7 cm] (2nxidots) {$\cdots$};
\node[vertex1, right of=2nxidots, node distance= 0.7 cm] (2chr-1nxi) {};

\draw (ch1nxi) -- (2ch1nxi);
\draw (ch1nxi) -- (2ch2nxi);
\draw (ch1nxi) -- (2chr-1nxi);

\draw (ch2nxi) -- (2ch1nxi);
\draw (ch2nxi) -- (2ch2nxi);
\draw (ch2nxi) -- (2chr-1nxi);

\draw (chr-1nxi) -- (2ch1nxi);
\draw (chr-1nxi) -- (2ch2nxi);
\draw (chr-1nxi) -- (2chr-1nxi);

\node[vertex] (nxpi) at (11.2,3.6) {$\neg x'$};

\draw (2ch1nxi) -- (nxpi);
\draw (2ch2nxi) -- (nxpi);
\draw (2chr-1nxi) -- (nxpi);

\draw (xpi) -- (lpi) -- (mpi) -- (rpi) -- (nxpi);
\node () at (4.7,-0.6) {\textbf{$H(x)$}};
\end{scope}
\end{tikzpicture}
 \caption{Block $B_{2r+1}(c)$ for clause $c$ is depicted on the left-hand side. The subgraph $H(x)=(V_x,E_x)$ is shown on the right-hand side. Edges of $U$ are drawn in bold.}\label{Fig:r-EC}
\end{figure}
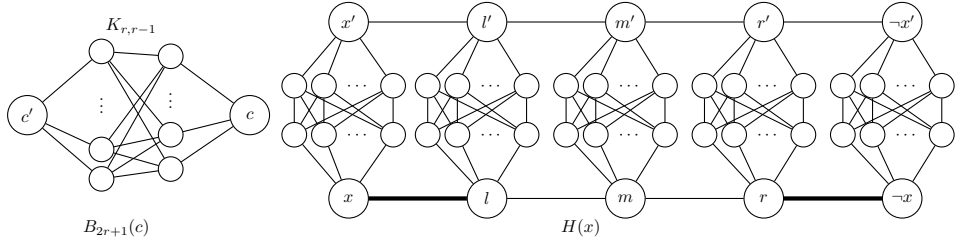

We claim that there is a truth assignment $T$ of $I$ which satisfies all clauses iff $G$ admits a minimal $r$-EC solution $H=(V,S)$ where $U\subseteq S$ of $G$.
\medskip

If $T$ is a truth assignment of $I$ which satisfies all clauses, a minimal $r$-EC solution $H=(V,S)$ can be constructed as follows:

\begin{itemize}
\item[$\bullet$]  For each variable $x$, if $T(x)=\textit{true}$, $\{xc: x \mbox{~appears positively in~} c\}\cup (E_x\setminus \{lm, l'm'\})\subseteq S$, and if $T(x)=\textit{false}$, $\{\neg xc:  x \mbox{~appears negatively in~} c\}\cup (E_x\setminus \{mr, m'r'\})\subseteq S$.

\item[$\bullet$] Since $T$ is a satisfying assignment, we have already added in the previous step $k\geq 1$ crossing edges connected to block $B_{2r+1}(c)$ for each clause $c$. Then, we delete arbitrarily $k-1$ edges $S_c$ of $B_{2r+1}(c)$  incident to $c$, and we add $E(B_{2r+1}(c))\setminus S_c$ to $S$.
\end{itemize}

Conversely, assume that $H=(V,S)$ is a minimal $r$-EC solution of $G$ containing $U$. By considering the variable gadget $H(x)$ and in order to keep minimality
$S$ contains either $lm$ or $rm$ (not both since $\{xl,\neg xr\}\subset S$ by hypothesis and since all edges in the block $B_{2r}(v)$ for $v\in\{l,m,r\}$ have to be included into the edge cover by our previous observations). Hence, we set $T(x_i)=\textit{true}$ if $rm\in S$ and $T(x_i)=\textit{false}$ if $lm\in S$. Since $H$ has to get a minimum degree at least $r$ for each vertex and vertex $c$ has a degree $r-1$ in clause gadget $B_{2r+1}(c)$, then the partial graph $H$ contains at least one crossing edge incident to each $c$ and hence $T$ is a valid assignment of $I$.
\end{proof}

\subsection{Parameterized perspective}
Consider a graph $G=(V,E)$ and let $U\subseteq E$; recall that $G[U]=(V,U)$ and $G[{\overline{U}}]=(V,E\setminus U)$ denote the  partial graphs induced by $U$
and $E\setminus U$, respectively. Finally, $V(U)$ denotes the endpoints of edges in $U$ (or equivalently the non-isolated vertices of $G[U]$). The next property is quite immediate for each solution $G'=(V,E')$ of \textsc{Ext $r$-EC}.

\begin{lemma}\label{lem: r-LEC}
$\mathop{ext}(G,U)\neq \emptyset$ iff there is an $r$-EC solution $G'=(V,E')$ where $E' \supseteq U$  such that $S_{G'}=\{v\in V(U)\suchthat d_{G'}(v)>r\}$
is an independent set of $G[U]$.
\end{lemma}
\begin{proof}
The condition is obviously necessary, as an edge among two vertices $x,y\in  U$ 
of a minimal extension $X\supseteq U$ certifying that $\mathop{ext}(G,U)\neq \emptyset$ can only exist if $x$ or $y$ is, or both  $x$ and $y$ are, incident to at most $r$ edges from $X$ because of minimality.

Let us look into the other direction. Let $G'=(V,E')$ be a partial subgraph of $G$ with $U\subseteq E'$
and  $d_{G'}(v)\geq r$ for all $v\in V$. Moreover, assume $S_{G'}=\{v\in V(U)\suchthat d_{G'}(v)>r\}$ is an independent set of $G[U]$. Consider any minimal
partial subgraph $H=(V,E_H)$ of $G'=(V,E')$ maintaining the property   $d_{G'}(v)\geq r$ for all $v\in V$. Since $S_{G'}$ is an independent set of $G[U]$,
$U\subseteq E_H$ and therefore, $E_H\in \mathop{ext}(G,U)$. 
\end{proof}

Now let us establish a relation between the instances of the two problems \textsc{Ext $r$-EC} and \textsc{Min\nolinebreak LUCP}.
Let $(G,U)$ be a \yes-instance of \textsc{Ext $r$-EC} where $G=(V,E)$ is a graph of minimum degree at least $r$ and $U\subseteq E$.
So, $\mathop{ext}(G,U)\neq \emptyset$ which implies by Lemma \ref{lem: r-LEC} the existence of a particular independent set $S$ of $G[U]$.
We build an instance $(G[{\overline{U}}],a,b)$,  $\overline{U}=E\setminus U$, of \textsc{Min\nolinebreak LUCP}, where  $a,b$ are two non-negative
functions defined as follows:

$$a(v): =\begin{cases}r & \text{if } v\in V\setminus V(U)\\r-d_{G[U]}(v)& \text{if } v\in V(U), \end{cases}$$

\noindent and

$$b(v): =\begin{cases}d_{G}(v) & \text{if } v\in \left(V\setminus V(U)\right)\cup S\\r-d_{G[U]}(v)& \text{if } v\in V(U)\setminus S.\end{cases}$$

\noindent
The next property is rather immediate.

\begin{theorem}\label{theo: r-LEC and L-U-EC}
If there is a solution of \textsc{Min\nolinebreak LUCP} for the instance $(G[{\overline{U}}],a,b)$, then $\mathop{ext}(G,U)\neq \emptyset$.
\end{theorem}

\begin{proof}
Assume that instance $(G[{\overline{U}}],a,b)$ of \textsc{Min\nolinebreak LUCP} admits a feasible solution and let $G^*=(V,E^*)$ be an optimal solution. Then, the partial graph $H=(V,E^*\cup U)$ satisfies the hypothesis of Lemma  \ref{lem: r-LEC} (actually, $H$ is already minimal
with respect to property  $\forall v\in V$, $d_{H}(v)\geq r$).
\end{proof}

\noindent Using the outcome given in Theorem \ref{theo: r-LEC and L-U-EC}, the next result is rather straightforward. We reformulate Theorem~\ref{theo:ext r-LEC FPT} for the reader's convenience.

\textsc{Ext $r$-EC}, with standard parameter, is in $\fpt$.

\begin{proof}
Consider the algorithm that lists all possible instances $(G[{\overline{U}}],a,b)$ for \textsc{Min\nolinebreak LUCP} by checking all
independent sets of $G[U]$ included in $V(U)$ from an instance $I=(G,U)$ of \textsc{Ext $r$-EC}. This means that we try different values for function~$b$. Since \textsc{Min\nolinebreak LUCP} is solvable in polynomial time \cite{AS03}, then the running time is dominated by the procedure that lists all possible independent sets of $G[U]$, i.e., there are $3^{|U|}$ possibilities: each vertex of each edge in $U$ can be either included or excluded of the independent set, except for taking both endpoints in.
\end{proof}

\section{Planar Problems (Proofs for Theorem~\ref{thm:np-completeness-planar-summary})}
\label{sec-planar}

\noindent
We already mentioned above that we know the extension variants of \textsc{Vertex Cover} and of \textsc{Dominating Set} are $\np$-hard on planar cubic graphs. In the following, we will consider \textsc{Ext EC, DS, EDS, EM} in planar bipartite graphs and will show that all the problems are $\np$-hard even we restrict the graphs to be subcubic planar bipartite. In order to prove these results, we will give some reductions from \textsc{4-Bounded Planar 3-Connected SAT problem} (\textsc{4P3C3SAT} for short) which was already explained above. 
\smallskip

Let $I=(X,C)$ be an instance of \textsc{4P3C3SAT} where $X=\{x_1,\dots,x_n\}$ and $C=\{c_1,\dots,c_m\}$ are variable and clause sets of $I$ respectively. By definition, the graph $G=(V,E)$ with $V=\{c_1,\dots,c_m\}\cup \{x_1,\dots,x_n\}$ and $E=\{c_ix_j\colon x_j \text{ or }\neg x_j \text{ appears in } c_j\}$ is planar. In the following, we always assume that the planar graph comes with an embedding in the plane.
Informally, we are looking a new construction by putting some gadgets instead of vertices $x_i$ of $G$ which satisfies two following conditions: (1) as it can be seen in Fig.~\ref{GadgfigExtEC}, the constructions distinguishes between the cases that a variable $x_i$ appears positively and negatively in some clauses (2) the new construction preserves planarity. 
\smallskip

Suppose that a variable $x_i$ appears in at most four clauses $c_1,c_2,c_3,c_4$ of the original instance $I$ such that in the induced (embedded) subgraph $G_i=G[\{x_i,c_1,c_2,c_3,c_4\}]$, $c_1x_i$, $c_2x_i$, $c_3x_i$, $c_4x_i$ is an anti-clockwise ordering of edges around $x_i$. By looking at $G_i$ and considering the fact that variable $x_i$ appears negated or non-negated in the four clauses $c_1,c_2,c_3,c_4$ in $I$, the construction should satisfies the 3 following cases:
\pagebreak[2]

\begin{itemize}
\item[$\bullet$] case 1: $x_i\in c_1, c_2$ and $\neg x_i \in c_3,c_4$,
\item[$\bullet$] case 2: $x_i\in c_1,c_3$ and $\neg x_i \in c_2,c_4$,
\item[$\bullet$] case 3: $x_i\in c_1,c_2,c_3$ and $\neg x_i \in c_4$.
\end{itemize}

Note that all other cases are included in these 3 cases by rotations and / or replacing $x_i (\neg x_i)$ with $\neg x_i (x_i)$.
\bigskip

In Theorem \ref{Bip_Ext_EC} we show that \textsc{Ext EC} is $\np$-hard for subcubic bipartite graphs. In order to this, we proposed a reduction from \textsc{(3,B2)-SAT} in which the corresponding construction does not preserve planarity. In the following, we will propose a new construction containing 3 different variable gadgets.


\begin{theorem}\label{thm: EXT planar EC}
\textsc{Ext EC} is  $\np$-hard for planar bipartite graphs of maximum degree 3.
\end{theorem}

\begin{proof}
The proof is based on a reduction from \textsc{4P3C3SAT}. We start from graph $G$ and build a planar bipartite graph $H=(V_H,E_H)$ by replacing every node $x_i$ in $G$ with one of the three variable gadgets $H(x_i)$ which are illustrated in Fig.~\ref{Fig:Ext planar EC}. The forced edge set $U_i$, corresponding to variable gadget $H(x_i)$, contains $t_il_i, r_if_i$ for case 1, $t_i^{1}l_i^{1}, r_i^{1}f_i^{1}, t_i^{2}l_i^{2}, r_i^{2}t_i^{2}$ for case 2 and $t_i^{1}l_i^{1}, t_i^{2}l_i^{2}, r_if_i$ for case 3. Let $U=\bigcup_{1\leq i\leq n}U_i$, the set of forced edges of $H$. This construction can be done in polynomial time and the final graph $H$ is planar bipartite with maximum degree 3. We now claim that $I$ is satisfiable iff $H$ admits a minimal edge cover containing $U$.

\smallskip

Suppose $T$ is a truth assignment of $I$ which satisfies all clauses. For each clause $c_j$, let $h(j)$ be an index such that variable $x_{h(j)}$ satisfies clause $c_j$ for $T$ and let $J=[n]\setminus h([m])$ be the unused indices by mapping $h$. We construct a minimal edge cover $S$ of $H$ by considering all possibilities of $H(x_i)$:

\begin{itemize}
\item[$\bullet$] for each variable gadget $H(x_i)$ which complies with "case 1" we set:
\begin{equation*}
\begin{split}
		S_1 := & \{t_{h(j)}c_j, m_{h(j)}r_{h(j)} \colon T(x_{h(j)})=\textit{true}, x_{h(j)} \text{ appears positively in } c_j\}  \\
		 \cup& \{f_{h(j)}c_j, m_{h(j)}l_{h(j)} \colon T(x_{h(j)})=\textit{false}, x_{h(j)} \text{ appears negatively in } c_j\}\\
		\cup &\{m_ir_i \colon i\in J\}.
		\end{split}
\end{equation*}

\item[$\bullet$] for each variable gadget $H(x_i)$ which complies with "case 2" by assuming $h(j)=k$ we set:

\begin{equation*}
\begin{split}
		S_2 := & \{t_k^{1}c_j \text{ }(t_k^{2}c_j), m_k^{1}r_k^{1},r_k^{1}p_k^{1}, m_k^{2}r_k^{2},r_k^{2}p_k^{2}\colon T(x_k)=\textit{true} \land t_k^{1}c_j\in E_H \text{ }( t_k^{2}c_j\in E_H)\}  \\
		 \cup& \{f_k^{1}c_j \text{ }(f_k^{2}c_j), m_k^{1}l_k^{1},l_k^{2}p_k^{1}, m_k^{2}l_k^{2},l_k^{1}p_k^{2}\colon T(x_k)=\textit{false} \land f_k^{1}c_j\in E_H \text{ }( f_k^{2}c_j\in E_H)\}\\
		\cup &\{l_i^{1}p_i^{2},l_i^{1}m_i^{1},l_i^{2}m_i^{2},l_i^{2}p_i^{1}\colon i\in J\}.
		\end{split}
\end{equation*}

\item[$\bullet$] for each variable gadget $H(x_i)$ which complies with "case 3" by assuming $h(j)=k$ we set:
\begin{equation*}
\begin{split}
		S_3 := & \{t_k^{1}c_j \text{ }(t_k^{2}c_j), m_k^{1}r_k,m_k^{2}r_k\colon T(x_k)=\textit{true} \land t_k^{1}c_j\in E_H \text{ }( t_k^{2}c_j\in E_H)\}  \\
		 \cup& \{f_kc_j, m_k^{1}l_k^{1},m_k^{2}l_k^{2} \colon T(x_k)=\textit{false}\}\\
		\cup &\{l_i^{1}m_i^{1},l_i^{2}m_i^{2}\colon i\in J\}.
		\end{split}
\end{equation*}
\end{itemize} 

Finally we set $S=S_1\cup S_2\cup S_3 \cup U$. One can easily check that $S$ is a minimal edge cover of $H$.

\smallskip
Conversely, suppose $S$ is a minimal edge cover of $H$ containing $U$. By minimality of $S$ we propose an assignment $T$ of $I$ depending on different types of variable gadgets of $H$ as follows:

\begin{itemize}

\item[$\bullet$] If $H(x_i)$ complies with case 1, in order to cover vertex $m_i$, the edge cover $S$ either contains $m_ir_i$ or $m_il_i$ (not both by minimality). This means that we set $T(x_i)=\textit{true}$ (resp., $T(x_i)=\textit{false}$) if  $m_ir_i\in S$ (resp., $m_il_i\in S$).

\smallskip

\item[$\bullet$] If $H(x_i)$ complies with case 2, in order to cover vertices $m_i^{1}, m_i^{2}, p_i^{1}, p_i^{2}$, the edge cover $S$ contains exactly one of edges in pairs $(l_i^{1}m_i^{1},r_i^{1}m_i^{1}), (l_i^{2}m_i^{2},r_i^{2}m_i^{2}), (l_i^{1}p_i^{2},r_i^{2}p_i^{2}), (r_i^{1}p_i^{1},l_i^{2}p_i^{1})$. Hence, we set
\begin{itemize}
\item 
$T(x_i)=\textit{true}$ if $\{l_i^{1}m_i^{1},l_i^{1}p_i^{2},l_i^{2}p_i^{1},l_i^{2}m_i^{2}\}\cap S=\emptyset$, and 
\item $T(x_i)=\textit{false}$ if $|\{l_i^{1}m_i^{1},l_i^{1}p_i^{2},l_i^{2}p_i^{1},l_i^{2}m_i^{2}\}\cap S|\geq 1$.
\end{itemize}
\smallskip

\item[$\bullet$] If $H(x_i)$ complies with 3, in order to cover vertices $m_i^{1}, m_i^{2}$, $S$ contains exactly one of edges in the pairs $(r_im_i^{1},l_i^{1}m_i^{1}),(r_im_i^{2},l_i^{2}m_i^{2})$. This means that we set 
\begin{itemize}
\item 
$T(x_i)=\textit{true}$ if $S\cap \{l_i^{1}m_i^{1},l_i^{2}m_i^{2}\}=\emptyset$ and 
\item $T(x_i)=\textit{false}$ if $|S\cap \{l_i^{1}m_i^{1},l_i^{2}m_i^{2}\}|\geq 1$.
\end{itemize}
\end{itemize}

We obtain a valid assignment $T$. Since $S$ covers all vertices of~$\mathcal{C}$, $T$ satisfies all clauses of $I$.
\end{proof}

\bigskip
In the previous construction, we started from a planar graph $G$ and made a new graph $H$ by replacing each vertex $x_i\in V$ with one of the three different variable gadgets $H(x_i)$ which are depicted in Figure~\ref{Fig:Ext planar EC}. 

\smallskip

In the following we will introduce three new constructions in order to prove $\np$-hardness of \textsc{Ext DS}, \textsc{Ext EDS} and \textsc{Ext EM} for subcubic planar bipartite graphs. All of the constructions include variable and clause gadgets, the clause gadgets are very similar to what we proposed for these problems before but in order to maintain planarity we introduce three different variable gadgets.

\begin{theorem}\label{thm: EXT DS planar}
\textsc{Ext DS} is $\np$-hard for planar bipartite graphs of maximum degree 3.
\end{theorem}

\begin{proof}
The proof is based on a reduction from \textsc{4P3C3SAT}. For an instance $I$ of \textsc{4P3C3SAT} with clause set $\mathcal{C}=\{c_1,\dots,c_m\}$ and variable set $\mathcal{X}=\{x_1,\dots,x_n\}$, we build a planar bipartite graph $H=(V_H,E_H)$ with maximum degree 3 together with a set $U\subseteq V_H$ of forced vertices as an instance of \textsc{Ext DS}.

\bigskip

For each variable $x_i$, similar to what we did in Theorem \ref{thm: EXT planar EC}, we propose 3 different gadgets $H(x_i)$. As is depicted in Fig.~\ref{Fig:Ext DS}, the forced vertex set $U_{x_i}$ corresponding to gadget $H(x_i)$ contains $m_i$ for case 1, $\{p_i^{1},p_i^{2},m_i^{1},m_i^{2}\}$ for case 2  and $\{p_i^{1},p_i^{2},p_i^{3}\}$ for case 3.

\bigskip

For each clause $c_j \in \mathcal{C}$, we consider a clause gadget $H(c_j)$ together with a forced vertex set $U_{c_j}$ completely similar to what is defined before in Theorem~\ref{Bip_Ext_DS} and illustrated in Fig.~\ref{fig:degree3_Ext_DS}. Moreover we interconnect with some crossing edges, the subgraphs $H(x_i)$ and $H(c_j)$ using the proposed instructions in Theorem \ref{Bip_Ext_DS}. We also set the forced vertex set $U=(\bigcup_{x_i\in \mathcal{X}}U_{x_i})\cup (\bigcup_{c_j\in \mathcal{C}}U_{c_j})$.

\bigskip

This construction computes in polynomial time, a planar bipartite graph with maximum degree 3. We now claim that $(H,U)$ is a \yes-instance of \textsc{Ext DS} iff $I$ has a satisfying assignment~$T$.

\bigskip

\begin{figure}
\centering
\begin{tikzpicture}[]
\tikzstyle{vertex}=[circle, draw, inner sep=2pt,  minimum width=1 pt, minimum size=0.1cm]
\tikzstyle{vertex1}=[circle, draw, inner sep=2pt, fill=black!100, minimum width=1pt, minimum size=0.1cm]


\node[vertex] (11) at (-2,-0.5) {};
\node[vertex1] (12) at (-1.3,-1.5) {};
\node[vertex] (13) at (-2,-2.5) {};


\draw (11)--(12)--(13);
\draw [dashed] (-2.8,0)--(11);
\draw [dashed] (-2.8,-1)--(11);
\draw [dashed] (-2.8,-2)--(13);
\draw [dashed] (-2.8,-3)--(13);
\node () at (-3.4,0) {$H(c_1)$};
\node () at (-3.4,-1) {$H(c_2)$};
\node () at (-3.4,-2) {$H(c_3)$};
\node () at (-3.4,-3) {$H(c_4)$};
\node () at (-2,-0.2) {$t_i$};
\node () at (-2,-2.8) {$f_i$};
\node () at (-0.95,-1.5) {$m_i$};
\node () at (-2.5,-4.5) {case 1};


\node[vertex] (21) at (2.5,0.5) {};
\node[vertex1] (22) at (3.1,0) {};
\node[vertex] (23) at (2.5,-0.5) {};
\node[vertex] (24) at (3.1,-1) {};
\node[vertex1] (25) at (2.5,-1.5) {};
\node[vertex] (26) at (3.1,-2) {};
\node[vertex] (27) at (2.5,-2.5) {};
\node[vertex1] (28) at (3.1,-3) {};
\node[vertex] (29) at (2.5,-3.5) {};
\node[vertex] (210) at (4,-3.3) {};
\node[vertex1] (211) at (4,-1.5) {};
\node[vertex] (212) at (4,0.3) {};


\node () at (0.95,0.5) {$H(c_1)$};
\node () at (0.95,-0.5) {$H(c_2)$};
\node () at (0.95,-2.5) {$H(c_3)$};
\node () at (0.95,-3.5) {$H(c_4)$};
\node () at (2.5,0.8) {$t_i^{1}$};
\node () at (2.5,-0.2) {$f_i^{1}$};
\node () at (3.4,0) {$p_i^{1}$};
\node () at (3.4,-1) {$l_i^{1}$};
\node () at (3.4,-2) {$r_i^{1}$};
\node () at (3.4,-3) {$p_i^{2}$};
\node () at (2.1,-1.5) {$m_i^{1}$};
\node () at (2.5,-2.8) {$t_i^{2}$};
\node () at (2.5,-3.8) {$f_i^{2}$};
\node () at (4.3,0.3) {$l_i^{2}$};
\node () at (4.35,-1.5) {$m_i^{2}$};
\node () at (4.3,-3.3) {$r_i^{2}$};
\node () at (2.5,-4.5) {case 2};

\draw (21)--(22)--(23)--(24)--(25)--(26)--(27)--(28)--(29)--(210)--(211)--(212)--(21);

\draw [dashed](21)--(1.5,0.5);
\draw [dashed](23)--(1.5,-0.5);
\draw [dashed](27)--(1.5,-2.5);
\draw [dashed](29)--(1.5,-3.5);


\node[vertex] (31) at (8,0.5) {};
\node[vertex1] (32) at (8.6,0.5) {};
\node[vertex] (33) at (9.2,0.5) {};
\node[vertex] (34) at (9.2,-0.5) {};
\node[vertex] (35) at (9.2,-1.5) {};
\node[vertex1] (36) at (8.6,-1.5) {};
\node[vertex] (37) at (8,-1.5) {};
\node[vertex] (38) at (9.2,-2.5) {};
\node[vertex] (39) at (9.2,-3.5) {};
\node[vertex1] (310) at (8.6,-3.5) {};
\node[vertex] (311) at (8,-3.5) {};

\node () at (8,0.8) {$t_i^{1}$};
\node () at (9.2,0.8) {$l_i^{2}$};
\node () at (8,-3.2) {$t_i^{2}$};
\node () at (9.5,-3.5) {$r_i^{2}$};
\node () at (8,-1.2) {$f_i$};
\node () at (8.6,0.8) {$p_i^{1}$};
\node () at (9.5,-0.5) {$l_i^{1}$};
\node () at (8.6,-1.2) {$p_i^{2}$};
\node () at (9.5,-2.5) {$r_i^{1}$};
\node () at (8.6,-3.2) {$p_i^{3}$};
\node () at (9.55,-1.5) {$m_i$};

\draw(31)--(32)--(33)--(34)--(35)--(36)--(37);
\draw(35)--(38)--(39)--(310)--(311);

\node () at (6.4,1) {$H(c_2)$};
\node () at (6.4,0) {$H(c_3)$};
\node () at (6.4,-1.5) {$H(c_4)$};
\node () at (6.4,-3.5) {$H(c_1)$};
\node () at (8,-4.5) {case 3};

\draw [dashed](31)--(6.95,1);
\draw [dashed](31)--(6.95,0);
\draw [dashed](37)--(6.95,-1.5);
\draw [dashed](311)--(6.95,-3.5);
\end{tikzpicture}
\caption{Variable gadgets $H(x_i)$ of Theorem~\ref{thm: EXT DS planar}. Cases 1, 2, 3 are corresponding to $H(x_i)$, depending on how $x_i$ appears (negated or non-negated) in the four clauses (Here case 3 is rotated). Black vertices denote elements of~$U_{x_i}$. Crossing edges are marked with dashed lines.}\label{Fig:Ext DS}
\end{figure}
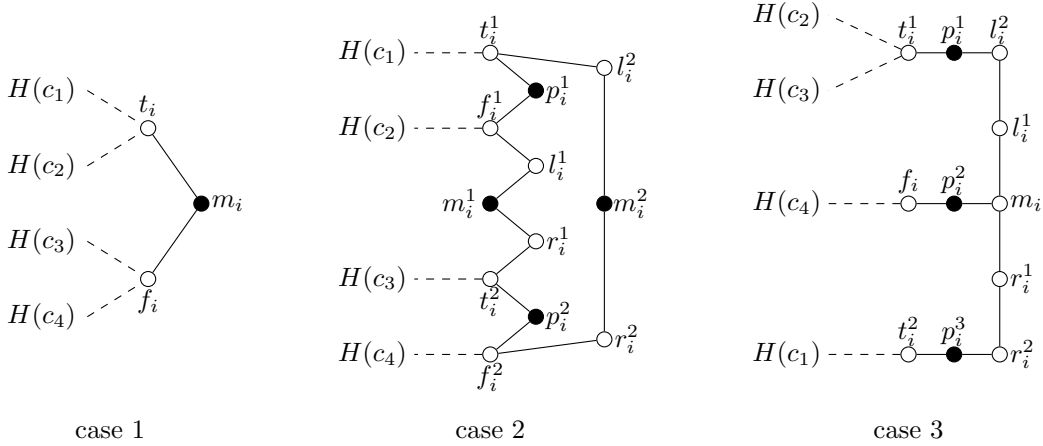

Suppose $T$ is a truth assignment of $I$ which satisfies all clauses. We construct a minimal dominating set $S$ from $U$ as follows:

\begin{itemize}
\item[$\bullet$] For each variable gadget $H(x_i)$ complies with "case 1",  add $t_i$ (resp., $f_i$) to $S$ if $T(x_i)=\textit{true}$ (resp., $T(x_i)=\textit{true}$).

\smallskip

\item[$\bullet$] For each variable gadget $H(x_i)$ complies with "case 2" add $t_i^{1},t_i^{2}$ (resp., $f_i^{1},f_i^{2}$) to $S$ if $T(x_i)=\textit{true}$ (resp., $T(x_i)=\textit{true}$).
 
 \smallskip
 
\item[$\bullet$] For each variable gadget $H(x_i)$ complies with "case 3" add $t_i^{1},t_i^{2},m_i$ (resp., $f_i,l_i^{2},r_i^{2}$) to $S$ if $T(x_i)=\textit{true}$ (resp., $T(x_i)=\textit{true}$).

\smallskip

\item[$\bullet$] For each clause $c\in \mathcal{C}$, add vertex $1_c$ to $S$ if $1'_c$ is not dominated by a variable vertex of $S$ and add $2_c$ to $S$ if $2'_c$ is not dominated by a variable vertex of $S$.
\end{itemize}

Since $T$ is satisfiable, for each clause gadget $H(c)$, at least one of $1'_c,2'_c$ is dominated by a variable vertex of $S$. Thus, $S$ is a dominating set that contains $U$ of $H$. If $S$ is not minimal, it can hence be turned into a minimal dominating set $S'$ by successively removing vertices without private neighbors from the ones that we added to $U$.

\bigskip

Conversely, suppose $S$ is a minimal dominating set  of $H$ with $U\subseteq S$. We show that in Theorem \ref{Bip_Ext_DS}, for each clause gadget $H(c)$ at least one of $1'_c,2'_c$ needs to be dominated by a variable gadget vertex thus there is an assignment $T$ which satisfies all clauses of $I$. We now show that $T$ is a valid assignment, in order to this, we consider all variable gadgets independently:

\begin{itemize}
\item[$\bullet$] If $H(x_i)$ complies with case 1, by minimality, $S$ cannot contain both $t_i,f_i$, So we set $T(x_i)=\textit{true}$ if $\{f_i\}\cap S=\emptyset$ and otherwise we set $T(x_i)=\textit{false}$.

\smallskip

\item[$\bullet$] If $H(x_i)$ complies with case 2, by minimality, $S$ cannot contain both vertices in each pair $(t_i^{1},f_i^{1}),(t_i^{1},f_i^{2}),(t_i^{2},f_i^{1}),(t_i^{2},f_i^{2})$, so we set $T(x_i)=\textit{true}$ if $S\cap \{f_i^{1},f_i^{2}\}=\emptyset$ and otherwise we set $T(x_i)=\textit{false}$.

\smallskip

\item[$\bullet$] If $H(x_i)$ complies with case 3, by minimality, $S$ cannot contain both vertices in each pair $(t_i^{1},f_i), (t_i^{2},f_i)$, hence we set $T(x_i)=\textit{true}$ if $S\cap \{f_i\}=\emptyset$ and otherwise, we set $T(x_i)=\textit{false}$.
\end{itemize}
\end{proof}

\bigskip

\begin{theorem}\label{thm: EXT EDS planar}
\textsc{Ext EDS} is $\np$-hard for planar bipartite graphs of maximum degree 3.
\end{theorem}

\begin{proof}
The proof is similar to the proof of Theorem \ref{thm: EXT DS planar}. For an instance $I$ of \textsc{4P3C3SAT} with clause set $\mathcal{C}=\{c_1,\dots,c_m\}$ and variable set $\mathcal{X}=\{x_1,\dots,x_n\}$, we build a planar bipartite graph $H=(V_H,E_H)$ with maximum degree 3 together with a set $U\subseteq E_H$ of forced edges as an instance of \textsc{Ext EDS}. 

\bigskip

For each variable $x_i$ we propose 3 different gadgets $H(x_i)$, which is depicted in Fig.~\ref{Fig:Ext EDS planar}, the forced edge set $U_{x_i}$ corresponding to gadget $H(x_i)$ contains $\{m_ir_i,r_ip_i\}$ for case 1, $\{p_i^{j}r_i^{j},r_i^{j}m_i^{j}\colon 1\leq j\leq 4\}$ for case 2  and $\{p_i^{1}p_i^{2}, p_i^{2}p_i^{3},p_i^{5}p_i^{6}, p_i^{6}p_i^{7},m_i^{2}f_i\}$ for case 3.

\begin{figure}
\centering
\begin{tikzpicture}[scale=0.87, transform shape]
\tikzstyle{vertex}=[circle, draw, inner sep=2pt,  minimum width=1 pt, minimum size=0.1cm]
\tikzstyle{vertex1}=[circle, draw, inner sep=2pt, fill=black!100, minimum width=1pt, minimum size=0.1cm]

\node[vertex] (11) at (-3.6,0) {};
\node[vertex] (12) at (-2.7,0.6) {};
\node[vertex,above of=12,node distance=0.6cm](13){};
\node[vertex,above of=13,node distance=0.6cm](14){};
\node[vertex,above of=14,node distance=0.6cm](15){};
\node[vertex] (16) at (-1.8,0) {};
\node () at (-3,2.4) {$l_i$};
\node () at (-3.05,1.8) {$m_i$};
\node () at (-3,1.2) {$r_i$};
\node () at (-3,0.6) {$p_i$};
\node () at (-3.85,0) {$t_i$};
\node () at (-1.5,0) {$f_i$};
\node () at (-4.3,-1.3) {$H(c_1)$};
\node () at (-3.2,-1.3) {$H(c_2)$};
\node () at (-2.25,-1.3) {$H(c_3)$};
\node () at (-1.2,-1.3) {$H(c_4)$};
\node () at (-1.7,-2) {case 1};

\draw (11)--(12)--(16);
\draw (12) edge [ultra thick](13);
\draw (13) edge [ultra thick](14);
\draw (14)--(15);
\draw [dashed](-4.1,-1)--(11)--(-3.1,-1);
\draw [dashed](-2.3,-1)--(16)--(-1.3,-1);

\node[vertex] (21) at (0,0) {};
\node[vertex] (22) at (0.6,0.6) {};
\node[vertex,above of=22,node distance=0.6cm](23){};
\node[vertex,above of=23,node distance=0.6cm](24){};
\node[vertex,above of=24,node distance=0.6cm](25){};

\node () at (-0.3,0) {$t_i^{1}$};
\node () at (0.3,0.6) {$p_i^{1}$};
\node () at (0.3,1.2) {$r_i^{1}$};
\node () at (0.3,1.8) {$m_i^{1}$};
\node () at (0.35,2.4) {$l_i^{1}$};

\node[vertex] (26) at (1.2,0) {};
\node[vertex] (27) at (1.8,0.6) {};
\node[vertex,above of=27,node distance=0.6cm](28){};
\node[vertex,above of=28,node distance=0.6cm](29){};
\node[vertex,above of=29,node distance=0.6cm](210){};

\node () at (0.85,0) {$f_i^{1}$};
\node () at (1.5,0.6) {$p_i^{2}$};
\node () at (1.5,1.2) {$r_i^{2}$};
\node () at (1.5,1.8) {$m_i^{2}$};
\node () at (1.5,2.4) {$l_i^{2}$};

\node[vertex] (211) at (2.4,0) {};
\node[vertex] (212) at (3,0.6) {};
\node[vertex,above of=212,node distance=0.6cm](213){};
\node[vertex,above of=213,node distance=0.6cm](214){};
\node[vertex,above of=214,node distance=0.6cm](215){};
\node[vertex] (216) at (3.6,0) {};

\node () at (2.1,0) {$t_i^{2}$};
\node () at (2.7,0.6) {$p_i^{3}$};
\node () at (2.7,1.2) {$r_i^{3}$};
\node () at (2.7,1.8) {$m_i^{3}$};
\node () at (2.7,2.4) {$l_i^{3}$};
\node () at (3.25,0) {$f_i^{2}$};
\node[vertex] (217) at (1.8,4) {};
\node[vertex,above of=217,node distance=0.6cm](218){};
\node[vertex,above of=218,node distance=0.6cm](219){};
\node[vertex,above of=219,node distance=0.6cm](220){};

\node () at (1.5,4.1) {$p_i^{4}$};
\node () at (1.5,4.6) {$r_i^{4}$};
\node () at (1.45,5.2) {$m_i^{4}$};
\node () at (1.5,5.8) {$l_i^{4}$};


\node () at (0,-1.3) {$H(c_1)$};
\node () at (1.2,-1.3) {$H(c_2)$};
\node () at (2.4,-1.3) {$H(c_3)$};
\node () at (3.6,-1.3) {$H(c_4)$};

\node () at (1.8,-2) {case 2};
\draw (21)--(22)--(26)--(27)--(211)--(212)--(216);
\draw (22) edge [ultra thick](23);
\draw (23) edge [ultra thick](24);
\draw (24)--(25);

\draw (27) edge [ultra thick](28);
\draw (28) edge [ultra thick](29);
\draw (29)--(210);

\draw (212) edge [ultra thick](213);
\draw (213) edge [ultra thick](214);
\draw (214)--(215);

\draw (217) edge [ultra thick](218);
\draw (218) edge [ultra thick](219);
\draw (219)--(220);
\draw (216) edge[bend right=40] (217);
\draw (21) edge[bend left=40] (217);

\draw [dashed](21)--(0,-1);
\draw [dashed](26)--(1.2,-1);
\draw [dashed](211)--(2.4,-1);
\draw [dashed](216)--(3.6,-1);


\node[vertex] (31) at (5.5,0) {};
\node[vertex] (32) at (6.1,0.6) {};
\node[vertex,above of=32,node distance=0.6cm](33){};
\node[vertex,above of=33,node distance=0.6cm](34){};
\node[vertex,above of=34,node distance=0.6cm](316){};

\node () at (5.2,0) {$t_i^{1}$};
\node () at (5.8,0.6) {$p_i^{1}$};
\node () at (5.8,1.2) {$p_i^{2}$};
\node () at (5.8,1.8) {$p_i^{3}$};
\node () at (5.8,2.4) {$p_i^{4}$};

\node[vertex] (35) at (6.7,0) {};
\node[vertex] (36) at (7.3,0.6) {};
\node[vertex] (37) at (7.9,1.2) {};
\node[vertex] (38) at (7.9,0.6) {};
\node[vertex] (39) at (7.9,0) {};
\node[vertex] (310) at (8.5,0.6) {};
\node[vertex] (311) at (9.1,0) {};

\node () at (6.4,0) {$l_i^{2}$};
\node () at (7.05,0.6) {$l_i^{1}$};

\node[vertex] (312) at (9.7,0.6) {};
\node[vertex,above of=312,node distance=0.6cm](313){};
\node[vertex,above of=313,node distance=0.6cm](314){};
\node[vertex,above of=314,node distance=0.6cm](317){};
\node[vertex] (315) at (10.3,0) {};

\node () at (10.6,0) {$t_i^{2}$};
\node () at (10,0.6) {$p_i^{5}$};
\node () at (10,1.2) {$p_i^{6}$};
\node () at (10,1.8) {$p_i^{7}$};
\node () at (10,2.4) {$p_i^{8}$};
\node () at (9.4,0) {$r_i^{2}$};
\node () at (8.8,0.6) {$r_i^{1}$};
\node () at (7.9,1.5) {$m_i^{1}$};
\node () at (7.6,0.6) {$m_i^{2}$};
\node () at (7.6,0) {$f_i$};

\node () at (5,-1.3) {$H(c_2)$};
\node () at (6.2,-1.3) {$H(c_3)$};
\node () at (10.3,-1.3) {$H(c_1)$};
\node () at (7.9,-1.3) {$H(c_4)$};
\node () at (8,-2) {case 3};

\draw(31)--(32)--(35)--(36)--(37)--(310)--(311)--(312)--(315);
\draw(32) edge [ultra thick](33);
\draw(33) edge [ultra thick](34);
\draw (34)--(316);

\draw(312) edge [ultra thick](313);
\draw(313) edge [ultra thick](314);
\draw (314)--(317);
\draw(37)--(38);
\draw(38) edge [ultra thick](39);


\draw [dashed](31)--(5,-1);
\draw [dashed](31)--(6,-1);
\draw [dashed](39)--(7.9,-1);
\draw [dashed](315)--(10.3,-1);
\end{tikzpicture}
\caption{Variable gadgets $H(x_i)$ of Theorem~\ref{thm: EXT EDS planar}. Cases 1, 2, 3 are corresponding to $H(x_i)$, depending on how $x_i$ appears (negated or non-negated) in the four clauses (Here case 3 is rotated). Bold edges denote elements of~$U_{x_i}$. Crossing edges are marked by dashed lines.}\label{Fig:Ext EDS planar}
\end{figure}
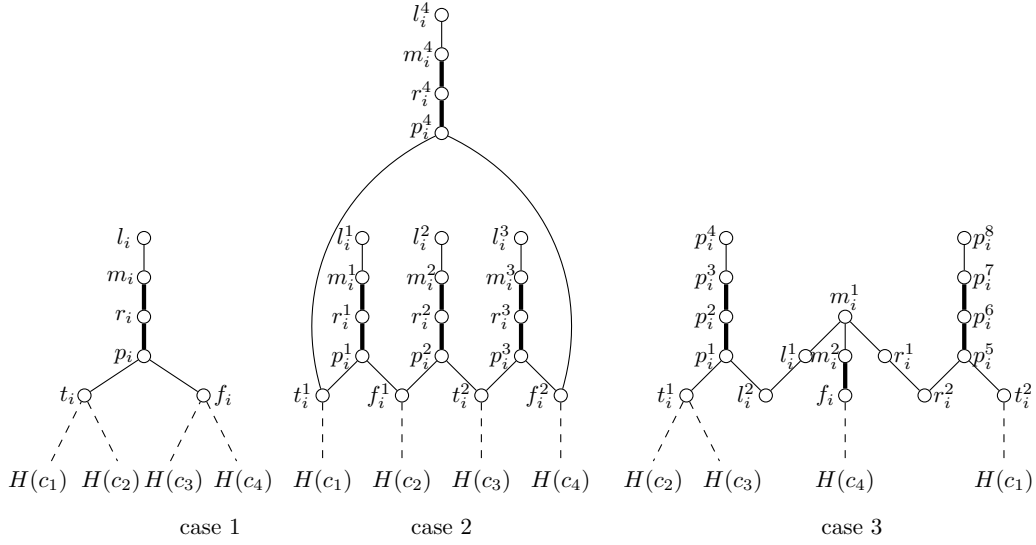

\bigskip

For each clause $c_j\in \mathcal{C}$, we consider a clause gadget $H(c_j)$ and a forced edge set $U_{c_j}$ similar to what we did in Theorem \ref{Bip_Ext_EDS}, each clause gadget $H(c_j)$ contains 8 vertices and 7 edges where $|U_{c_j}|=2$ (see Fig. \ref{fig:degree3_Ext_EDS}). Moreover we interconnect with some crossing edges the subgraphs $H(x_i)$ and $H(c_j)$ by the proposed instructions in Theorem \ref{Bip_Ext_EDS}. We also set the forced edge set $U=(\bigcup_{x_i\in \mathcal{X}}U_{x_i})\cup (\bigcup_{c_j\in \mathcal{C}}U_{c_j})$.

\bigskip

This construction computes in polynomial time, a planar bipartite graph with maximum degree 3. We now claim that $T$ is a satisfying assignment of $I$ iff $H$ has a minimal edge dominating set containing $U$.
\end{proof}

\bigskip

For the last proof in this section we will show  $\np$-hardness of \textsc{Ext EM}.

\begin{theorem}\label{thm: EXT EM planar}
\textsc{Ext EM} is $\np$-hard for planar bipartite graphs of maximum degree 3.
\end{theorem}

\begin{proof}
The proof is based on a reduction from \textsc{4P3C3SAT}. For an instance $I$ of \textsc{4P3C3SAT} with clause set $\mathcal{C}=\{c_1,\dots,c_m\}$ and variable set $\mathcal{X}=\{x_1,\dots,x_n\}$, we build a planar bipartite graph $H=(V_H,E_H)$ with maximum degree 3 together with a set $U\subseteq E_H$ of permitted edges as an instance of \textsc{Ext EM} as follows.

\begin{itemize}
\item[$\bullet$] For each clause $c_j$, we introduce a clause gadget $H(c_j)$ together with a permitted edge set $U_{c_j}$ which is already explained in detail in Theorem \ref{degree3_Ext_EM}.

\smallskip

\item[$\bullet$] For each variable $x_i$ depending on how $x_i$ appears (negated or non-negated) in clauses, we introduce 3 different gadgets $H(x_i)$ together with a set of permitted edges $U_{x_i}$ which is depicted in Fig.~\ref{Fig:Ext EM planar}.

\smallskip

\item[$\bullet$] We also interconnect $H(x_i)$ to $H(c_j)$ where $x_i$ appears positively or negatively in clause~$c_j$ by some crossing edges which explained before in Theorem~\ref{degree3_Ext_EM}. Let $U_{cross}$ be the set of all crossing edges. 
\end{itemize}

\begin{figure}
\centering
\begin{tikzpicture}[scale=0.77, transform shape]
\tikzstyle{vertex}=[circle, draw, inner sep=2pt,  minimum width=1 pt, minimum size=0.1cm]
\tikzstyle{vertex1}=[circle, draw, inner sep=2pt, fill=black!100, minimum width=1pt, minimum size=0.1cm]

\node[vertex] (1g1) at (-4,-6) {};
\node[vertex,left of=1g1,node distance=0.6cm](2g1){};
\node[vertex] (3g1) at (-5.2,-5) {};
\node[vertex] (4g1) at (-5.2,-7) {};
\node[vertex] (5g1) at (-5.8,-4.5) {};
\node[vertex] (6g1) at (-5.8,-5.5) {};
\node[vertex] (7g1) at (-5.8,-6.5) {};
\node[vertex] (8g1) at (-5.8,-7.5) {};
\node[vertex] (9g1) at (-6.5,-4.5) {};
\node[vertex] (10g1) at (-6.5,-5.5) {};
\node[vertex] (11g1) at (-6.5,-6.5) {};
\node[vertex] (12g1) at (-6.5,-7.5) {};

\node () at (-5.8,-4.2) {$g_i^{1}$};
\node () at (-6.5,-4.2) {$t_i^{1}$};
\node () at (-5.8,-5.2) {$g_i^{2}$};
\node () at (-6.5,-5.2) {$t_i^{2}$};
\node () at (-5.8,-6.2) {$h_i^{1}$};
\node () at (-6.5,-6.2) {$f_i^{1}$};
\node () at (-5.8,-7.2) {$h_i^{2}$};
\node () at (-6.5,-7.2) {$f_i^{2}$};
\node () at (-5.2,-4.7) {$l_i^{1}$};
\node () at (-5.2,-7.3) {$l_i^{2}$};
\node () at (-4.5,-5.75) {$m_i$};
\node () at (-4,-5.75) {$r_i$};

\draw (2g1)--(3g1);
\draw (2g1)--(4g1);
\draw (5g1)--(9g1);
\draw (6g1)--(10g1);
\draw (7g1)--(11g1);
\draw (8g1)--(12g1);

\draw (1g1) edge [ultra thick](2g1);
\draw (3g1) edge [ultra thick](5g1);
\draw (3g1) edge [ultra thick](6g1);
\draw (4g1) edge [ultra thick](7g1);
\draw (4g1) edge [ultra thick](8g1);
\draw [dashed](9g1)--(-7.4,-4.5);
\draw [dashed](10g1)--(-7.4,-5.5);
\draw [dashed](11g1)--(-7.4,-6.5);
\draw [dashed](12g1)--(-7.4,-7.5);
\node () at (-7.9,-4.5) {$H(c_1)$};
\node () at (-7.9,-5.5) {$H(c_2)$};
\node () at (-7.9,-6.5) {$H(c_3)$};
\node () at (-7.9,-7.5) {$H(c_4)$};
\node () at (-6.5,-9.5) {case 1};


\node[vertex] (1g2) at (2,-4) {};
\node[vertex] (2g2) at (2,-7.8) {};
\node[vertex] (3g2) at (1.3,-4) {};
\node[vertex] (4g2) at (1.3,-7.8) {};
\node[vertex] (5g2) at (0.8,-3.5) {};
\node[vertex] (6g2) at (0.8,-4.5) {};
\node[vertex] (7g2) at (0.8,-5.2) {};
\node[vertex] (8g2) at (0.8,-5.9) {};
\node[vertex] (9g2) at (0.8,-6.6) {};
\node[vertex] (10g2) at (0.8,-7.3) {};
\node[vertex] (11g2) at (0.8,-8.3) {};
\node[vertex] (12g2) at (0.1,-3.5) {};
\node[vertex] (13g2) at (0.1,-4.5) {};
\node[vertex] (14g2) at (0.1,-7.3) {};
\node[vertex] (15g2) at (0.1,-8.3) {};
\node[vertex] (16g2) at (-0.8,-3.5) {};
\node[vertex] (17g2) at (-0.8,-4.5) {};
\node[vertex] (18g2) at (-0.8,-7.3) {};
\node[vertex] (19g2) at (-0.8,-8.3) {};
\node[vertex] (20g2) at (2.8,-4.2) {};
\node[vertex] (21g2) at (2.8,-5.9) {};
\node[vertex] (22g2) at (2.8,-7.6) {};

\node () at (0.8,-3.2) {$p_i^{1}$};
\node () at (1.1,-4.5) {$p_i^{2}$};
\node () at (1.4,-3.7) {$p_i^{3}$};
\node () at (2.3,-4.05) {$p_i^{4}$};
\node () at (1.1,-7.25) {$p_i^{5}$};
\node () at (0.8,-8.6) {$p_i^{6}$};
\node () at (1.4,-8.05) {$p_i^{7}$};
\node () at (2.3,-7.75) {$p_i^{8}$};

\node () at (0.1,-3.2) {$g_i^{1}$};
\node () at (-0.8,-3.2) {$t_i^{1}$};
\node () at (0.1,-4.8) {$h_i^{1}$};
\node () at (-0.8,-4.8) {$f_i^{1}$};
\node () at (0.1,-7) {$g_i^{2}$};
\node () at (-0.8,-7) {$t_i^{2}$};
\node () at (0.1,-8.65) {$h_i^{2}$};
\node () at (-0.8,-8.65) {$f_i^{2}$};

\node () at (1.1,-5.2) {$l_i^{1}$};
\node () at (1.15,-5.9) {$m_i^{1}$};
\node () at (1.1,-6.6) {$r_i^{1}$};

\node () at (3.1,-4.2) {$l_i^{2}$};
\node () at (3.15,-5.9) {$m_i^{2}$};
\node () at (3.1,-7.7) {$r_i^{2}$};

\draw (3g2)--(5g2);
\draw (3g2)--(6g2);
\draw (4g2)--(10g2);
\draw (4g2)--(11g2);
\draw (12g2)--(16g2);
\draw (13g2)--(17g2);
\draw (14g2)--(18g2);
\draw (15g2)--(19g2);
\draw (20g2)--(21g2)--(22g2);
\draw (7g2)--(8g2)--(9g2);

\draw (1g2) edge [ultra thick](3g2);
\draw (2g2) edge [ultra thick](4g2);
\draw (5g2) edge [ultra thick](12g2);
\draw (6g2) edge [ultra thick](13g2);
\draw (6g2) edge [ultra thick](7g2);
\draw (9g2) edge [ultra thick](10g2);
\draw (10g2) edge [ultra thick](14g2);
\draw (1g2) edge [ultra thick](3g2);
\draw (11g2) edge [ultra thick](15g2);
\draw (5g2) edge[bend left=30, ultra thick] (20g2);
\draw (11g2) edge[bend right=30, ultra thick] (22g2);

\draw [dashed](16g2)--(-2,-3.5);
\draw [dashed](17g2)--(-2,-4.5);
\draw [dashed](18g2)--(-2,-7.3);
\draw [dashed](19g2)--(-2,-8.3);

\node () at (-2.5,-3.5) {$H(c_1)$};
\node () at (-2.5,-4.5) {$H(c_2)$};
\node () at (-2.5,-7.3) {$H(c_3)$};
\node () at (-2.5,-8.3) {$H(c_4)$};

\node () at (0.2,-9.5) {case 2};


\node[vertex] (1g3) at (8.5,-7.8) {};
\node[vertex] (2g3) at (7.8,-4.4) {};
\node[vertex] (3g3) at (7.8,-6.1) {};
\node[vertex] (4g3) at (7.8,-7.8) {};
\node[vertex] (5g3) at (7.3,-3.9) {};
\node[vertex] (6g3) at (7.3,-4.9) {};
\node[vertex] (7g3) at (7.3,-5.6) {};
\node[vertex] (8g3) at (7.3,-6.6) {};
\node[vertex] (9g3) at (7.3,-7.3) {};
\node[vertex] (10g3) at (7.3,-8.3) {};
\node[vertex] (11g3) at (7.3,-3.4) {};
\node[vertex] (12g3) at (6.8,-4.4) {};
\node[vertex] (13g3) at (6.8,-7.3) {};
\node[vertex] (14g3) at (6.8,-8.3) {};
\node[vertex] (15g3) at (6.2,-3.4) {};
\node[vertex] (16g3) at (6.2,-4.4) {};
\node[vertex] (17g3) at (6.2,-7.3) {};
\node[vertex] (18g3) at (6.2,-8.3) {};

\node () at (7.3,-3.1) {$g_i^{1}$};
\node () at (6.2,-3.1) {$t_i^{1}$};
\node () at (7.56,-3.78) {$l_i^{1}$};
\node () at (8.15,-4.4) {$m_i^{1}$};
\node () at (7.58,-4.9) {$r_i^{1}$};

\node () at (7.58,-5.6) {$l_i^{2}$};
\node () at (8.15,-6.1) {$m_i^{2}$};
\node () at (7.55,-6.7) {$r_i^{2}$};

\node () at (7.58,-7.25) {$l_i^{3}$};
\node () at (8.05,-7.6) {$m_i^{3}$};
\node () at (7.6,-8.3) {$r_i^{3}$};

\node () at (8.8,-7.8) {$p_i$};
\node () at (6.85,-4.67) {$g_i^{2}$};
\node () at (6.2,-4.67) {$t_i^{2}$};
\node () at (6.85,-7.02) {$g_i^{3}$};
\node () at (6.2,-7.02) {$t_i^{3}$};
\node () at (6.85,-8.6) {$h_i$};
\node () at (6.2,-8.6) {$f_i$};

\draw (5g3)--(2g3)--(6g3);
\draw (7g3)--(3g3)--(8g3);
\draw (9g3)--(4g3)--(10g3);
\draw (11g3)--(15g3);
\draw (12g3)--(16g3);
\draw (13g3)--(17g3);
\draw (14g3)--(18g3);
\draw (11g3) edge [ultra thick](5g3);
\draw (12g3) edge [ultra thick](5g3);
\draw (6g3) edge [ultra thick](7g3);
\draw (8g3) edge [ultra thick](9g3);
\draw (13g3) edge [ultra thick](9g3);
\draw (1g3) edge [ultra thick](4g3);
\draw (10g3) edge [ultra thick](14g3);

\draw [dashed](15g3)--(5.2,-3.4);
\draw [dashed](16g3)--(5.2,-4.4);
\draw [dashed](17g3)--(5.2,-7.3);
\draw [dashed](18g3)--(5.2,-8.3);
\node () at (4.7,-3.4) {$H(c_1)$};
\node () at (4.7,-4.4) {$H(c_2)$};
\node () at (4.7,-7.3) {$H(c_3)$};
\node () at (4.7,-8.3) {$H(c_4)$};

\node () at (6.8,-9.5) {case 3};
\end{tikzpicture}
\caption{Variable gadgets $H(x_i)$ of Theorem~\ref{thm: EXT EM planar}. Cases 1, 2, 3 are corresponding to $H(x_i)$, depending on how $x_i$ appears (negated or non-negated) in the four clauses. Edges not in~$U_{x_i}$ are drawn bold. Crossing edges are marked with dashed lines.}\label{Fig:Ext EM planar}
\end{figure}
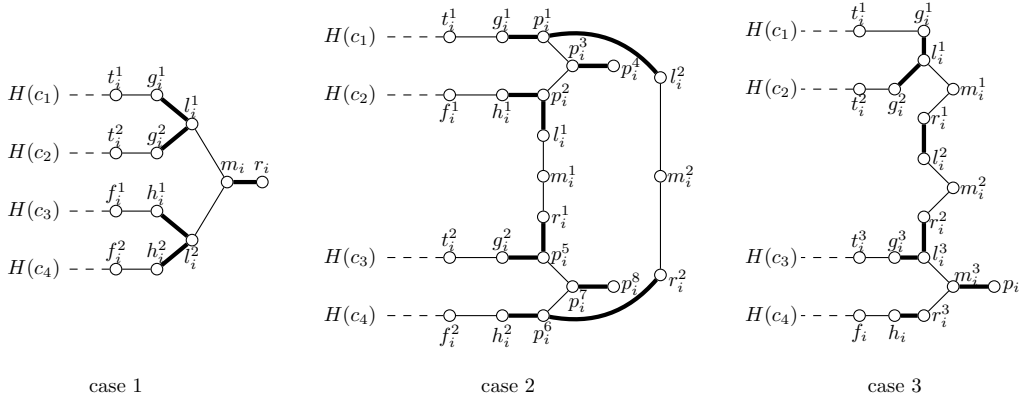

Let $U=(\bigcup_{c_j\in \mathcal{C}}U_{c_j})\cup (\bigcup_{x_i\in \mathcal{X}}U_{x_i})\cup U_{cross}$. This construction computes in polynomial time a planar bipartite graph of maximum degree 3. We now claim that $(H,U)$ is a  \yes-instance of \textsc{Ext EM} iff $T$ is a satisfying assignment of $I$.

\bigskip

Suppose $T$ is a truth assignment of $I$ which satisfies all clauses. We produce a maximal edge matching $S\subseteq U$as follows: the method of choosing edges from clause gadgets and crossing edges is already explained in Theorem \ref{degree3_Ext_EM}, so we here just show which edges of each $H(x_i)$ should be in $S$:
\begin{itemize}
\item[$\bullet$] If $H(x_i)$ complies with case 1, we add $\{m_il_i^{1},h_i^{1}f_i^{1},h_i^{2}f_i^{2}\}$ (resp., $\{m_il_i^{2},g_i^{1}t_i^{1},g_i^{2}t_i^{2}\}$) if $T(x_i)=\textit{true}$ (resp., $T(x_i)=\textit{false}$),

\smallskip 

\item[$\bullet$] If $H(x_i)$ complies with case 2, we add $\{f_i^{1}h_i^{1}, p_i^{1}p_i^{3}, m_i^{2}r_i^{2} ,f_i^{2}h_i^{2}, p_i^{5}p_i^{7}, l_i^{1}m_i^{1}\}$ if $T(x_i)=\textit{true}$; otherwise, if $T(x_i)=\textit{false}$ we add $\{p_i^{2}p_i^{3},t_i^{1}g_i^{1},l_i^{2}m_i^{2},p_i^{6}p_i^{7},t_i^{2}g_i^{2},m_i^{1}r_i^{1}\}$,

\smallskip

\item[$\bullet$] If $H(x_i)$ complies with case 3, we add $\{l_i^{1}m_i^{1},l_i^{2}m_i^{2},l_i^{3}m_i^{3},f_ih_i\}$ if $T(x_i)=\textit{true}$; otherwise, if $T(x_i)=\textit{false}$ we add $\bigcup_{1\leq j\leq 3}\{r_i^{j}m_i^{j},t_i^{j}g_i^{j}\}$.
\end{itemize}

Conversely, suppose $S\subseteq U$ is a maximal edge matching of $H$. Because of maximality, for each clause gadget $H(x_j)$ there exists at least one crossing edge in $S$ incident to a vertex of $H(x_j)$. This means that there is an assignment  $T$ which satisfies all clauses of $I$. We now show that $T$ is a valid assignment:

\begin{itemize}
\item[$\bullet$] If $H(x_i)$ complies with case 1, by maximality either $l_i^{1}m_i$ or $l_i^{2}m_i$ (not both) is in $S$, hence we set $T(x_i)=\textit{true}$ (resp., $T(x_i)=\textit{false}$) if $l_i^{1}m_i^{1}\in S$ (resp., $l_i^{2}m_i$),

\smallskip 

\item[$\bullet$] If $H(x_i)$ complies with case 2, by maximality either $S_1=\{p_i^{1}p_i^{3}, m_i^{2}r_i^{2}, p_i^{5}p_i^{7}, l_i^{1}m_i^{1}\}$ or $S_2=\{p_i^{2}p_i^{3},l_i^{2}m_i^{2},p_i^{6}p_i^{7},m_i^{1}r_i^{1}\}$ (not both) is in $S$, so we set $T(x_i)=\textit{true}$ (resp., $T(x_i)=\textit{false}$)if $S_1\in S$ (resp., $S_2\in S$),

\smallskip

\item[$\bullet$] If $H(x_i)$ complies with case 3, by maximality either $S_1=\bigcup_{1\leq j\leq 3}\{l_i^{j}m_i^{j}\}$ or $S_2=\bigcup_{1\leq j\leq 3}\{r_i^{j}m_i^{j}\}$ (not both) is in $S$, hence we set $T(x_i)=\textit{true}$ (resp., $T(x_i)=\textit{false}$) if $S_1\in S$ (resp., $S_2\in S$).
\end{itemize}
\end{proof}

 \section{Exponential Time Hypothesis and Exact Algorithms}
 \label{sec-ETH}

Impagliazzo, Paturi and Zane initiated the analysis of computationally hard problems under the hypothesis that there are (basically) no $2^{o(n)}$ (i.e., no sub-exponential) algorithms for solving
3-SAT, where $n$ is the number of variables; the number of clauses is somehow subsumed into this expression, as this number can be assumed to be sub-exponential in $n$ (after applying the famous sparsification procedure); cf.~\cite{ImpPatZan2001}.
This hypothesis\footnote{Actually, a slightly different variant of it.} is also known as \emph{Exponential Time Hypothesis}, or ETH for short.
This type of study was furthered in many directions; for us, the most relevant seems to be recent work of Jonsson \emph{et al.}~\cite{JonLNZ2017} who showed, among other things, that ETH extends to \textsc{Not-All-Equal-3SAT}, or NAE-3SAT, for short.
So, if ETH is true, then there exists no sub-exponential algorithm for  NAE-3SAT, either. We also refer to the survey article~\cite{LokMarSau2011b}.

Re-analyzing our $\np$-hardness (or $\wone$-hardness) constructions, we can now immediately deduce the following results. This is interesting to see, as for all our graph extension problems, we can easily find $2^{O(n+m)}$ algorithms on instances with $n$ vertices and $m$ edges. Hence, the following results show that, up to constant-factor improvements of the basis of the exponential term (which is of course crucial for some practical considerations), these trivial algorithms are optimal under ETH.

\begin{corollary} Assuming ETH, there is no $2^{o(n+m)}$-algorithm for solving $n$-vertex, $m$-edge instances of 
\textsc{Ext EC} in bipartite graphs of maximum degree 3, even if the required edges build an edge matching. (From Theorem~\ref{Bip_Ext_EC})
\end{corollary}

\begin{corollary} Assuming ETH, there is no $2^{o(n+m)}$-algorithm for solving $n$-vertex, $m$-edge instances of 
\textsc{Ext EM}  in bipartite graphs of maximum degree~3. (From Theorem~\ref{degree3_Ext_EM})
\end{corollary}

\begin{corollary} Assuming ETH, there is no $2^{o(n+m)}$-algorithm for solving $n$-vertex, $m$-edge instances of 
\textsc{Ext DS}  in bipartite graphs of maximum degree~3. (From Theorem~\ref{Bip_Ext_DS})
\end{corollary}

\begin{corollary} Assuming ETH, there is no $2^{o(n+m)}$-algorithm for solving $n$-vertex, $m$-edge instances of 
\textsc{Ext EDS}  in bipartite graphs of maximum degree~3. (From Theorem~\ref{Bip_Ext_EDS} or from Theorem~\ref{theoParameterized:Ext EDS})
\end{corollary}

We are now turning to planar variants of our problems. Instead of going through them one-by-one, we only provide one corollary summarizing our observations. The main reason for being able to state \LV{this c}\SV{C}orollary\SV{~\ref{cor:ETH-planar}} that easily is that all our reductions presented in Section~\ref{sec-planar} start out from the \textsc{4P3C3SAT} problem introduced earlier by J. Kratochv\'il~\cite{Kra94}, and all these reductions do not blow up the instances too much.
Now, Kratochv\'il's $\np$-hardness proof for \textsc{4P3C3SAT} is based on Lichtenstein's construction~\cite{Lichtenstein82} that in turn proved that \textsc{3SAT}, restricted to planar instances, is still $\np$-hard. Now, while Lichtenstein's construction contains a quadratic blow-up, Kratochv\'il's refinement is linear in size only, so that we can use it first to state an ETH result for \textsc{4P3C3SAT}  and then make use of it for the problems we are interested in for our paper.  This reasoning shows Proposition~\ref{prop:SAT} and finally 
Corollary~\ref{cor:ETH-planar}.


\textsc{Ext BP} delivers a nice example that it is not always that easy to take the textbook construction from Garey and Johnsson for
$\np$-hardness to immediately get hardness results under ETH that match existing algorithms. The classical reduction to show $\np$-hardness of \textsc{3-Partition} is from \textsc{4-Partition}, which again reduces from  \textsc{3D-Matching}, which reduces from \textsc{3-SAT}. Instances of  \textsc{3D-Matching} have $2mn+3m+2m^2n(n-1)$ many triples according to \cite{GJ79}. Each of the three involved sets $W,X,Y$ has $O(nm)$ many elements. Then, the corresponding  \textsc{4-Partition} instance has $O(|M|)$ many elements. The corresponding \textsc{3-Partition} instance has $O(x^2)$ many elements if the  \textsc{4-Partition} instance has $x$ elements.
Putting things together, a straight analysis would only give a relatively weak non-existence of an $O(2^{o(\sqrt[8]{\ell}}))$ algorithm for \textsc{Ext BP} with $\ell$ items under ETH.

We have to apply another route to obtain the desired result.
According to~\cite[Corollary 3.2]{JanLanLan2016}, there is no algorithm deciding \textsc{4-Partition} with $n$ items in time $2^{o(n)} \times|I|^{O(1)}$, unless  ETH fails.
As we can adapt our reduction from Theorem~\ref{thm:Ext BP np-hard} to reduce from this problem, we can formulate Corollary~\ref{cor:eth-ExtBP}.

\paragraph*{An Exact Algorithm for \textsc{Ext BP}}

The correctness of the DP algorithm claimed in Theorem~\ref{thm:ExtBPDP} relies on the following characterization.

\begin{lemma}
A partition solution $\pi$ is  minimal if and only if there is a constant $0<\delta< 1/2$
such that for the two bins (sets) $X_1,X_2\in\pi$ of smallest weight, 
with $w(X_1)\leq w(X_2)$, 
 $w(X_1)>\delta$ and $w(X_2)\geq 1-\delta$.
\end{lemma}

\begin{proof} The partition 
$\pi$ is not minimal if and only if there are two bins, i.e., sets in the partition, that can be merged without violating the size constraint of one. For a set of bins, any two of them, say, $A$ and $B$, can be merged if and only if there is some  $\delta$, $0<\delta< 1/2$, such that the weight of $A$ is at most $\delta$ and the weight of $B$ is at most $1-\delta$. To avoid this situation,
we only have to consider merging the two bins of smallest weight.
\end{proof}

The bound of  the previous lemma is only helpful, because we do not have to test all (uncountably many)  values of $\delta$, but only those that can be realized within our given instance  $(X,w,\pi_U)$ of \textsc{Ext BP}.

\begin{obs}
It is sufficient to test all $\delta$, $0<\delta< 1/2$, such that there is a subset of items of weight $\delta$ or $1-\delta$ in the given instance. 
\end{obs}

Hence, for all $\delta$, $0<\delta< 1/2$ that can be realized by the instance $(X,w,\pi_U)$ of \textsc{Ext BP}, we consider the following problem; notice
that only $2^n$ problems are generated this way.

\begin{center}
\fbox{\begin{minipage}{.95\textwidth}
\noindent{\textsc{Ext $\delta$-BP}}\\\nopagebreak
{\bf Input:} A set 
$X=\{x_1,\dots,x_n\}$ of items, a weight function $w$ associating rational numbers $w(x_i)\in(0,1)$ to items, a partition  $\pi_U$ of $X$.\\\nopagebreak
{\bf Question:} Does there exist a minimal feasible partition $\pi_U'$ with $\pi_U\leq_X\pi_U'$ such that, except for $Y_{\text{min}}\in\pi_U'$, all sets $Y\in\pi_U'$ satisfy $1-\delta\leq w(Y)\leq 1$, while $Y_{\text{min}}$ satisfies $\delta<w(Y_{\text{min}})\leq 1$?
\end{minipage}}
\end{center}

The overall idea is similar to DP solutions solving well-known vertex ordering problems; see \cite{BodFKKT2012}. One of the main ideas is to study all orderings in which items could have been put into bins. 
This would be too expensive for our purposes. Therefore, we use dynamic programming in order to maintain the possibility to consider all but one of the bins to be closed. 
It is sufficient to store the set of items that has already been packed into bins that are considered closed, plus one that is currently open. Also, we will maintain the invariant that all
closed bins contain at least $(1-\delta)$ of weight, so that potentially only the last, still open, bin might violate the constraint of the previous lemma and contain at most $\delta$ of weight. 

To this purpose, let $T_\delta[Y,L]=1$ if, maintaining the restriction imposed by $\pi_U$, there is a possibility to pack $Y\subseteq X$ into bins such that all but the last bin described by $L$  contain  a weight of at least $1-\delta$, while $w(L)\in [0,1]$. If such a special packing does not exist, we set $T_\delta[Y,L]=0$. In our setting, this means that $T_\delta[Y,L]=1$ certifies the existence of a partition $\pi$ of $Y$ such that,  for all  sets $B\in \pi$ different from $L$, $w(B)\geq 1-\delta$. Also, whenever two elements $x,y\in Y$ are put into two different sets by $\pi_U$, this is the case for $\pi$, as well.

Clearly, there are $3^n$ many possibilities for choosing $Y\subseteq X$ and $L$ as a subset of $Y$. This determines the space requirements of our algorithm.

In our dynamic programmming solution, we compute the entries of the table $T_\delta[\cdot,\cdot]$ for all $Y\subseteq X$ with increasing size. For the smallest size, $|Y|=0$, we initialize $T_\delta[\emptyset,\emptyset]=1$. 
By induction, assume that $T_{\delta}[Y,L]$ has been correctly filled for all $Y$ such that $|Y|\leq c$ and for all $L\subseteq Y$.
Now, consider some $Y$ with $c+1$ many elements.
Initialize $T_{\delta}[Y,L]=0$ for all $L\subseteq Y$.
Walk through all $y\in Y$, basically looking for the last element that was put in, in order to produce a solution for $Y$.
If  $T_{\delta}[Y\setminus\{y\},L]=1$ and $w(L\cup\{y\})\leq  1$ for any $L\subseteq Y\setminus\{y\}$, we can update  $T_{\delta}[Y,L\cup\{y\}]=1$, because we can add $y$ to the last,
not yet closed bin, unless the open bin would then violate a constraint imposed by $\pi_U$, i.e., unless there is some $x\in L$ and some $A,B\in \pi_U$ with $x\in A$, $y\in B$, $A\neq B$. 
If  $T_{\delta}[Y\setminus\{y\},L]=1$ and $w(L)\geq 1-\delta$ for any $L\subseteq Y\setminus\{y\}$, we can update  $T_{\delta}[Y,y]=1$. 
The last case corresponds to closing a previously open bin.
This way, we can correctly fill all table entries of $T_\delta[Y,L]$ for $|Y|=c+1$ and $L\subseteq Y$. By induction, the whole table $T_\delta[\cdot,\cdot]$ can be correctly filled.

After having filled the table for $T_{\delta}[X,L]$ for all $L\subseteq X$ as described by the recursion, we walk once more through all these table entries and look for some $L\subseteq X$ such that $w(L)> \delta$ and  $T_{\delta}[X,L]=1$. If such an $L$ can be found, then $X$ can be partitioned such that $\pi_U$ is obeyed and such that all but one bin contain weight at least $1-\delta$, while the exceptional bin contains weight more than $\delta$. Hence, this partition is a minimal extension of $\pi_U$ as required. If no such $L$ can be found, then there is no minimal extension of $\pi_U$ that satisfies the weight restriction imposed by $\delta$.

The overall running time for solving \textsc{Ext $\delta$-BP}
is clearly dominated by filling the DP table $T_\delta[\cdot,\cdot]$.
This can be estimated by
$$\sum_{c=0}^{n}\binom{n}{c}c2^c\leq n3^n\,.
$$
As $2^n$ many such tables have to be computed, we arrive at the claimed running time for solving \textsc{Ext BP}.

\section{Exact algorithms for the graph problems}

In the following, we will first explain quite to some detail both the intuition and a rather concrete implementation of an algorithm for \textsc{Ext VC}, because --- although not being in the focus of the studies in this paper --- this is the simplest of all considered problems from the 
point of view of developing algorithms. In the next subsection, we then show how to adapt these ideas for developing algorithms for the other graph problems. Finally, we argue for exact exponential-time algorithms. All these results are claimed and addressed in the main part of the paper. 
\subsection{A treewidth algorithm for Ext VC}


We assume that the graph $G$ that we are considering has treewidth at most $t$ and that we are given a nice tree decomposition of $G$.

We have to specify what happens at the leaf nodes, and then (in the recursion)
what to do in introduce, forget and join nodes, assuming we are given a nice tree decomposition.

We first define the concept of states appropriate for \textsc{VC Ext}.

Of course, there should be a distinction between a vertex being or not being in the vertex cover.
However, in contrast to the classical decision version, this is not sufficient, because once a vertex is put into the cover (or is assumed to be there),
we also must store the fact if (or if not) a vertex in the cover has a private edge, as this proved minimality of the solution.

\smallskip
\noindent
\emph{initialization}: As we are assuming a nice tree decomposition, initial nodes contain one vertex only.
This vertex is either (1) not in the cover or (2) it is in the cover but does not have (yet) a private edge.
In the case that the specific vertex that we add a vertex that is bound to belong to the (given) pre-solution, 
then only case (2) would apply.
This describes all situations corresponding to (potential) solutions of the VC extension problem.

\smallskip
\noindent
\emph{introduce nodes}: There is a (potential) solution to the extension problem if one of the following cases apply.
\begin{enumerate}
\item The new vertex is determined not to be in the vertex cover.
Of course, this also means that it does not belong to the specified pre-solution.
This is only feasible if all neighbors of this vertex (in the current bag) are determined to belong to the vertex cover.
Also, if they might not have seen a private edge so far, now they do.
\item The new vertex is determined to belong to the vertex cover.
This new vertex is having a private edge if and only if one of its neighbor in the current bag is determined not to belong to the vertex cover.
\end{enumerate}

\smallskip
\noindent
\emph{forget nodes}: One of the following two cases might apply.

\begin{enumerate}
\item The  vertex that we like to forget is determined not to be in the vertex cover and this is also the case for one of its neighbor, or it is put into the cover but does not have a private edge yet.
This is surely not leading to any feasible (extension) solution, so that such situations must be excluded.
\item If the vertex that we like to forget is put into the vertex cover and has a private edge, or if it is determined not to be in the vertex cover but all neighbors are in the cover, then we can safely remove it.
\end{enumerate}

Let us present our algorithm more formally in the following.
To each bag with $k$ vertices, we associate (conceptually) a table with $3^k$ many rows, indexed by tuples from $\{0,1,2\}^k$, which either carries the value 1
(meaning that the association of 0, 1, 2 to the vertices of the bag may still be extended to some (minimal) extension of the given pre-solution)
or the value 0 (no extension is possible).
Alternatively, we can view this table as a subset of $\{0,1,2\}^k$, containing the rows of the table with value one. 
Here, the value 0 means that the associated vertex does not belong to the cover, 1 means that this vertex does belong to the cover but does not have a private edge (so far) and 2
means that the vertex belongs to the cover and possesses a private edge.

Let $U$ be the given pre-solution and $G=(V,E)$ be the given graph.
Also, we are given a nice tree decomposition of $G$ of width at most $t\leq |V|$.
Finally, we assume a strict  linear order $<$ on $V$, which naturally transfers to subsets of vertices, 
so that we can interpret tuples  from $\{0,1,2\}^k$ in a unique way as assigning 0, 1, 2 to vertices in a bag consisting of $k$ vertices.

If $v$ is a vertex that is belonging to a leaf node of the tree decomposition, then we associate the set $\{0,1\}$ to this bag if $v\notin U$,
and $\{1\}$ if $v\in U$.

Assume that $B'$ is a bag in the decomposition such that $B:=B'\cup\{v\}$ is the parent node in the tree decomposition.
In other words, $B$ is an \emph{introduce node}.
Assume that the set (table) $T'$ is associated to $B'$. We are going to describe the table $T$ associated to $B$ in the following.
Assume, without loss of generality, that $v$ is bigger than any vertex from $B'$ in the ordering on $V$.
Let $1\leq i_1<i_2<\dots<i_\ell\leq|B'|$ be the indices of the  vertices of $N(v)\cap B$. 
Also, define $\widehat{0}=0$ and $\widehat{1}=\widehat{2}=2$. 
For all tuples $(x_1,\dots,x_{|B'|})$ in $T'$, we put $(\widehat{x_1},\dots,\widehat{x_{|B'|}},0)$ into $T$ iff $x_{i_1},\cdots,x_{i_\ell}\in\{1,2\}$ and $v\notin U$. Namely, all neighbors of $v$ have now found a private edge. 
Moreover, for all tuples $(x_1,\dots,x_{|B'|})$ in $T'$, we put $({x_1},\dots,{x_{|B'|}},1)$ into $T$ iff $x_{i_1},\cdots,x_{i_\ell}\in\{1,2\}$.
Finally, for all tuples $(x_1,\dots,x_{|B'|})$ in $T'$, we put $({x_1},\dots,{x_{|B'|}},2)$ into $T$ iff, for some $j\in\{1,\dots,\ell\}$, $x_{i_j}=0$.

Assume that $B'$ is a bag in the decomposition such that $B:=B'\setminus\{v\}$ is the parent node in the tree decomposition.
In other words, $B$ is a forget node and $v$ is the vertex that has to be forgotten.
Assume, w.l.o.g., that $v$ is last in the ordering of the vertices of $B'$.
Assume that the set (table) $T'$ is associated to $B'$. We are going to describe the table $T$ associated to $B$ in the following.
If $(x_1,\dots,x_{|B|},2)\in T'$, then $(x_1,\dots,x_{|B|})\in T$.

We can always do a sort of clean-up operation after setting up one new table:
Whenever $(x_1,\dots,x_{|B|}),(x_1',\dots,x_{|B|}')\in T$ such that always either $x_i=x_i'$ or $x_i=2$ and $x_i'=1$, then we can delete $(x_1',\dots,x_{|B|}')$ from $T$.
Namely, there seems to have been a way to determine a private edge for the vertex corresponding to $x_i$ in $B$ outside of $B$ in some solution that, from the point of view of $B$, looks identical to another solution where $x_i$ does not have a private edge so far.
As finally the (positive) decision about extensibility will be based on finding $2$ and $0$ only in the settings, we can ignore the setting that we propose to delete.
In conclusion, $T$ will at no point contain more than $2^{|B|}$ many elements.

Assuming clean-up, we can also add the following rule in the forget-node case: 
If $(x_1,\dots,x_{|B|},0)\in T'$, then $(x_1,\dots,x_{|B|})\in T$.

This clean-up also simplifies the join operation considerably. Now, we first turn the two sets into tables with binary (0/1) $k$-tuples that map to the unique (due to clean-up) tuple from $\{0,1,2\}^k$. Then, one performs the following:
If a binary $k$-tuple $(x_1,\dots,x_{|B|})$ has corresponding ternary tuples $(y_1,\dots,y_{|B|})$, $(z_1,\dots,z_{|B|})$ in the two child bags, we will put $$(\max\{y_1,z_1\},\dots,\max\{y_{|B|},z_{|B|}\})$$ into the parent's bag. Conversely, if a binary tuple is only showing up in the table of one child node, then this must mean that some 0-entry of this tuple would correspond to a situation that is unfeasible in the other child node, because then some of the edges connected to the corresponding vertex in the subgraph that is treated by that child bag is not covered. 
Hence, in such a situation no tuple is added to the parent table.


This reasoning proves that \textsc{Ext VC} can be solved in time  $\Oh^*(2^{tw})$, which matches the time known for the classical corresponding decision problem.

\subsection{Treewidth-based algorithms for other graph problems}

This running time of $\Oh^*(2^{tw})$ can be also matched by the extension variant of domination.
The intuitive reason behind is that we are not interested in the value of a solution (not aiming at solutions of small cardinality) but only at the existence of such solutions.
This also allows us to avoid using subset folding and other similar techniques; it is sufficient to update the information about unique concretizations of the binary vectors associated to vertices in each bag.

To be more concrete, again (for dominating set) we basically have to store a collection of binary vectors (telling if a vertex belongs to a dominating set or not) and then a pointer to the more refined version saying if a vertex in a dominating set already has a private neighbor or not and also if a vertex not in the dominating set is neighbor of a vertex in a dominating set or not.
In both cases, we have a kind of domination relation between the states (as with vertex cover).

Alternatively, we may refer to the algorithm given for {\sc Upper Domination} parameterized by treewidth in \cite{BazBCFJKLLMP2018}.
\medskip

For edge domination as well as for edge cover and for maximum matching, one has to adapt our ideas a bit more. \LV{This has been made explicit for a related decision problem in~\cite{HanOno2015}.} 
To sum up these ideas:
\begin{itemize}
\item Given a nice tree decomposition, we can express for each vertex if it is incident to an edge that is belonging to a solution or not. If it is incident to an edge belonging to a solution, one can distinguish between the case that this is already satisfied or whether this will be satisfied only in the future.
\item Minimality and the related notion of privateness is specific to each problem. We will detail this in the following.
\end{itemize}

We are now discussing the three problem variants separately, sketching some further ideas.

\begin{itemize}
\item For matching, there should be no possibility to insert another edge in the current solution. This means that we  should keep track of if a vertex that is not incident to an edge in the currently considered matching is (not) neighbor of a vertex that is incident to some edge in this matching. Conceptually, this means that we have to consider the following states in our dynamic program:\\
$0$: a vertex that is not incident to an edge of the matching but is known to be neighbor to some vertex incident to an edge of the matching.\\
$\hat 0$:  a vertex that is not incident to an edge of the matching and is not yet proven to be a neighbor to any vertex incident to an edge of the matching.\\
$1$: a vertex that has already been paired up with another neighboring vertex, building together an edge of the matching.\\
$\hat 1$: a vertex that is assumed to be incident to an edge in the matching, but so far has not found its partner.

In the initialization at the leaf nodes, we only use $\hat 0$ or $\hat 1$. Clearly, if $x$ is incident to an edge contained in the pre-solution, then the leaf nodes corresponding to $x$ will only include the case $\hat 1$ as the labeling for $x$.

If we introduce a vertex in an introduce node, then the status $\hat 0$ might be updated towards $0$, and similarly the status $\hat 1$ might be updated to $1$, assuming that the conditions are satisfied by setting the new vertex to $\hat 1$ or to $1$ (in the second case, when the status of a neighbor turns from $\hat 1$ to $1$, clearly this pairing-up means that the status of the newly introduced vertex is $1$; also notice that in that case, there is exactly one neighbor of the newly introduced vertex that had status $\hat 1$ and no neighbor with status $1$).
In a  forget node, if we associate status $0$ or $1$ to the vertex $x$ that is going to be forgotten, this is fine, but table entries that associate $\hat 0$ or $\hat 1$ to $x$ will be ignored in (not carried over to) the table associated to the forget node.

In any case, we have to make sure that the DP takes care of the natural consistency requirements of the status information, which is as follows.

\begin{itemize}
\item There should never be two neighboring vertices labeled $0$ or $\hat 0$.
\item Finally, we have to look into the table of the root bag if there is any table entry consisting of $0$ and $1$ only (i.e., no $\hat 0$ or $\hat 1$ are around). There is an extension to the given pre-solution if and only if this is the case.
\end{itemize}

\item For edge cover, the states of the vertices can keep track of the information if they serve as a private vertex to some edge. In fact, this is only possible if there is exactly one edge in the cover that is incident to this vertex. Hence, the natural states of a vertex are one of the following:\\
$0$: not being incident to any edge in the cover but should be incident to an edge later on;\\ 
$1$: being incident to exactly one edge in the cover (and hence being private to some edge);\\
$\hat 1$: being incident to at least one edge in the cover (but not being private to any edge).

Initially, at leaf nodes, vertices have status $0$.
Whenever a vertex is inserted (introduce node),
the status of neighbors changes from $0$ to either $1$ or $\hat 1$ (both is possible).
In introduce nodes, also the pre-solution is taken care of, as then $0$ is not an admissible state. When looking into forget nodes, table rows that contain $0$ as a state of the vertex that is forgotten will not contribute to the new table.
With join nodes, table rows combine according to rules like: If there is status $0$ in one table row but $1$ in the other, then this combines to $1$ in the parent table row; etc. 
\item For edge domination, the situation is the most complicated. We have the following states.\\
$0$: No edges from the EDS are incident to such vertices.
This means that all edges incident to a vertex colored $0$ can serve as private edges. Also, in order to form a valid edge domination set at all, it is clear that the set of all $0$-colored vertices forms an independent set.\\
$1$: At least one neighbor is $0$-colored, and at most one edge from the EDS is incident to a $1$-colored vertex. Hence, the edge from the EDS that is incident with a $1$-colored vertex has a private neighboring edge emanating from the $1$-colored vertex. In the course of the DP algorithm, $1$ signals that indeed one neighbor is $0$-colored, while the variant\\ $\hat 1$ has the same semantics as $1$ has, except that one still expects the $0$-colored neighbor to appear in the future. Hence, in forget nodes, table entries that contain $\hat 1$ for the vertex to be forgotten have to be omitted.\\
$2$: No neighbor is (ever) $0$-colored, but there is an edge of the EDS incident to a $2$-colored vertex.\\
In the terminology introduced so far, all neighbors of a vertex colored $0$ will be colored~$1$ (there will be another variant introduced below that is also fine here).

In the course of the algorithm, the actual EDS should be determined based on the colorings described so far.
This leads to three further color variants:
$1^\circ$, $\hat 1^\circ$ and $2^\circ$. The ring symbol should indicate that such vertices are already paired-up. Here, certain rules have to be obeyed for pairing-up neighboring vertices $x$ and $y$.
\begin{itemize}
\item If $x,y$ both have the status $1$, then they can both obtain the status $1^\circ$.\item If $x,y$ both have the status $\hat 1$, then they can both obtain the status $\hat 1^\circ$.\item If $x$ has the status $\hat 1$ and $y$ has the status $1$, then $x$ can  obtain the status $\hat 1^\circ$ and $y$ can obtain the status $1^\circ$.
\item If $x$ has status $1$ and $y$ has status $2$ or $2^\circ$, then $x$ can get the status $1^\circ$ and $y$ the status $2^\circ$.\item If $x$ has status $\hat 1$ and $y$ has status $2$ or $2^\circ$, then $x$ can get the status $\hat 1^\circ$ and $y$ the status $2^\circ$.
\end{itemize}
In other words, vertices with color $1$ (or $\hat 1$) can only find one partner (defining the edge in the EDS), while vertices with color $2$ can find multiple partners, but all of them carry color $1$ or $\hat 1$.
\end{itemize}
Initially, at the leaf nodes, only the colors $0$, $\hat 1$, and $2$ are available. Moreover, if an  edge of the pre-solution is incident to the vertex introduced in the leaf node, then $0$ is not a possible color. 
Moreover, if $x$ is incident to at least two edges from the given pre-solution, then $x$ will have color $2$.

In introduce nodes, we could again first distinguish the cases of colors $0$, $\hat 1$, and $2$, again keeping in mind  that if an  edge of the pre-solution is incident to the vertex introduced, then $0$ is not a possible color.
Moreover, if $x$ is incident to at least two edges from the given pre-solution, then $x$ will have color $2$.
Let us look more carefully at the situation when introducing vertex $x$.

Assume we are looking at a particular table row $r$ (given by a tuple of colors uniquely associated to the vertices in the child bag $B$), with $c_1,\dots,c_j$ being in particular the colors of the neigbors of $x$ among the vertices in $B$. 
\begin{itemize}
\item If one of the $c_i$ is $0$, then we cannot color $x$ with $0$, nor with $2$ (or with $2^\circ$). Rather, $x$ has to be colored with $1$.
\item If one of the $c_i$ is $\hat 1$ and we color $x$ with $0$, then $c_i=\hat 1$ will be replaced by $1$ in the table row(s) of the parent bag.
\end{itemize}

In forget nodes, we have to pair-up yet unpaired vertices that are forgotten.
More specifically, assume we are about to forget vertex $x$ and let $c_1,\dots,c_j$ be the colors of the neigbors of $x$.
Let $r$ be a specific row in the child bag. For simplicity, let $r-x$ mean the row obtained from $r$ by omitting the color entry corresponding to $x$.
\begin{itemize}
\item If the color of $x$ in $r$ is $0$ or $1^\circ$, then $r-x$ is put as a row in the table of the parent bag.
\item Likewise, if $x$ is colored $2^\circ$. In addition, we change the colors $c_i=1$ to $c_i=1^\circ$ or from $c_i=\hat 1$ to $c_i=\hat 1^\circ$, either testing all such possibilities (and hence introducing quite a number of new rows in the table of the parent bag) or this can be also enforced, namely, if the corresponding edge incident to $x$ was belonging to the given pre-solution.
\item Moreover, if $x$ is colored $1$, then pick any neighbor of $x$ with color $c_i=1$ (or $c_i=2$ or $c_i=2^\circ$, resp.) and turn this chosen color to $c_i=1^\circ$ (or $c_i=2^\circ$, resp.), implicitly also turning the color of $x$ to $1^\circ$ before removing $x$. Hence, $x$ is paired-up if necessary.
Notice that this pick is deterministic if $x$ is incident to some edge from the given pre-solution.
\end{itemize}

In join nodes, consistency checks are necessary.
Clearly, identical rows can be copied as long as they do not contain paired-up vertices with a variant of color $1$, but if, for instance, vertex $x$ is colored $1^\circ$ in both bags, then this row cannot be copied, because this means that in $x$ (with color $1$) has already found its partner in both sub-parts of the graph (represented by the two children bags), which contradicts our semantics of the color $1$.  Conversely, there are also compatible vertex colors that are not identical.
For instance, if vertex $x$ is colored $\hat 1$ in one bag but $1$ in the other, then in the parent bag, the situation is reflected by the color $1$. Likewise, pairing-up of a vertex in one child bag is fine and will combine with the non-paired variant of the same color.
Finally, color $2^\circ$ is compatible both with $2$ and with $2^\circ$. 

Finally, a given pre-solution is extendible if and only if we can find in the root bag a table row where all vertices not colored $0$ are already paired-up or (if not) the remaining vertices can be paired-up by  computing a matching to match all vertices
colored $1$ with themselves or with vertices colored $2$ or $2^\circ$, so that all vertices previously colored $1$ or $2$ have found their partner.

\subsection{Exact exponential-time algorithms}
Let us finally sketch how the ideas for treewidth-based algorithms can be used to construct exact algorithms for the graph problems that we study.
First, for the vertex problems (like vertex cover or dominating set), it is trivial to come up with an algorithm with running time $\Oh^*(2^n)$: just cycle through all possible subsets and test for feasibility of the solution, i.e., is this a minimal extension of the given pre-solution?
For the edge problems, we can mis-use the earlier derived treewidth-based algorithms,
simply observing that a graph with $n$ vertices has pathwidth bounded by $n$. 

Clearly, it would be an idea to develop faster exact algorithms than the ones that follow in a rather trivial way as sketched in the previous paragraph. We leave this for future research.

\end{document}